\documentclass{article}

\usepackage{amsmath}
\usepackage{amsthm}
\usepackage{amssymb}
\usepackage{rotating}
\usepackage{array}
\usepackage{enumitem}
\usepackage{mathtools}
\usepackage{setspace}
\usepackage{fullpage}
\usepackage{hyperref}
\usepackage{algorithm}
\usepackage{booktabs}
\usepackage{xspace}
\usepackage{authblk}
\usepackage{todonotes}
\usepackage[T1]{fontenc}

\usepackage[noend]{algpseudocode}
\algnewcommand\Input{\item[\textbf{Input:}]}
\algnewcommand\Output{\item[\textbf{Output:}]}

\newtheorem{theorem}{Theorem}
\newtheorem{lemma}{Lemma}
\newtheorem{definition}{Definition}

\newcommand {\ceil}[1] {\left\lceil #1 \right\rceil}
\DeclareMathOperator* {\argmax} {arg\,max}

\newcommand{\etal}{\textit{et al.}}

\newcommand{\R}{\mathbb{R}}
\newcommand{\N}{\mathbb{N}}

\newcommand{\ER}{\bar{E}}

\newcommand{\G}{\mathbb{G}}
\newcommand{\ps}{\Pi}
\newcommand{\rnk}{\#}

\newcommand{\sd}{\sigma}
\newcommand{\sddegr}{\sd_{degr}}
\newcommand{\sdclos}{\sd_{clos}}
\newcommand{\sdbetw}{\sd_{betw}}
\newcommand{\sdeig}{\sd_{eig}}
\newcommand{\sdrumor}{\sd_{rumor}}
\newcommand{\sdrwalk}{\sd_{rwalk}}
\newcommand{\sdmcarlo}{\sd_{mcarlo}}

\newcommand{\vs}{v^\dagger}
\newcommand{\FA}{\widehat{A}}
\newcommand{\FR}{\widehat{R}}
\newcommand{\FS}{\widehat{S}}
\newcommand{\Add}{A^*}
\newcommand{\Rem}{R^*}

\newcommand{\Subo}{\nabla}

\newcommand{\II}{I}
\newcommand{\thr}{\omega}
\newcommand{\goal}{goal}

\newcommand{\BAn}{Barab\'{a}si-Albert\xspace}
\newcommand{\ERn}{Erd\H{o}s-R\'{e}nyi\xspace}

\title{Social diffusion sources can escape detection}

\author[a]{Marcin Waniek}
\author[b]{Manuel Cebrian}
\author[c]{Petter Holme}
\author[a]{Talal Rahwan}

\affil[a]{New York University Abu Dhabi, Abu Dhabi, UAE}
\affil[b]{Max Planck Institute for Human Development, Berlin, Germany}
\affil[c]{Tokyo Institute of Technology, Tokyo, Japan}

\date{}

\begin{document}

\maketitle

\begin{abstract}
Influencing (and being influenced by) others through social networks is fundamental to all human societies. Whether this happens through the diffusion of rumors, opinions, or viruses, identifying the diffusion source (i.e., the person that initiated it) is a problem that has attracted much research interest. Nevertheless, existing literature has ignored the possibility that the source might strategically modify the network structure (by rewiring links or introducing fake nodes) to escape detection. Here, without restricting our analysis to any particular diffusion scenario, we close this gap by evaluating two mechanisms that hide the source---one stemming from the source's actions, the other from the network structure itself. This reveals that sources can easily escape detection, and that removing links is far more effective than introducing fake nodes. Thus, efforts should focus on exposing concealed ties rather than planted entities; such exposure would drastically improve our chances of detecting the diffusion source.
\end{abstract}

\section{Introduction}

As humans, we are perpetually involved in, and affected by, things spreading in networks---from infections~\cite{block2020social,chiu2020state} to ideas, from financial distress to fake news~\cite{barrat2008dynamical,lehmann2018complex}.
Furthermore, we live in an increasingly networked world---the networks of our society are becoming denser, meaning that spreading phenomena happen with accelerating speed~\cite{barabasi2016network}. Occasionally, we do not know who started the spreading. However, recent research has shown that we can often infer the source with great accuracy~\cite{lokhov2013inferring,jain2016fast,shah2020finding}.
If the spreading has an illicit intent or negative consequences for the source---like bioterrorism, disinformation, or whistleblowing in an authoritarian society---the source would want to hide from such source detection algorithms.
In this article, we study the conditions under which such hiding can be successful.
To keep the results as widely applicable as possible, and to conform to the praxis of Network Science~\cite{lokhov2013inferring,jain2016fast,shah2020finding,de2013anatomy,was2020manipulability,waniek2018hiding}, we avoid making domain-specific assumptions about the nature of the social diffusion.

The main challenge to understand the possible obfuscation of the diffusion source is that it is a problem with two components. First, networks have an innate ability to hide the source; although most research in the literature has focused on designing source detection algorithms, the efficiency of these algorithms strongly depends on the network structure.
Second, by changing its local network surrounding, the source can hide its identity. So far, no theory has been able to separate these two factors, and understand how they are affected by the dynamics of the spreading and the timing of the source detection.
This article builds such a theory from the systematic simulations,
with and without active obfuscation of the source. 
For the sake of parsimony, we base our study on simple contagion~\cite{lehmann2018complex,goffman1964generalization},
while varying the 
network structure, source detection, and timing of these.

We begin our analysis by investigating the theoretical limits of the problem of hiding the source of diffusion. We mathematically prove that identifying an optimal solution to this problem is practically impossible. Based on this, we turn our attention to feasible---albeit not optimal---ways in which adversaries may hide the source. The first way involves executing several heuristics that either introduce new nodes or rewire existing links. The second way relies on the network structure itself, without any intervention from the source. We evaluate both ways by running simulations in real-life networks, as well as synthetic networks with varying structure and density. This evaluation is done via exact methods when the network is of moderate size and approximation when the network is massive. Although most of our analysis focuses on a simple contagion model, we also perform simulations with alternative diffusion models, including variants of complex contagion. We also study how the source's actions affect other nodes part of the same cascade and how the source detection algorithms are influenced, given imperfect knowledge about the network structure. We conduct several sensitivity analyses, varying how the source is selected, the parameters for generating synthetic networks, and the approximations by which we evaluate the heuristics. Finally, we validate our findings using a large-scale dataset of new hashtags introduced to Twitter, where both the structure of the network and the information cascades are real. Altogether, our work presents the first systematic analysis of hiding the source of social diffusion. 

\section{Results}

\subsection{Problem Overview and Theoretical Analysis}

We consider the problem of hiding the source of a diffusion process in a network. In particular, we consider an undirected network $G=(V,E)$, where one of the nodes, $\vs \in V$, starts a diffusion process, resulting in a subset of nodes $\II$ becoming infected. In this work we assume that this process follows the Susceptible-Infected (SI) model~\cite{kermack1927contribution}; see Methods for a formal description of the network notation and the SI model. We consider situations where $\vs$ wishes to avoid being detected as the source of the diffusion process. Hence, we call the node $\vs$ \textit{the evader}. We also assume the existence of another entity, called \textit{the seeker}, whose goal is to identify the origin of the diffusion using source detection algorithms. In our analysis, we consider source detection algorithms that return a ranking of network nodes~\cite{comin2011identifying,shah2011rumors,jain2016fast,antulov2015identification}, where the node at the top position in the ranking is identified as the source; see Methods for more details. The goal of the evader is then to introduce modifications to the network structure (after the diffusion has taken place) in order to avoid being identified by the seeker as the source.
We consider the two (as we argue below) most realistic types of such modifications: (1) adding nodes and (2) modifying edges.

Next, we provide more details about the two types of modifications, starting with the one in which nodes are added to the network. For instance, these can be individuals working for the evader, who deliberately position themselves in strategic locations within the social network to confuse the seeker. We refer to such individuals as ``confederates'' throughout the article (although this term suggests that they are willingly and knowingly helping the evader, this does not have to be the case). Then, the problem faced by the evader is to determine the contacts of each confederate. Thus, although the evader wishes to hide by adding nodes, the optimization problem faced by the evader is to choose which edges, not nodes, to add to the network. Note that this is a variation of the well-known Sybil attack~\cite{douceur2002sybil}, where an entity affects a system by using multiple identities. We also consider an alternative way in which the evader may conceal their true nature as the source of the diffusion. Instead of adding confederates to the network, the evader can modify (i.e., add or remove) the network edges after the diffusion has taken place. For instance, the evader could claim to have met someone when in reality they have not, or vice versa, hoping that such modifications would mislead the source detection algorithms.

\begin{figure}[t!]
\centering
\includegraphics[width=\linewidth]{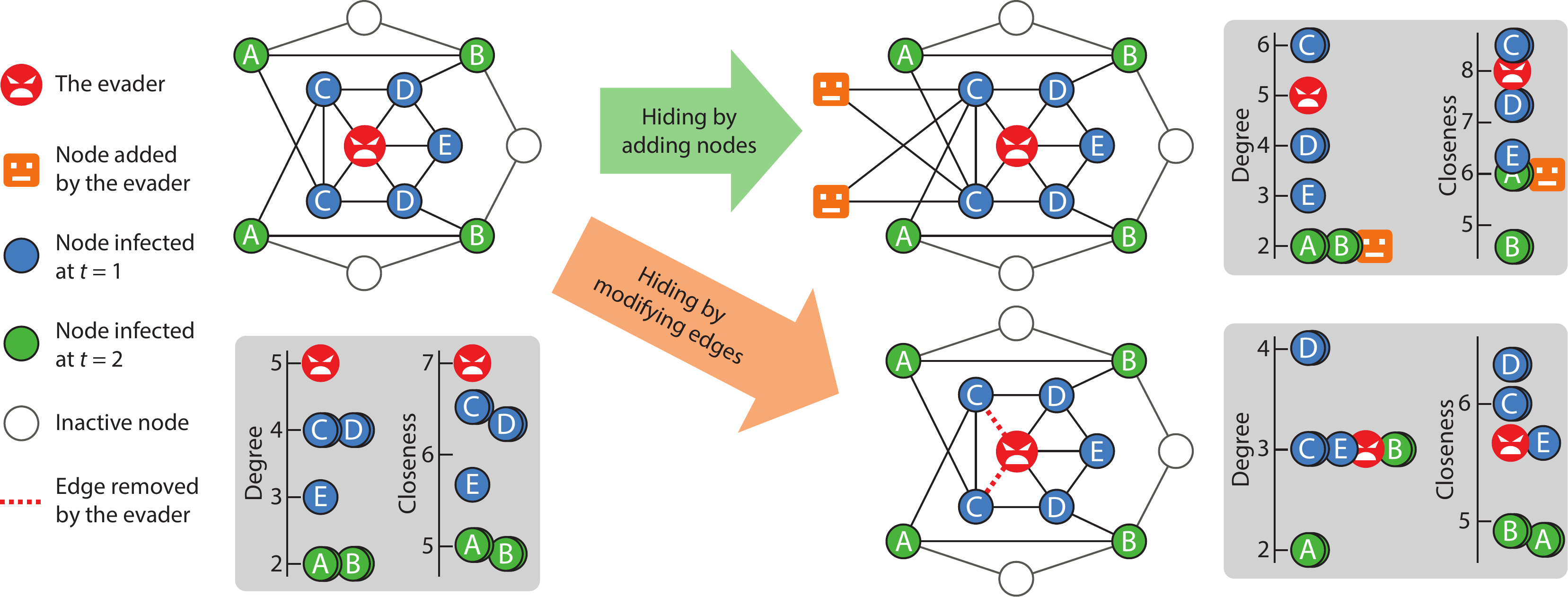}
\caption{
\textbf{An overview of the hiding process.}
The network on the left represents the original structure before the hiding process. Here, the red node is the evader, and the blue nodes are those infected in the first round of diffusion, i.e., at time step $t=1$, while the green nodes are those infected in $t=2$. The grey frame next to this network depicts the rankings computed by two source detection algorithms, namely Degree and Closeness (nodes that have the same ranking were given the same label, e.g., see the two nodes labeled $C$). The numbers on the scale represent the score assigned by each algorithm, implying that the node positioned at the top is the most likely source according to that algorithm. The network in the top-right corner represents a possible scenario where the evader tries to hide by adding two confederates, and connecting them to the nodes labeled $C$. The network in the bottom-right corner represents another scenario where the evader tries to hide by removing the dotted edges. The gray frames next to these networks show how the rankings change as a result of these modifications.
}
\label{fig:infographic}
\end{figure}

Figure~\ref{fig:infographic} illustrates an example of the hiding process. The network on the left represents the original structure, i.e., the one in which the diffusion takes place. Here, we focus on two source detection algorithms, namely Degree (which counts the number of infected neighbors that one has) and Closeness (which measures one's distance to other infected nodes); see Methods for formal definitions of these algorithms. According to both algorithms, the evader (represented as the red node) is identified as the source of diffusion (see how the evader occupies the top position in the ranking produced by each algorithm). The networks on the right illustrate two possible scenarios in which the evader hides its identity. The evader could avoid detection by introducing newly infected nodes, as exemplified in the top-right network. After adding the two confederates, the evader connects them to the nodes labeled $C$ (e.g., asks them to claim they were in contact with the nodes labeled $C$). Consequently, the evader drops to the third position in the rankings produced by both the Degree and Closeness algorithms. Alternatively, the evader could try to avoid detection by removing the edges between themselves and the two nodes labeled $C$ as exemplified in the bottom-right network (a possible interpretation of this action is that the evader denies having met these two individuals). As a result, the evader drops to the third position in the ranking produced by the Degree source detection algorithm, and to the fifth position in the ranking produced by the Closeness source detection algorithm, thereby concealing their identity as the source of diffusion. As can be seen, by performing a relatively small number of network modifications, the evader can lower their chances of being identified as the source. It should be noted that the scenario illustrated in Figure~\ref{fig:infographic} is just an example; in principle the added nodes do not necessarily have to be connected to the source's neighbors, but can be connected to any other node in the network. Likewise, modifying the network edges does not have to be done by removing existing ones; it could also be done by adding edges to the network. The effectiveness of the different choices will be evaluated in our subsequent experiments.

\begin{table}[t!]
\centering
\begin{tabular}{ l c c }
\toprule
Source detection algorithm & Adding nodes & Modifying edges \\
\midrule
Degree & P & NP-complete \\
Closeness & NP-complete & NP-complete \\
Betweenness & NP-complete & NP-complete \\
Rumor & NP-complete & NP-complete \\
Random Walk & NP-complete & NP-complete \\
Monte Carlo & NP-complete & NP-complete \\
\bottomrule
\end{tabular}
\caption{\textbf{Summary of our computational complexity results.} For different source detection algorithms, we consider the decision problem that the evader must solve in order hide optimally from the algorithm by either adding nodes or modifying edges. P = solvable in polynomial time; NP-complete = Non-deterministic Polynomial-time complete, implying that no known algorithm can solve it in polynomial time.}
\label{tab:decision-summary}
\end{table}

The first question we investigate is: \textit{How difficult is it to find an optimal way of hiding the source of diffusion?} Formal definitions of the decision problems faced by the evader are presented in Appendix~\ref{app:formal-definitions}. Our theoretical findings regarding the computational complexity of these problems are summarized in Table~\ref{tab:decision-summary}, while the formal proofs are presented in Appendix~\ref{app:complexity-proofs}. As can be seen in the table, in almost all cases, the problems under consideration are NP-complete (Non-deterministic Polynomial-time complete), implying that no known algorithm can solve them in polynomial time relative to the network size. Hence, finding an optimal way of preventing the evader from being identified as the source of diffusion is a computationally intractable task that cannot be completed efficiently, especially for large networks.

\subsection{Heuristics}

Given the computational complexity of identifying an optimal way of hiding the source of diffusion, we will now focus on heuristic methods instead. The first class of heuristics that we consider is adding confederates to the network. We use the term ``supporters'' to describe the nodes that are already present in the network and are willing to accept connections from the confederates. These supporters do not necessarily need to be the evader's associates, and are not restricted to those who are intentionally cooperating to hide the source. Then, the evader must optimize the list of contacts of each confederate, which can include any of the supporters and any of the other confederates. Let us first consider how contacts are chosen from the list of supporters. Assuming that the evader wishes to connect each confederate to $k$ supporters, we consider three alternative heuristics:
\begin{itemize}\itemsep-0.1cm
\item \textbf{Hub}---for each confederate, connect it to the $k$ supporters with the greatest degrees (this way, all confederates get connected to the same $k$ supporters);
\item \textbf{Degree}---for each confederate, connect it to the $k$ supporters with the greatest degree out of those who are not yet connected to any other confederate (if no such supporters exist, select from the ones connected to the smallest number of confederates);
\item \textbf{Random}---for each confederate, connect it to $k$ supporters chosen uniformly at random.
\end{itemize}

Each of the above heuristics has two versions, depending on how the confederates are connected to each other. In particular, we consider two possibilities:
\begin{itemize}\itemsep-0.1cm
\item \textbf{Just supporters}---every confederate is connected only to supporters, implying that there are no edges between confederates;
\item \textbf{Clique}---every confederate is connected to every other confederate, implying that the confederates form a clique.
\end{itemize}
For each heuristic, we add the word ``clique'' to indicate that the confederates are connected to each other, e.g., by writing ``Hub clique''. Otherwise, if there are no connections between the confederates, we write the name as it is, e.g., ``Hub''.

Now that we have presented our first class of heuristics, which add confederates to the network, let us now consider the second class of heuristics, which modify edges in the network. We assume that the evader can only add or remove edges between themselves and a specific subset of nodes. To determine which of those nodes to connect to, and which to disconnect from, we consider three alternative heuristics:

\begin{itemize}\itemsep-0.1cm
\item \textbf{Max degree}---choose the nodes with the greatest degree;
\item \textbf{Min degree}---choose the nodes with the smallest degree;
\item \textbf{Random}---choose nodes uniformly at random.
\end{itemize}
All ties are broken uniformly at random. Each of the above heuristics has two versions, depending on whether the evader is adding new connections, or removing existing ones. We write the word ``adding'' to indicate that the evader is adding new connections, e.g., by writing ``Adding max degree''. Otherwise, we write the word ``removing'' to indicate that the evader is removing existing connections, e.g., ``Removing max degree''.

\subsection{Simulation Analysis}

The experimental procedure for a given network $G = (V,E)$ is as follows. First, we select the evader $\vs$ uniformly at random from the top $10\%$ of nodes according to degree ranking, provided that its degree is at least $10$.
We imposed this restriction because some of our heuristics remove edges that are incident to the evader, implying that the evader needs to have enough edges to actually execute these heuristics. Besides, the spreaders of fake news tend to be well connected~\cite{bovet2019influence}. After the evader is selected, we spread the diffusion starting from $\vs$ to obtain the set of infected nodes $\II$. In our simulations, we use the SI model with the probability of diffusion being $p=0.15$ and the number of rounds being $T=5$. We then perform the hiding process using different heuristics, recording the position of $\vs$ in the rankings generated by each source detection algorithm after each step of the hiding process.

The majority of the source detection algorithms considered in our experiments disregard the nodes that are not infected. Hence, when choosing the confederate's contacts, it makes sense for the evader to consider only infected supporters. In our simulations, we assume that all infected nodes other than the evader are supporters. Furthermore, we assume that the evader can only remove edges that they are part of and only add edges between themselves and their neighbors' neighbors.

\begin{figure}[t!]
\centering
\includegraphics[width=\linewidth]{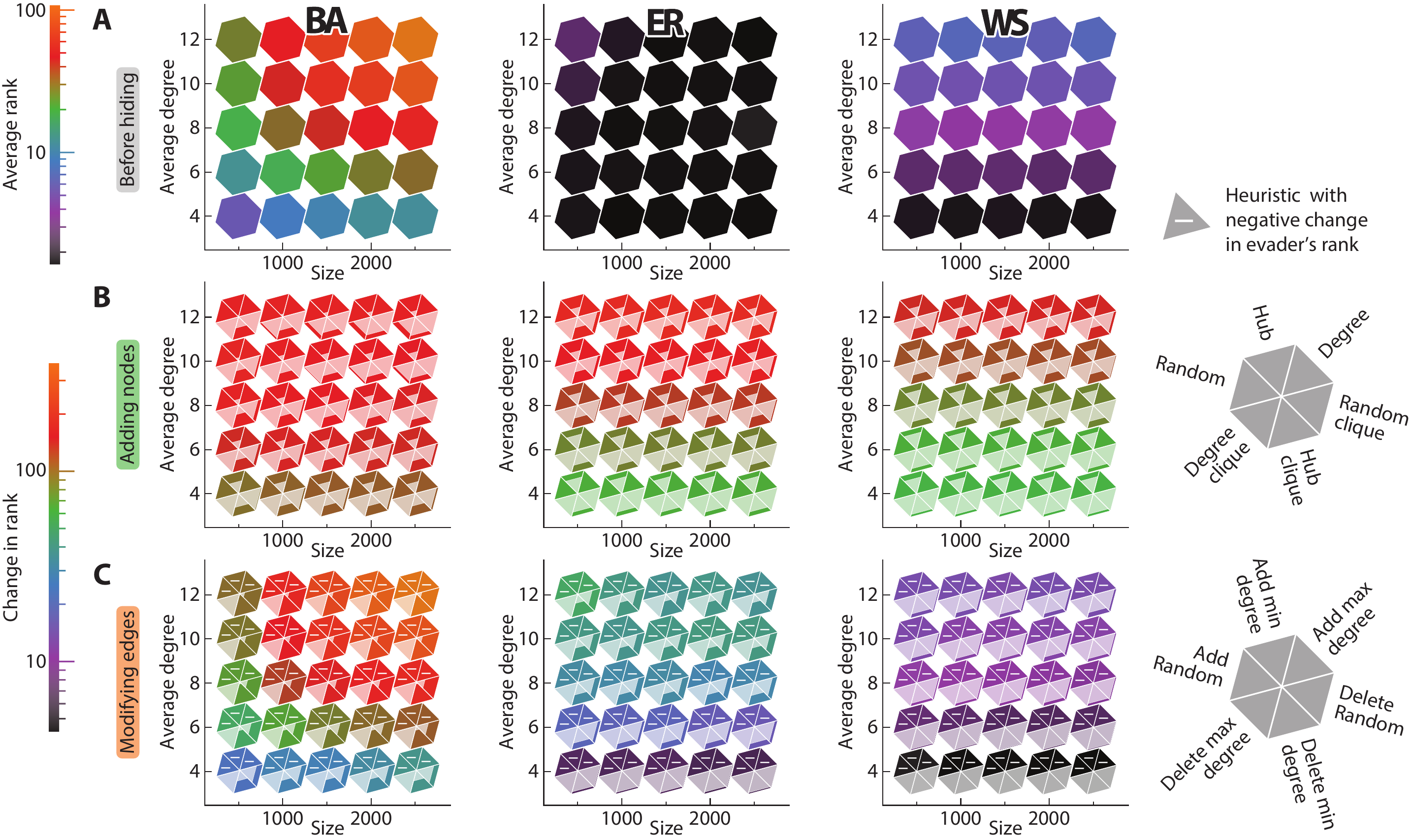}
\caption{
\textbf{The efficiency of hiding from the Eigenvector source detection algorithm in networks with varying structure, size and density.}
In each color-coded tessellation, the x-axis represents the number of nodes in the network, while the y-axis represents the average degree.
The first column corresponds to scale-free networks generated using the \BAn (BA) model, the second corresponds to random networks generated using the \ERn (ER) model, and the third corresponds to small-world networks generated using the Watts-Strogatz (WS) model. The results are presented as an average over $100$ networks and over $10$ evaders in each network. 
The colors of the first row (A) gives the average position of the evader in the ranking computed using the Eigenvector source detection algorithm before the hiding process on a logarithmic scale (the lower the ranking, the more exposed is the evader).
The second (B) and third (C) rows depict the average difference in the evader's ranking as a result of the hiding process, where each hexagon indicates the efficiency of six different hiding heuristics, and the color represents the effectiveness of the best heuristic (on a logarithmic scale).
More specifically, row (B) depicts the effect of adding $50$ confederates connected to $3$ supporters each, while row (C) depicts the effect of either adding or removing $5$ edges, depending on the heuristic being used. Additionally, each triangular sector is filled up with a lighter shaded area (overlaid the hexagons), representing the relative effectiveness of the corresponding heuristic compared to the best one (the greater the shaded area, the better the heuristic). As such, triangular sectors that have no shaded area have zero effectiveness. However, there are cases where the effectiveness is even less than zero, i.e., it backfires and ends up exposing the evader even more. In such cases, the triangular sector has no shaded area and is also marked with a minus.
}
\label{fig:profile}
\end{figure}

In our experiments, we will first disentangle two different aspects of hiding. The first aspect relates to the network topology itself, which can provide some concealment even without the evader manipulating it. The second aspect comes from an evader strategically manipulating the network after the inception of the diffusion process. To separate the two notions of hiding, we ran experiments on networks with varying structure, size, and density; see Figure~\ref{fig:profile}. The figure presents the results for the Eigenvector source detection algorithm, in particular. We chose this example since it is one of the best-performing algorithms and yields the most pronounced differences between the best and worst hiding heuristics;
see Appendix~\ref{app:hiding-profiles} for the results pertaining to other source detection algorithms. Figure~\ref{fig:profile}A presents the results for the first notion of hiding, i.e., the one stemming from the network structure itself, whereas Figures~\ref{fig:profile}B and \ref{fig:profile}C present the results for the second notion of hiding, which stems from strategically adding confederates or modifying edges, respectively.
In all of our simulations, we report the absolute, rather than relative, ranking of the evader according to the source detection algorithm in question. However, in principle, one could suspect an individual to be the source of diffusion if that individual is, e.g., among the top 10 nodes, or the top 1\% of nodes. Since the two are correlated, we chose one of them---the absolute ranking---and used it throughout all of our experiments.

As can be seen in Figure~\ref{fig:profile}A, out of the three network structures considered in our experiments---scale-free, small-world, and random---the one that provides the greatest level of concealment to the evader is the scale-free structure. Moreover, independent of the network model, the denser the network, the more concealed is the evader. As for the network size, having a larger number of nodes results in a greater level of concealment for scale-free networks, but results in a negligible effect for small-world or random networks.
When it comes to strategic hiding via network manipulations, Figures~\ref{fig:profile}B and \ref{fig:profile}C show that it is generally more efficient to strategically hide in networks with greater density. When comparing the different structures in terms of how they facilitate the strategic hiding, our heuristics are most efficient in scale-free networks, and least efficient in small-world networks, regardless of whether the evader is adding confederates, or modifying edges. Finally, commenting on how the network size affects the efficiency of strategic hiding, the effect is relatively small. The only exception is when hiding by modifying edges in scale-free networks, which is considerably more effective in larger networks.
Next, we compare heuristics of the same type, starting with the ones that add confederates, to determine whether they should create a clique amongst themselves or remain disconnected from one another, and determine which supporters to connect to which confederates. As for the former question, creating a clique is consistently superior (see how the shaded area of the triangular sectors in Figure~\ref{fig:profile}B is greater for heuristics with ``clique'' in their name). As for the question of which supporters to connect to which confederates, when confederates form a clique, it is more fruitful to connect confederates to different supporters (using either the Random clique or Degree clique heuristics) than connecting them all to the same supporters (using the Hub clique heuristic). On the other hand, when confederates are disconnected from each other, the results vary depending on the source detection algorithm being used; see Appendix~\ref{app:hiding-profiles}.
Having compared the heuristics that add confederates, we now compare the heuristics that modify (some of) the edges that are incident to the evader, to determine whether we should add or remove edges. Our results indicate that the latter is significantly more effective. In fact, adding new edges often backfires, and ends up exposing the evader even more to the source detection algorithm. The only remaining question is to determine which edges to remove from the network. Our results show that the most effective choice is to disconnect the evader from the neighbors with the greatest degrees, and the least effective choice is to disconnect from those with the lowest degrees. All the results in Figure~\ref{fig:profile} are shown after the heuristics have made all the modifications to the network. To see how the evader's ranking changes after each such modification, see Appendices~\ref{app:simulation-subos} and \ref{app:simulation-edges} for the heuristics that add confederates and modify edges, respectively.

\begin{figure}
\centering
\includegraphics[width=0.85\linewidth]{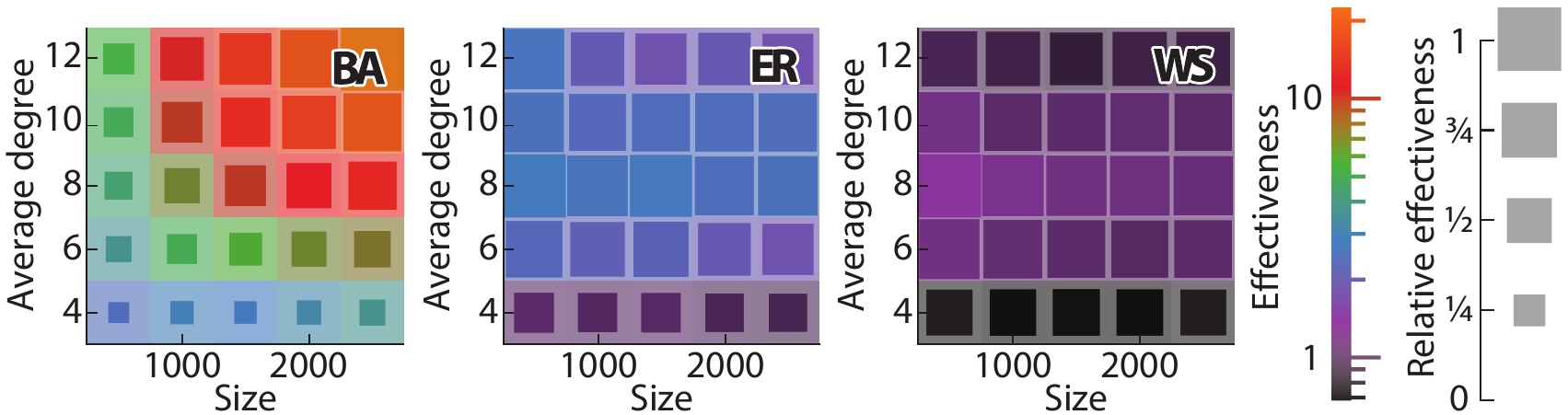}
\caption{
\textbf{Comparing the effectiveness of adding confederates vs.~modifying edges when hiding from the Eigenvector source detection algorithm.}
The effectiveness of modifying an edge is calculated as the number of confederates that must be added to achieve the same change in the evader's ranking (assuming that the modification of edges and the addition of confederates are done using the best respective heuristics). The heatmaps evaluate the effectiveness of modifying an edge while varying the network structure. In particular, the left-most heatmap corresponds to scale-free networks generated using the \BAn (BA) model, the central heatmap corresponds to random networks generated using the \ERn (ER) model, and the right-most heatmap corresponds to small-world networks generated using the Watts-Strogatz (WS) model. In each heatmap, the x-axis represents the number of nodes in the network, while the y-axis represents the average degree. The color of each cell corresponds to the average effectiveness of modifying an edge. The size of the unshaded (darker) area in each cell corresponds to the average effectiveness of modifying an edge in that cell, relative to the maximum effectiveness across all cells of that panel. The results are presented as an average over $100$ networks and over $10$ evaders in each network, after adding $50$ confederates or modifying $5$ edges. The colors in the heat maps reflect a logarithmic scale.
}
\label{fig:exchange}
\end{figure}

It is difficult to compare the effectiveness of adding nodes and modifying edges based solely on Figure~\ref{fig:profile}, since the figure shows only the impact of adding $50$ nodes and modifying $5$ edges. To facilitate this comparison, Figure~\ref{fig:exchange} shows how many nodes must be added to the network in order to have the same effect as modifying a single edge. As can be seen, in the majority of cases the effect of modifying a single edge is equivalent to adding several confederates. In fact, the number of confederates needed to achieve the same effect as modifying a single edge is surprisingly large (and may even reach tens) in scale-free networks. The only exception is the case of sparse, small-world networks, where adding one confederate affects the evader's ranking more than modifying a single edge (as indicated by the values smaller than $1$ in the heatmap). The results depicted in the figure are for the Eigenvector source detection algorithm; the results for other source detection algorithms are qualitatively similar as shown in Appendix~\ref{app:exchange}.

\begin{figure}[t!]
\centering
\includegraphics[width=\linewidth]{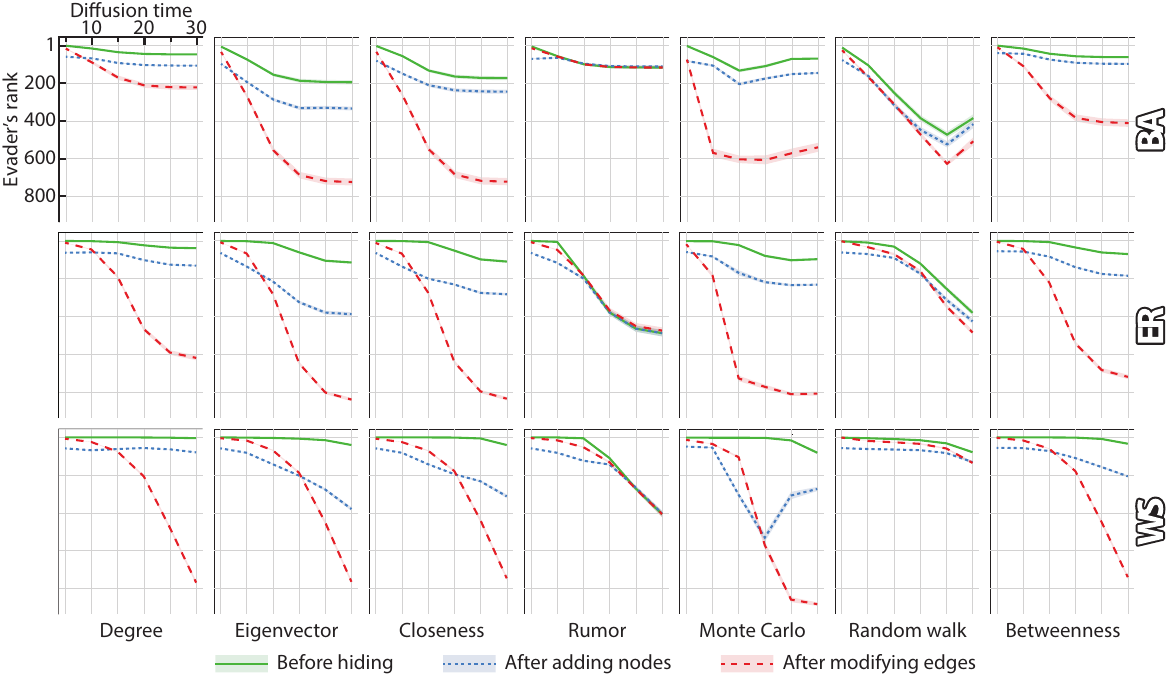}
\caption{
\textbf{The impact of extending the diffusion time on the effectiveness of hiding.}
In each plot, the x-axis represents the diffusion time, i.e., the total number of rounds in the diffusion process. The y-axis represents the evader's ranking according to the source detection algorithm; this ranking is computed before the hiding process (green line), after running the best heuristic that adds confederates (blue line), and after running the best heuristic that modifies edges (red line). Each column corresponds to a different source detection algorithm, while each row corresponds to a different model used to generate the network in which the diffusion takes place. Results are averaged over $10$ evaders and over $100$ networks generated using either the \BAn (BA), the \ERn (ER), or the Watts-Strogatz (WS) model, consisting of $1,000$ nodes each, with the average degree being $4$. The axes in all plots are identical to those used in the upper-left corner. Shaded areas (which are often too small to see) represent $95\%$-confidence intervals.
}
\label{fig:duration}
\end{figure}

Another aspect that may impact the effectiveness of hiding the evader is the diffusion time, i.e., the total number of rounds completed in the diffusion process before the source detection algorithm analyzes the network. The results of this analysis can be found in Figure~\ref{fig:duration}. In the vast majority of cases, the evader becomes more hidden as diffusion time increases. This is true not only when the evader performs no modifications to the network, but also when they modify the network by adding confederates or by removing edges following the most effective heuristic (although the effectiveness of removing edges grows at a greater rate than that of adding confederates). This suggests that if our goal is to identify the source of diffusion, we should start our investigation as early as possible. Shah~\etal~\cite{shah2020finding} reported similar findings, but for different diffusion models than ours, namely Susceptible-Infected-Recovered (SIR) and Susceptible-Exposed-Infected-Recovered (SEIR). Next, we compare the source detection algorithms to each other. As can be seen from the figure, the diffusion time's sensitivity varies greatly from one algorithm to another. When the evader performs no modifications, the Degree and Betweenness algorithms prove to be the most resilient. Similar results are observed when the evader adds confederates to the network. In contrast, when modifying edges, the most resilient algorithms are Rumor and Random Walk, i.e., they are the least affected by changing the diffusion time. Interestingly, all three types of network structures show relatively similar patterns, suggesting that the source detection algorithms' inner workings play a more important role in determining how the diffusion time affects the effectiveness of hiding, rather than the network characteristics.

\begin{figure}[t!]
\centering
\includegraphics[width=.9\linewidth]{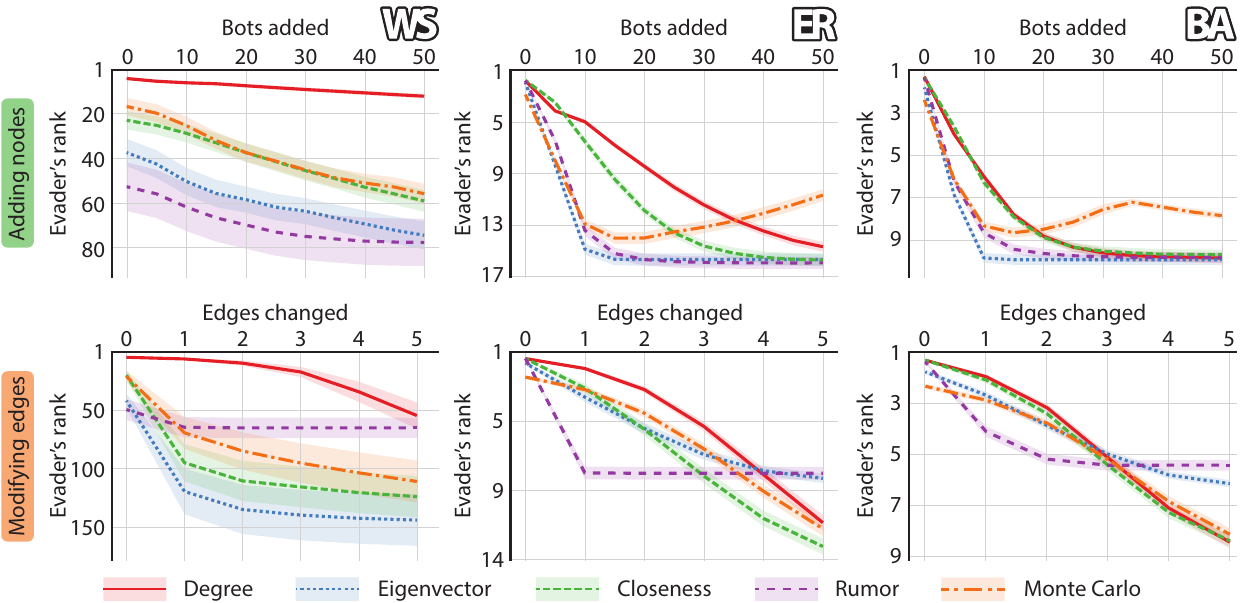}
\caption{
\textbf{The effectiveness of hiding in massive networks.}
Given networks of $100,000$ nodes with an average degree of $4$, the figure depicts the evader's ranking (y-axis) as a function of the number of network modifications (x-axis), using the best heuristic for adding nodes (first row) or modifying edges (second row). Different colors represent different source detection algorithms. The results are presented as averages taken over $100$ networks generated using either the \BAn (BA), the \ERn (ER), or the Watts-Strogatz (WS) model, and over $10$ different evaders in each network.
Shaded areas represent $95\%$-confidence intervals.
}
\label{fig:large}
\end{figure}

So far in our analysis, we only considered networks of up to $2,500$ nodes, since the analysis involved taking an average over a large number of cases, and increasing the number of nodes would have taken excessive time. Fortunately, when it comes to evaluating the impact of the hiding process, we can approximate it even for massive networks. Based on this, we increased the number of nodes to $100,000$, and approximated the evader's ranking after each step of the hiding process. The approximation is done by computing the ranking of the evader not among all nodes, but rather among $10,000$ nodes, consisting of the $5,000$ infected nodes with the greatest degrees and another $5,000$ infected nodes chosen uniformly at random from the remaining ones. Furthermore, in our approximation we do not consider the Betweenness and Random walk source detection algorithms, since their ranking cannot be efficiently computed for just a selected subset of nodes. The results of this analysis are presented in Figure~\ref{fig:large}. In networks generated using the \ERn model and the Watts-Strogatz model, the evader usually occupies the top position of the ranking before hiding, whereas in networks generated using the \BAn model, they are in the top $50$ positions. Moreover, in the former two types of networks, the hiding process seems much less effective than in the latter, regardless of whether the hiding is done by adding nodes or by modifying edges. Notice that the effectiveness of hiding in these massive networks is considerably reduced compared to smaller networks with $1,000$ nodes, the results for which were presented in previous figures. These findings are all based on the best heuristic of each type; see Appendix~\ref{app:simulation-large} for an evaluation of the remaining heuristics.

It should be noted that the aforementioned findings are all based on cascades generated using the SI model, taking place in randomly generated networks. In Appendix~\ref{app:simulation-real} we replace the random networks with real ones (while still using the SI model), and in Appendix~\ref{app:real-cascades} we present results where both the networks and the cascades are taken from real data.
Altogether, the results presented in Appendices~\ref{app:simulation-real} and~\ref{app:real-cascades} validate our findings. More specifically, they confirm that our heuristics are capable of reducing the effectiveness of source detection algorithms. They also confirm that adding edges can backfire and end up exposing the source even more, while adding nodes and removing edges rarely backfires. Finally, the most effective heuristics on synthetic data tend to also be among the most effective ones on real data.

We conclude our analysis by exploring different variations of the experimental setup. First, we ran simulations with alternative models of diffusion and network generation (Appendix~\ref{app:altmodels}); the results were qualitatively similar to those obtained from our main experimental setup. Second, we investigated how the evader's ranking is affected when another node, $v$, located in the evader's direct network vicinity runs our heuristics (Appendix~\ref{app:distance}); we found that the evader becomes hidden as a result of $v$'s actions, albeit not as effectively as in the case when the evader is the one running the heuristic. Third, we investigated how the completeness of the seeker's knowledge about the network's structure affects the effectiveness of the evader's hiding (Appendix~\ref{app:fuzzy}); we found that in the vast majority of cases the evader's hiding becomes significantly more effective as the knowledge of the seeker dwindles.
Fourth, instead of selecting the evader randomly from the $10\%$ of nodes with the highest degrees (as in our basic experiments), we analyzed the case where the evader is randomly selected out of all nodes (Appendix~\ref{app:any-evader}). The results were broadly similar to those observed in our basic experiments, i.e., when the evader does not try to hide, they are ranked high according to source detection algorithms, but when the evader runs the hiding heuristics, they are able to significantly reduce the likelihood of being identified. 
Fifth, we performed various sensitivity analyses by modifying the parameters used in our original simulations (Appendix~\ref{app:sensitivity}). To this end, we started by analyzing how the hiding process is affected by the way in which the evader is selected. Specifically, instead of selecting the evader from the top $10\%$ of nodes with the highest degrees (as was the case with our main experiments), we selected the evader from the top $x\%$ where $10\leq x \leq 90$; the results suggest that the value of $x$ has a relatively small effect on the outcome. Next, instead of setting the rewiring probability of the Watts-Strogatz model to be $0.25$, we varied this probability from $0.05$ to $0.5$, and found that the hiding process tends to be slightly more effective given greater values of this parameter. Finally, instead of approximating the evader's ranking by computing it among $10,000$ infected nodes, we varied this number between $2,000$ and $18,000$, and found this to have a relatively small effect on the outcome.

\section{Discussion}

In this work, we analyze the possibility of obfuscating the origin of diffusion, both as a result of spreading it in a specific type of network structure, and via strategic network manipulations. On the one hand, our theoretical analysis indicates that finding an optimal way of hiding the source of diffusion, either by adding confederates or by modifying the networks' edges, is computationally intractable. On the other hand, our computational experiments demonstrate that even without any strategic manipulations, the structure of the network itself can greatly hinder the efforts to identify the diffusion source. This seems to be the case especially in scale-free networks---an observation that is particularly alarming since many real-life social networks exhibit this property. We also find that the task of identifying the source of diffusion is more challenging in networks that are dense, allowing the culprit to hide in the crowd. Moreover, an adversarial agent can utilize several network modifications simultaneously to obfuscate the source even more. Particularly effective strategies in this regard are based on attaching a densely connected group of confederates to the network and removing connections between the source of the diffusion and its most well-connected neighbors after the diffusion has started. Our analysis also confirms previous findings derived from alternative diffusion models~\cite{shah2020finding}, indicating that the longer the diffusion takes, the more difficult it is to pinpoint its origin, highlighting the importance of prompt reaction to any potential epidemic threat. Finally, our experiments indicate that finding the diffusion source is easier in massive, sparse networks, regardless of whether the evader strategically manipulates the network to hide its identity. Still, given current source detection algorithms, if the diffusion source tries to hide, it probably will succeed. Future algorithms will need other types of information capturing the time evolution of both the diffusion and the network.

There exists a growing literature on avoiding detection by a wide range of social network analysis tools. Such hiding techniques can be used to prevent a closely-cooperating group of nodes from being identified by community detection algorithms~\cite{waniek2018hiding}, or prevent the leader of the organization from being recognized by centrality measures in both standard~\cite{waniek2017construction,was2020manipulability} and multilayer networks~\cite{waniek2020hiding}. Similar techniques can be used to prevent an undisclosed relationship from being exposed by link prediction algorithms~\cite{waniek2019hide,zhou2019attacking}. Nevertheless, none of the existing works considered strategically hiding the source of diffusion from source detection algorithms. What is more, they typically only study hiding by adding edges to, and removing edges from, a network while disregarding the possibility of avoiding detection by adding nodes to the network.
Another body of work that is relevant to our study is the one considering the reconstruction of diffusion cascades. Typically in this literature, the party analyzing the network has at its disposal more information than what is available to the seeker in our setting. Specifically, many works require knowledge about either the exact moment when each node was reached by the diffusion process~\cite{taxidou2014online,farajtabar2015back,xiao2018reconstructing}, or the order in which the nodes became infected~\cite{ghosh2011framework}. Others consider settings where available data includes edges over which the diffusion took place~\cite{sadikov2011correcting}. Some studies assume that the underlying structure is a temporal network~\cite{rozenshtein2016reconstructing}, where each edge exists only in a specific moment in time, or assume the availability of an API serving data about the diffusion~\cite{cogan2012reconstruction}. Yet another body of literature tries to infer which nodes are actually in the infected state, based on partial information~\cite{xiao2018robust,sundareisan2015hidden}. Finally, we mention a growing literature focusing on the reconstruction of the network structure based on information about diffusion cascades~\cite{shen2014reconstructing,gomez2013structure,braunstein2019network,pajevic2009efficient,chwistek2020network,li2017reconstruction}.

In our study, we kept the models, the simulations and the theoretical analysis as generic as possible, without making assumptions about the nature of the social diffusion under consideration.
Nevertheless, it should be noted that the effectiveness of our heuristic depends on the specifics of the scenario at hand. For instance, hiding certain facts (such as coming in contact with a certain individual) might be easier if the diffusion is taking place in the real world as opposed to the virtual world. Consider social media platforms such as Twitter or Facebook, where all the online activities are logged by the platform. In such cases, the source can no longer hide its actions from the platform administrators (since they have access to these logs). Having said that, the source may still be able to hide from other observers who cannot access such logs and can only view publicly available information. Another noteworthy aspect of our analysis is that we only consider source detection from the perspective of identifying the node (e.g., the Facebook account) that initiated the diffusion process. This should not be confused with the problem of identifying the actual entity represented by the source node (e.g., the person or organization behind the Facebook account); such identification is out of the scope of our study. Finally, we comment on the applicability of source detection algorithms in the real world. As mentioned earlier, we avoided incorporating domain-specific details in our experiments to keep the implications as broad as possible. In practice, however, we suspect that source detection algorithms would be augmented by additional, domain-specific information. Having said that, our evaluation of these algorithms on real-life cascades (Appendix~\ref{app:real-cascades}) suggests that, even without such information, the algorithms can be effective, at least in narrowing down the search for the source.

The direct policy implication of our work is that one has to be sure the source has not been trying to hide itself to trust the source detection algorithms of today. This is because, as our experiments have shown, such algorithms can easily be fooled by a strategic source who is actively attempting to escape detection. Having said that, our work also points to the future---the necessary elements of the next generation’s source detection algorithms. In particular, since hiding by manipulating edges is so efficient compared to adding fake nodes, new algorithms need to identify spurious links---connections added since the start of the diffusion. Arguably, the need for such algorithms is more pressing than ever. Developing tamper-proof source detection algorithms would improve our chances of detecting the source of a viral cascade in the future.

\section{Methods}

\subsection{Basic Network Notation}

Let us denote by $G = (V, E) \in \G$ a network, where $V$ is the set of $n$ nodes and $E \subseteq V \times V$ is the set of edges. We denote an edge between nodes $v$ and $w$ by $(v,w)$, and we only consider \textit{undirected} networks, implying that we do not discern between edges $(v,w)$ and $(w,v)$. Moreover, we assume that networks do not contain self-loops, i.e., $\forall_{v \in V}(v,v) \notin E$. We denote by $\ER$ the set of all non-edges, i.e., $\ER = (V \times V) \setminus \left(E \cup \bigcup_{v \in V} (v,v)\right)$. A path in a network $G = (V,E)$ is an ordered sequence of distinct nodes, $\langle v_1, \ldots, v_k\rangle$, in which every two consecutive nodes are connected by an edge in $E$. We consider the length of a path to be the number of edges in that path. The set of all shortest paths between a pair of nodes, $v,w \in V$ is denoted by $\ps_G(v,w)$, while the distance between a pair of nodes $v,w \in V$, i.e., the length of a shortest path between them, is denoted by $d_G(v,w)$. Furthermore, a network is said to be \textit{connected} if and only if there exists a path between every pair of nodes in that network. We denote by $N_G(v)$ the set of \emph{neighbors} of $v$ in $G$, i.e., $N_G(v) = \{w \in V : (v,w) \in E\}$. We denote by $G^{V'}$ the subnetwork of $G$ induced by the nodes in $V' \subseteq V$, i.e., $G^{V'} = (V', E \cap (V' \times V'))$. Finally, for $E' \subseteq V \times V$ we denote by $G \cup E'$ the effect of adding set of edges $E'$ to $G$, i.e., $G \cup E' = (V, E \cup E')$. To make the notation more readable, we will often omit the network itself from the notation whenever it is clear from the context, e.g., by writing $d(v,w)$ instead of $d_G(v,w)$. This applies not only to the notation presented thus far, but rather to all notation in this article.

\subsection{Susceptible-Infected Model and Source Detection Algorithms}

In the Susceptible-Infected (SI) model, every node in the network is in one of two states: either susceptible (prone to be affected by the phenomenon) or infected (already affected by the phenomenon). The modeled process consists of discrete rounds. At the beginning of the process only the nodes belonging to the \textit{seed set} are in the infected state (in this work, the seed set consists of only the evader $\vs$). In every round $t$, every infected node makes each of its susceptible neighbors infected with probability $p$. The process ends after a certain number of rounds $T$. We denote the set of infected nodes after $T$ rounds by $\II$.

A source detection algorithm is a procedure that, based on the network $G$ and the set of infected nodes $\II$, aims at determining the source of diffusion. Every source detection algorithm considered in this work can be represented as a function $\sd: V \times \G \times 2^V \rightarrow \R$ that assigns the score $\sd(v,G,\II)$ to any node $v$, where the node with the highest score is selected by the algorithm as the most probable source of diffusion. We will assume that for any node $v \notin \II$ and any source detection algorithm we have $\sd(v,G,\II) = -\infty$ (as the seed node has to be infected and there is no mechanism of coming back to the susceptible state). In this work we focus on the source detection algorithms that are designed to detect the source of a diffusion process with a seed set consisting of only one node (see Shelke and Attar~\cite{shelke2019source} for a review of multiple source detection algorithms). More specifically, we consider the following source detection algorithms:

\begin{itemize}

\item Degree~\cite{comin2011identifying}---the score assigned to a given $v \in \II$ is the degree centrality of $v$ in $G^\II$, i.e.:
$$
\sddegr(v,G,\II) = |N_{G^\II}(v)|;
$$

\item Closeness~\cite{comin2011identifying}---the score assigned to a given $v \in \II$ is the closeness centrality of $v$ in $G^\II$, i.e.:
$$
\sdclos(v,G,\II) = \frac{1}{\sum_{w \in \II} d_{G^\II}(v,w)};
$$

\item Betweenness~\cite{comin2011identifying}---the score assigned to a given $v \in \II$ is the betweenness centrality of $v$ in $G^\II$, i.e.:
$$
\sdbetw(v,G,\II) = \sum_{u \neq w : u,w \in \II \setminus \{v\}} \frac{|\{\pi \in \Pi_{G^\II}(u,w): v \in \pi\}|}{|\Pi_{G^\II}(u,w)|};
$$

\item Eigenvector~\cite{comin2011identifying}---the score assigned to a given $v \in \II$ is the eigenvector centrality of $v$ in $G^\II$, i.e.:
$$
\sdeig(v,G,\II) = x_v
$$
where $x$ is the eigenvector corresponding to the largest eigenvalue of the adjacency matrix of $G^\II$;

\item Rumor~\cite{shah2011rumors}---the score assigned to a given $v \in \II$ is the rumor centrality of $v$ in $G^\II$, i.e.:
$$
\sdrumor(v,G,\II) = \frac{|\II|!}{\prod_{w \in \II} \Theta^v_w}
$$
where $\Theta^v_w$ is the size of the subtree of $w$ in the BFS (Breadth-First Search) tree of $G^\II$ rooted at $v$;

\item Random Walk~\cite{jain2016fast}---intended to approximate diffusion by random walks. The score of a given node $v$ is:
$$
\sdrwalk(v,G,\II) =
\begin{cases}
	\phi_0(v) & \mbox{if } \forall_{w \in \II} d_G(v,w) \leq T \\
	0 & \mbox{otherwise}
\end{cases}
$$
where $T$ is the number of rounds in the SI model and $\phi$ is as:
$$
\phi_t(v) =
\begin{cases}
	1 & \mbox{if } t=T \\
	(1-p) \phi_{t+1}(v) + \sum_{w \in N(v) \cap \II} \frac{p}{|N(v)|} \phi_{t+1}(w)  & \mbox{otherwise}
\end{cases}
$$
where $p$ is the probability of infection in the SI model.

\item Monte Carlo~\cite{antulov2015identification}---where for each node we repeated run a diffusion starting with that node and investigate for which of the nodes the infected set is the most similar to $\II$ (using Jaccard similarity). The score of a given node $v$ is:
$$
\sdmcarlo(v,G,\II) = \frac{1}{m} \sum_{i=1}^{m}\exp\left( \frac{-(\psi_J(\II,\II_{v,i}) - 1)^2}{a^2} \right)
$$
where $m$ is the number of Monte Carlo samples for each node, $\psi_J(A,B)=\frac{|A \cap B|}{|A \cup B|}$ is the Jaccard similarity measure, $\II_{v,i}$ is the set of infected nodes in the $i$-th Monte Carlo sample where the diffusion starts with $v$, and $a$ is the soft margin parameter.

\end{itemize}

There also exist more advanced source detection algorithms that are specifically designed to find the source of diffusion in tree networks~\cite{wang2014rumor,wang2015rooting,cai2018information}. However, applying them to general (i.e., cyclic) networks is exceedingly expensive in terms of computation time even for very small structures~\cite{antulov2015identification}. Other algorithms use different frameworks than the one considered in our work, where they analyze the problem of placing a number of sensors in a network that notify the user when diffusion reaches a specific node~\cite{pinto2012locating,xu2015scalable,paluch2018fast}.

\clearpage

\bibliographystyle{abbrv}
\bibliography{bibliography-source-detection}

\appendix

\clearpage
\section{Formal Definitions of the Decision Problems}
\label{app:formal-definitions}

We now formally define the computational problems faced by the evader. In what follows, let $\rnk(v,G,\sd,\II)$ denote the ranking position of $v$ among all nodes in $G$ according to source detection algorithm $\sd$ when the set of infected nodes is $\II$. More formally:
$$
\rnk(v,G,\sd,\II) = \left|\left\{ w \in V : \sd(w,G,\II) > \sd(v,G,\II) \right\}\right|.
$$
The goal of the evader is to hide by decreasing their position in the ranking produced by  $\sd$ (notice that decreasing the position in the ranking corresponds to maximizing the value of $\rnk(v,G,\sd,\II)$).

The first method of hiding that we consider is to add confederates to the network. Then, the problem faced by the evader is to determine the contacts of each confederate.
Since not every node in the network is necessarily willing to accept connections from these confederates, we define a subset of nodes, $\FS \subseteq V$, that would accept such connections. More formally, the problem can be defined as follows:

\begin{definition}[Hiding Source by Adding Nodes]
\label{def:hiding-source-subos}
The problem is defined by a tuple, $(G,\vs,\II,\sd,\thr,b,\Subo,\FS)$, where $G=(V,E)$ is a network, $\vs \in V$ is the evader, $\II \subseteq V$ is the set of infected nodes, $\sd$ is a source detection algorithm, $\thr \in \N$ is a safety threshold specifying the smallest ranking that the evader deems acceptable, $b \in \N$ is a budget specifying the maximum number of edges that can be added, $\Subo$ is the set of confederates to be added to the network, and $\FS \subseteq V$ is the set of nodes that the evader can connect to the confederates. The goal is then to identify a set $\Add \subseteq (\Subo \times \Subo) \cup (\Subo \times \FS)$ such that $|\Add| \leq b$, $(V \cup \Subo, E \cup \Add)^\II$ is connected and:
$$
\rnk\left( \vs,(V \cup \Subo, E \cup A),\sd,\II \cup \Subo \right) \geq \thr.
$$
If the algorithm $\sd$ is nondeterministic, then we require the above condition to be met for every possible realization of the algorithm.
\end{definition}

We also consider an alternative way in which the evader may conceal their true nature as the source of the diffusion. Instead of adding confederates to the network, the evader can modify (i.e., add or remove) the network edges after the diffusion has taken place. In this case, the problem faced by the evader can be defined as follows:

\begin{definition}[Hiding Source by Modifying Edges]
\label{def:hiding-source-rewiring}
The problem is defined by a tuple, $(G,\vs,\II,\sd,\thr,b,\FA,\FR)$, where $G=(V,E)$ is a network, $\vs \in V$ is the evader, $\II \subseteq V$ is the set of infected nodes, $\sd$ is a source detection algorithm, $\thr \in \N$ is a safety threshold specifying the smallest ranking that the evader deems acceptable, $b \in \N$ is a budget specifying the maximum number of edges that can be added or removed, $\FA \subseteq \ER$ is the set of edges that can be added, and $\FR \subseteq E$ is the set of edges that can be removed.
The goal is then to identify two sets, $\Add \subseteq \FA$ and $\Rem \subseteq \FR$, such that $|\Add|+|\Rem| \leq b$, $(V, (E \cup \Add) \setminus \Rem)^\II$ is connected and:
$$
\rnk\left( \vs,(V, (E \cup A) \setminus R ),\sd,\II \right) \geq \thr.
$$
If the algorithm $\sd$ is nondeterministic, then we require the above condition to be met for every possible realization of the algorithm.
\end{definition}

\clearpage

\section{Proofs of the Computational Complexity Results}
\label{app:complexity-proofs}

Table~\ref{tab:decision-summary-theorems} summarizes our findings and refers to the theorem corresponding to each result.

\begin{table}[tbh]
\centering
\begin{tabular}{ l c c }
Source detection algorithm & Modifying Edges & Adding Nodes \\
\hline
Degree & P (Theorem~\ref{thrm:npc-subos-degree}) & NP-complete (Theorem~\ref{thrm:npc-rewiring-degree}) \\
Closeness & NP-complete (Theorem~\ref{thrm:npc-subos-closeness}) & NP-complete (Theorem~\ref{thrm:npc-rewiring-closeness}) \\
Betweenness & NP-complete (Theorem~\ref{thrm:npc-subos-betweenness}) & NP-complete (Theorem~\ref{thrm:npc-rewiring-betweenness}) \\
Rumor & NP-complete (Theorem~\ref{thrm:npc-subos-rumor}) & NP-complete (Theorem~\ref{thrm:npc-rewiring-rumor}) \\
Random Walk & NP-complete (Theorem~\ref{thrm:npc-subos-rwalk}) & NP-complete (Theorem~\ref{thrm:npc-rewiring-rwalk}) \\
Monte Carlo & NP-complete (Theorem~\ref{thrm:npc-subos-mcarlo}) & NP-complete (Theorem~\ref{thrm:npc-rewiring-mcarlo}) \\
\hline
\end{tabular}
\caption{Summary of our computational complexity results. For different source detection algorithms, we consider the decision problem that the evader must solve in order hide optimally from the algorithm by either adding nodes or modifying edges. P = solvable in polynomial time; NP-complete = Non-deterministic Polynomial-time complete, implying that no known algorithm can solve it in polynomial time.}
\label{tab:decision-summary-theorems}
\end{table}

\begin{algorithm}[tbh!]
\caption{Finding an optimal solution for the Hiding Source by Adding Nodes problem given the Degree source detection algorithm.}
\label{alg:subos-degree}
\begin{algorithmic}[1]

\small

\Input{
Network $G=(V,E)$, evader $\vs \in V$, set of infected nodes $\II \subseteq V$, safety threshold $\thr \in \N$, budget $b \in \N$, set of evader-controlled nodes $\Subo$, set of supporters $\FS \subseteq V$.
}

\Output{
Solution $\Add$ to instance $(G,\vs,\II,\sddegr,\thr,b,\Subo,\FS)$ of the Hiding Source by Adding Nodes problem or $\perp$ if there is no solution.
}

\State $\goal \gets |N_{G^\II}(\vs)| + 1$
\State $\thr^* \gets \thr - |\{ v \in \II : |N_{G^\II}(v)| \geq g \}|$
\If {$\thr^* \leq 0$}
	\Return $\emptyset$
\EndIf

\State $S \gets \{ v \in \FS : |N_{G^\II}(v)| < \goal \}$

\For {$m \gets \max(0, \thr^*-|\Subo|), \ldots, \min(|S|,\thr^*)$}
	\State $S^* \gets \langle s_i \rangle_{i=1}^m$ such that $s_i$ is $i$-th node from $S$ in order of non-increasing $|N_{G_\II}(s_i)|$
	\If {$\left(m = 0 \lor \goal - |N_{G^\II}(s_m)| < |\Subo| \right) \land \left( m = \thr^* \lor |\FS|+|\Subo| > \goal \right)$}
		\State $\Add \gets \emptyset$
		\State $\Subo^* \gets \langle \delta_i \rangle_{i=1}^{\thr^*-m}$ such that $\forall_i \delta_i \in \Subo$ and $\forall_{i \neq j} \delta_i \neq \delta_j$
		\If {$m > 0$}
			\State $j^* \gets 1$
			\For {$s_i \in S^*$}
				\State $x_i \gets \goal - |N_{G^\II \cup \Add}(s_i)|$
				\If {$x_i < |\Subo^*|$}
					\While {$x_i > 0$}
						\State $\Add \gets \Add \cup \{(s_i,\delta_{j^*})\}$
						\State $j^* \gets \left( j^* \mod |\Subo^*| \right) + 1$
						\State $x_i \gets x_i - 1$
					\EndWhile
				\Else
					\State $\Add \gets \Add \cup \left( \{s_i\} \times \Subo^* \right) \cup select(x_i - |\Subo^*|, \Subo \setminus \Subo^*)$
				\EndIf
			\EndFor
		\EndIf
		\State $\delta^* \gets \argmax_{\delta_i \in \Subo^*} \left( \goal - |N_{G^\II \cup \Add}(\delta_i)| \right)$
		\If {$\goal - |N_{G^\II \cup \Add}(\delta^*)| \geq |\Subo^*|$}
			\For {$\delta_i \in \Subo^*$}
				\If {$\goal-|N_{G^\II \cup \Add}(\delta_i)| \geq |\Subo^*|$}
					\State $\Add \gets \Add \cup \left( \{\delta_i\} \times select \left( \goal - |N_{G^\II \cup \Add}(\delta_i)| - |\Subo^*| + 1, (\FS \cup \Subo) \setminus (\Subo^* \cup N_{G^\II \cup \Add}(\delta_i)) \right) \right)$
				\EndIf
			\EndFor
		\Else
			\If {$\sum_{\delta_i \in \Subo^*} \left( \goal - |N_{G^\II \cup \Add}(\delta_i)| \right) \mod 2 = 1$}
				\State $\Add \gets \Add \cup \left( \{\delta^*\} \times select(1, (\FS \cup \Subo) \setminus (\Subo^* \cup N_{G^\II \cup \Add}(\delta^*)))\right)$
			\EndIf
		\EndIf
		\State Connect nodes in $\Subo^*$ into a network such that degree of $\delta_i$ is $g - |N_{G^\II \cup \Add}(\delta_i)|$ using Havel-Hakimi algorithm
		\If {$m = 0$ and $G^\II \cup \Add$ disconnected}
				\State $\Add \gets \Add \cup \left( \{\delta_1\} \times select(1, \FS \setminus N_{G^\II \cup \Add})\right)$
		\EndIf
		\If {$|\Add| \leq b$}
			\State \Return $\Add$
		\EndIf
	
	\EndIf
\EndFor

\State \Return $\perp$

\end{algorithmic}
\end{algorithm}

\begin{theorem}
\label{thrm:npc-subos-degree}
The problem of Hiding Source by Adding Nodes is in P given the Degree source detection algorithm.
In particular, Algorithm~\ref{alg:subos-degree} finds a solution to the given instance $(G,\vs,\II,\sddegr,\thr,b,\Subo,\FS)$ of the problem.
\end{theorem}

\begin{proof}
We will analyze Algorithm~\ref{alg:subos-degree} and show that it finds a solution to the instance $(G,\vs,\II,\sddegr,\thr,b,\Subo,\FS)$ of the problem.

In order for a given set of edges $A$ to be a solution, we need to have at least $\thr$ infected nodes with degrees at least $\goal = |N_{G^\II}(\vs)|$ (the value computed in line~1).
Notice that some infected nodes might already have the required degree, so we only need to increase the degrees of $\thr^* = \thr - |\{ v \in \II : |N_{G^\II}(v)| \geq g \}|$ (the value computed in line~2).
If there are already at least $\thr$ infected nodes with degrees greater than $\vs$, the solution is the empty set, returned in line~3.

Notice that by adding a set of edges $\Add \subseteq (\Subo \times \Subo) \cup (\Subo \times (\FS \cup \{\vs\}))$ we can only increase degrees of nodes in $\Subo \cup \FS \cup \{\vs\}$.
Since it is never beneficial to increase the degree of $\vs$, a solution has to increase the degree of at least $\thr^*$ nodes in $\Subo \cup \FS$, that initially have lower degree than $\vs$, to at least $\goal$.
In line~4 we identify the set of nodes in $\FS$ that have degrees lower than $\vs$, and thus are candidates for satisfying the threshold $\thr^*$ (notice that we already counted the nodes in $\FS$ whose degrees are at least $\goal$ in line~2).

In lines~6-31 we will compute a smallest set of edges $\Add$ that needs to be added to $G$ so that the threshold $\thr^*$ is satisfied by $m$ nodes from $S$ and $\thr^*-m$ nodes from $\Subo$, for every potential value of $m$ (the loop in line~5).
Notice that if $\thr^* > |\Subo|$ we need to increase the degree of at least $\thr^*-|\Subo|$ nodes in $S$, otherwise it is possible to satisfy the threshold with just nodes in $\Subo$ (see the expression $\max(0, \thr^*-|\Subo|)$ in line~5).
Notice also that we never need to increase the degree of more than $\thr^*$ nodes in $S$ (see the expression $\min(|S|,\thr^*)$ in line~5).
If the said smallest set of edges $\Add$ for a given $m$ is within the evader's budget (the test performed in line~32), we return it in line~33.
Notice that if for every $m$ the size of such smallest $\Add$ is greater than the budget, then there is no solution to the problem (the value $\perp$ returned in line~34).

Since all edges added to nodes in $S$ (and increasing their degree) must connect them to nodes in $\Subo$, increasing the degrees of nodes that already have high degrees will result in the smallest possible size of $\Add$ (notice that if we were allowed to add edges between the nodes in $S$, we would need to take existing edges in $S \times S$ into consideration).
Hence, in line~6 we select the sequence $S^*$ of $m$ nodes from $S$ that will count towards satisfying the threshold $\thr^*$ as nodes with greatest degrees.
Notice that if $m=0$ then $S^*$ is empty.

There are two more conditions necessary for the existence of a solution $\Add$ for a given $m$ (both tested in line~7).
Every one node $s_i \in S^*$ contributing to satisfying the threshold $\thr^*$ needs to be connected with at least $\goal - |N_{G^\II}(s_i)|$ nodes from $\Subo$, and expression $\goal - |N_{G^\II}(s_m)| < |\Subo|$ in line~7 checks this condition for node $s_m$, which needs the greatest number of connections.
Notice that if $m=0$ then there is no need to check this condition.
The second condition is that every node from $\Subo$ contributing to satisfying the threshold $\thr^*$ needs to be connected with at least $\goal$ nodes from $\Subo \cup \FS$ (as initially its degree is $0$), and the expression $|\FS|+|\Subo| > \goal$ in line~7 checks this condition.
Notice that if $m=\thr^*$ then there is no need to check this condition.

In line~8 we initialize the solution $\Add$ (as we are now sure it exists), whereas in line~9 we select the sequence $\Subo^*$ of $\thr^*-m$ nodes from $\Subo$ that will count towards satisfying the threshold $\thr^*$.
Notice that if $m=\thr^*$ then $\Subo^*$ is empty.

In lines~10-20 we increase the degree of all nodes in $\S^*$ by connecting them primarily with nodes in $\Subo^*$ (as either way we need to increase their degrees and this way we obtain the smallest size of $\Add$), and then, if there are not enough nodes in $\Subo^*$, with other nodes from $\Subo$.
Let $y_i$ be the number of additional edges we need to connect to $\delta_i \in \Subo^*$ to increase its degree to $\goal$ after executing lines~10-20, i.e., $y_i = \goal - |N_{G^\II \cup \Add}(\delta_i)|$.
Notice that because of the way we distribute the connections with $S^*$ among the nodes in $\Subo^*$, we have that $\forall_{\delta_i, \delta_j \in \Subo^*} |y_i - y_j| \leq 1$ (see Figure~\ref{fig:subos-degree}).

\begin{figure}[tbh!]
\centering
\includegraphics[width=.8\linewidth]{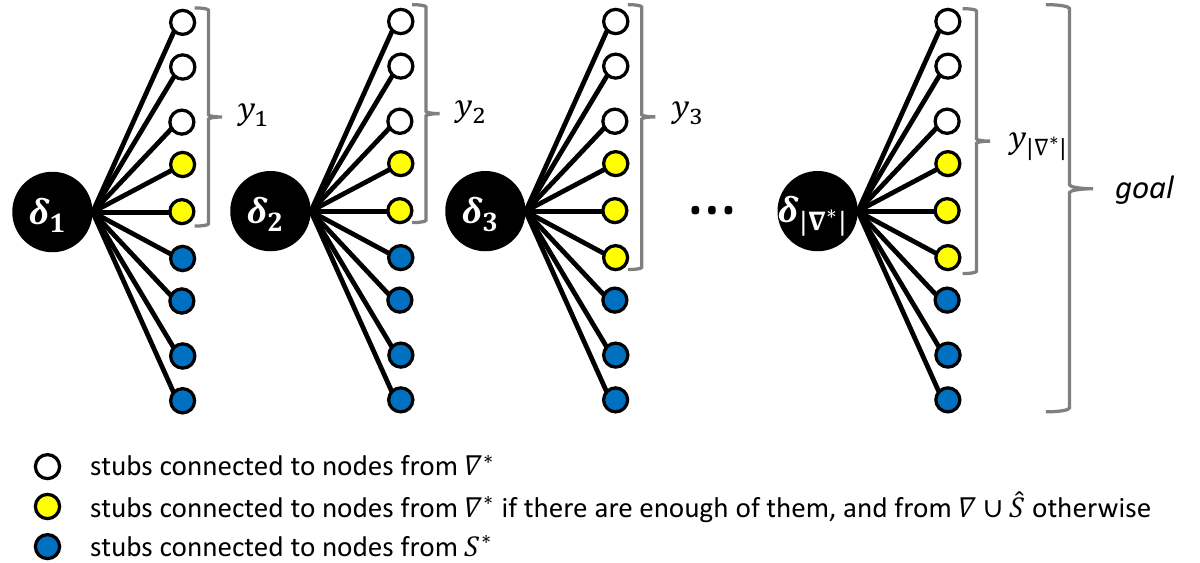}
\caption{
Distribution of stubs for nodes in $\Subo^*$ in the proof of Theorem~\ref{thrm:npc-subos-degree}.
}
\label{fig:subos-degree}
\end{figure}

Assume that for given $S^*$ and $\Subo^*$ there exists a minimal size solution such that $\exists_{\delta_i, \delta_j \in \Subo^*} |y_i - y_j| > 1$.
Without loss of generality, assume that $y_i > y_j + 1$, i.e., $\delta_j$ is connected with at least two more nodes from $S^*$ than $\delta_i$.
We can disconnect $\delta_j$ with any $v \in \FS$ among its neighbors not connected to $\delta_i$ (notice that as $y_i > y_j$, there must exist at least one such node), and instead connect $v$ with $\delta_i$.
We decreased the difference between $y_i$ and $y_j$ by one.
If by performing this operation we decreased the degree of $\delta_j$ below $\goal$, we should disconnect one of the neighbors of $\delta_i$ outside $S^*$ not connected to $\delta_j$ (again, from $y_i > y_j$ it follows that there exists at least one such node), and instead connect it to $\delta_j$, thus ensuring the correctness of the solution.
By repeating this operation we can decrease the maximal difference between any $y_i$ and any $y_j$ to $1$.
Hence, if there exists a minimal size solution such that $\exists_{\delta_i, \delta_j \in \Subo^*} y_i - y_j > 1$ then there also exists the same size solution such that $\forall_{\delta_i, \delta_j \in \Subo^*} |y_i - y_j| \leq 1$.

After increasing the degrees of all nodes in $S^*$ to at least $\goal$, we now need to increase the degrees of the nodes in $\Subo^*$.
Again, to obtain the smallest possible size of $\Add$, we will add as many edges as possible from $\Subo^* \times \Subo^*$, as opposed to between members of $\Subo^*$ and nodes from outside $\Subo^*$.
Let us denote the minimal number of necessary new connections among the members of $\Subo^*$ by $z$.
In line~21 we identify $\delta^*$ as the member of $\Subo^*$ that needs the greatest number of connections to be added to it (i.e., with the greatest value of $y_i$).
Notice that because of the way we constructed the set $\Add$ thus far, all other nodes in $\Subo^*$ need exactly as many new connections as $\delta^*$, or at most one less.
Hence, either all nodes in $\Subo^*$ need exactly $z$ new connections, or some of them need $z$ new connection, while others (including $\delta^*$) need $z+1$ new connections (all considered at the moment of executing line~21).

In line~22 we check whether the maximal number of required new edges is greater than the size of $\Subo^*$.
If that is the case, it is inevitable to connect some nodes in $\Subo^*$ with nodes from outside of $\Subo^*$, which we do in lines~23-25.
Notice that the choice of nodes from outside of $\Subo^*$ does not matter (as either way only one end of the edge will contribute to satisfying the threshold $\thr^*$), so we use the function $select(k,X)$ that selects $k$ elements from the set $X$.
Let us also assume that it prefers members of $\FS$.
Notice also that after this operation all nodes in $\Subo^*$ will need exactly $|\Subo^*|-1$ new connections.

Moreover, the sum of degrees in a network induced by $\Subo^*$ has to be even, hence in lines~27-28 we add additional edge with one end in $\Subo^*$ if necessary.
Notice that if we executed lines~23-25 then the sum of degrees is guaranteed to be even (as it is $|\Subo^*|(|\Subo^*|-1)$).
Notice also that since we add this edge to $\delta^*$, it is still true that either all nodes in $\Subo^*$ need $z$ more connections or some of them need $z$, while others need $z+1$ connections.

In line~29 we connect the nodes in $\Subo^*$ into a network that finally satisfies the threshold $\thr^*$.
We do it using the Havel-Hakimi algorithm~\cite{havel1955remark,hakimi1962realizability}, which connects a given set of nodes into a network with a given sequence of degrees if it is possible.

We will now show that this is indeed possible. We will do so using the Erd\H{o}s-Gallai theorem~\cite{erdos1960graphs}, which states that a given sequence $d_1 \geq \ldots \geq d_n$ can be realized as a network if an only if $\sum_{i=1}^n d_i$ is even and:
\begin{equation}
\label{eq:erdos-gallai}
\forall_{1 \leq k \leq n} \sum_{i=1}^k d_i \leq k(k-1) + \sum_{i=k+1}^{n} \min(d_i,k).
\end{equation}
As argued above, the sum of the number of new connections required by nodes in $\Subo^*$ to satisfy the threshold $\thr^*$ is even.
Let $n$ denote the size of $|\Subo^*|$, and let $m$ denote the number of nodes in $\Subo^*$ that require $z+1$ new connections (notice that $0 \leq m < n$).
The sequence of the degrees is such that $d_i=z+1$ if $i \leq m$ and $d_i=z$ otherwise.
We can assume that $z<n-1$, as otherwise all nodes in $\Subo^*$ need exactly $z=n-1$ new connections (as we executed lines~22-25 before), and the sequence of degrees can be realized by connecting nodes in $\Subo^*$ into a clique.
In what follows let $\mathcal{L}$ denote the left hand side of equation~\ref{eq:erdos-gallai}, and let $\mathcal{R}$ denote the right hand side of equation~\ref{eq:erdos-gallai}.
We will now show that Equation~\ref{eq:erdos-gallai} holds for nodes in $\Subo^*$, by performing calculations for four different cases:

\begin{itemize}
\item \textbf{Case I} $k \geq m \land z \geq k$:
\begin{align*}
\mathcal{L} &= (z+1)m + (k-m)z = kz+m \\
\mathcal{R} &= k(k-1)+(n-k)k = kn - k \\
\mathcal{R} - \mathcal{L} &= kn - k -kz -m \geq k(z+2) -2k - kz = 0;
\end{align*}

\item \textbf{Case II} $k \geq m \land z < k$:
\begin{align*} 
\mathcal{L} &= (z+1)m + (k-m)z = kz+m \\
\mathcal{R} &= k(k-1)+(n-k)z = k^2+nz-k-kz \\
\mathcal{R} - \mathcal{L} &= k^2+nz-k-2kz-m \geq k(z+2)+nz-2k-2kz = (n-k)z \geq 0;
\end{align*} 

\item \textbf{Case III} $k < m \land z \geq k$:
\begin{align*}
\mathcal{L} &= k(z+1) = kz+k \\
\mathcal{R} &= k(k-1)+(n-k)k = nk-k \\
\mathcal{R} - \mathcal{L} &= nk - kz - 2k = k(n-(z+2)) \geq 0;
\end{align*}

\item \textbf{Case IV} $k < m \land z < k$:
\begin{align*}
\mathcal{L} &= k(z+1) = kz+k \\
\mathcal{R} &= k(k-1)+(m-k)(z+1)+(n-m)z = k^2+m+nz-2k-kz \\
\mathcal{R} - \mathcal{L} &= k^2+m+nz-3k-2kz = k^2+m-k+nz-2k(z+1)+(z+1)^2-(z+1)^2 \\
&\geq (k-z-1)^2 +(m-k) + (z+2)z - (z+1)^2 = (k-z-1)^2 +(m-k) - 1 \geq 0 + 1 - 1 = 0.
\end{align*}
\end{itemize}

Finally, notice that if $m=0$ and so far we only added edges between the members of $\Subo^*$, we need to connect them to the rest of the network, which we do in lines~30-31.
Thanks to our assumption that the function $select$ prioritize nodes in $\FS$, if we added at least one edge between a member of $\Subo^*$ and a node from outside $\Subo^*$ then the network is already connected.
\end{proof}

\begin{theorem}
\label{thrm:npc-subos-closeness}
The problem of Hiding Source by Adding Nodes is NP-complete given the Closeness source detection algorithm.
\end{theorem}

\begin{proof}
The problem is trivially in NP, since after adding a given set of edges $\Add$, it is possible to compute the closeness centrality ranking of all nodes in $G^\II$ in polynomial time.

We will now prove that the problem is NP-hard.
To this end, we will show a reduction from the NP-complete \textit{Dominating Set} problem.
The decision version of this problem is defined by a network, $H=(V',E')$, where $V'=\{v_1,\ldots,v_n\}$, and a constant $k \in \N$, where the goal is to determine whether there exist $V^* \subseteq V'$ such that $|V^*|=k$ and every node outside $V^*$ has at least one neighbor in $V^*$, i.e., $\forall_{v \in V' \setminus V^*} N_H(v) \cap V^* \neq \emptyset$.

Let $(H,k)$ be a given instance of the Dominating Set problem.
We will now construct an instance of the Hiding Source by Adding Nodes problem.

First, let us construct a network $G=(V,E)$ where:
\begin{itemize}
\item $V = V' \cup \{ \vs, x, u, w, a_1, a_2, a_3 \} \cup \bigcup_{i=1}^{2n+k-1} \{y_i\}$,
\item $E = E' \cup \{(\vs,x), (x,u), (u,w), (w,a_1), (w,a_2), (w,a_3)\} \cup \bigcup_{i=1}^{2n+k-1} \{(u,y_i)\} \cup \bigcup_{i=1}^n \{(w,v_i)\}$.
\end{itemize}
An example of the construction of the network $G$ is presented in Figure~\ref{fig:npc-subos-closeness}.

\begin{figure}[tbh!]
\centering
\includegraphics[width=.9\linewidth]{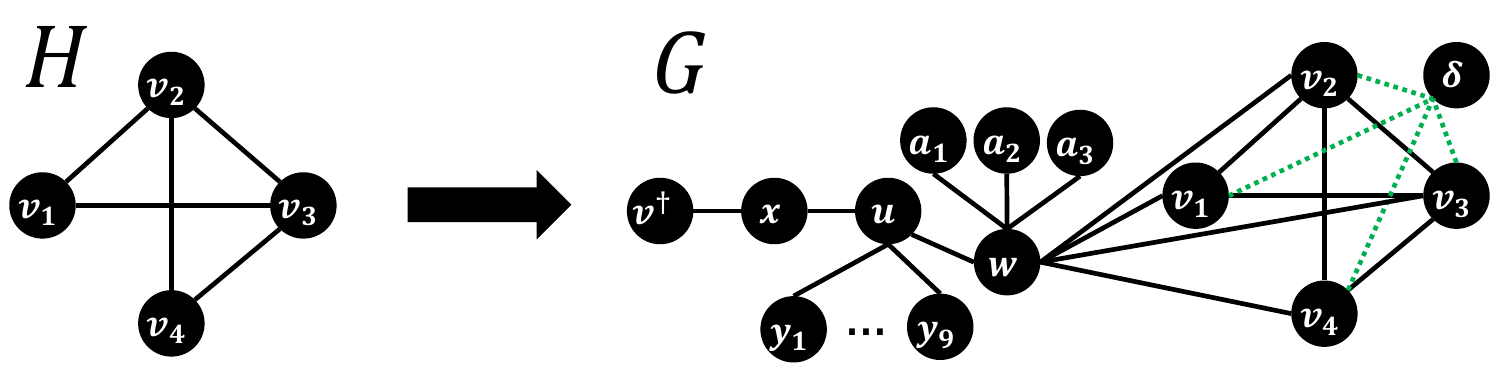}
\caption{
Construction of the network used in the proof of Theorem~\ref{thrm:npc-subos-closeness}.
Green dotted edges are allowed to be added.
}
\label{fig:npc-subos-closeness}
\end{figure}

Now, consider the instance $(G,\vs,\II,\sd,\thr,b,\Subo,\FS)$ of the Hiding Source by Adding Nodes problem, where:
\begin{itemize}
\item $G$ is the network we just constructed,
\item $\vs$ is the evader,
\item $\II=V$,
\item $\sd$ is the Closeness source detection algorithm,
\item $\thr=3n+k+6$ is the safety threshold,
\item $b=k$, where $k$ is the size of the dominating set from the Dominating Set problem instance,
\item $\Subo=\{\delta\}$,
\item $\FS = V'$, i.e., additional node can only be connected with the nodes in $V'$.
\end{itemize}

First, let us analyze the closeness centrality values of the nodes in $G$ after the addition of any $A \subseteq \Subo \times \FS$.
Let $D_v$ denote the sum of distances from $v$ to all nodes in the network, i.e., $D_v = \sum_{w \in V} d(v,w)$.
Notice that $\sdclos(v,G,\II) = \frac{1}{\sum_{w \in \II} D_v}$, which implies that a greater value of $D_v$ leads to a lower position of the ranking of nodes according to the Closeness source detection algorithm.
Moreover, let $z_A$ denote the sum of distance between $\delta$ and members of $V'$ after the addition of $A$.
Table~\ref{tab:npc-subos-closeness} presents the computation of $D_v$ for every node $v \in V \cup \Subo$ after the addition of a given $A \subseteq \Subo \times \FS$.

\begin{table}[tbh]
\small
\centering
\begin{tabular}{ l c c c c c c c c c }
$v$ & $d(v,\vs)$ & $d(v,x)$ & $d(v,u)$ & $d(v,w)$ & $\sum_j d(v,y_j)$ & $\sum_j d(v,a_j)$ & $\sum_j d(v,v_j)$ & $d(v,\delta)$ & $D_v$ \\
\hline
$\vs$ & $0$ & $1$ & $2$ & $3$ & $3(2n+k-1)$ & $12$ & $4n$ & $5$ & $10n+3k+20$ \\
$x$ & $1$ & $0$ & $1$ & $2$ & $2(2n+k-1)$ & $9$ & $3n$ & $4$ & $7n+2k+15$ \\
$u$ & $2$ & $1$ & $0$ & $1$ & $2n+k-1$ & $6$ & $2n$ & $3$ & $4n+k+12$ \\
$w$ & $3$ & $2$ & $1$ & $0$ & $2(2n+k-1)$ & $3$ & $n$ & $2$ & $5n+2k+9$ \\
$y_i$ & $3$ & $2$ & $1$ & $2$ & $2(2n+k-2)$ & $9$ & $3n$ & $4$ & $7n+2k+17$ \\
$a_i$ & $4$ & $3$ & $2$ & $1$ & $3(2n+k-1)$ & $4$ & $2n$ & $3$ & $8n+3k+14$ \\
$v_i$ & $4$ & $3$ & $2$ & $1$ & $3(2n+k-1)$ & $6$ & $\leq 2(n-1)$ & $\leq 3$ & $\leq 8n+3k+14$ \\
$\delta$ & $5$ & $4$ & $3$ & $2$ & $4(2n+k-1)$ & $9$ & $z_A$ & $0$ & $8n+4k+19+z_A$ \\
\hline
\end{tabular}
\caption{Sums of distances between nodes of the network after the addition of $A \subseteq \Subo \times \FS$, used in the proof of Theorem~\ref{thrm:npc-subos-closeness}.}
\label{tab:npc-subos-closeness}
\end{table}

Since the safety threshold is $\thr=3n+k+6$, all other nodes (including $\delta$) must have greater closeness centrality than $\vs$ after the addition of a given $A$ in order for the said $A$ to be a solution to the constructed instance of the problem of Hiding Source by Adding Nodes.
Notice that after adding any $A$ to the network we have $D_v < D_{\vs}$ for any $v \in V \setminus \{\vs,\delta\}$ (based on the formulas for $D_v$ in Table~\ref{tab:npc-subos-closeness}).
Hence, a given $A$ is a solution to the constructed instance of the Hiding Source by Adding Nodes problem if and only if we have $D_\delta < D_{\vs}$ after the addition of $A$.

Let us now analyze the value of $D_\delta$ after the addition of a given $A$.
We have that:
$$
D_\delta = 8n+4k+19+z_A = 8n+4k+19+|A|+2|V'_A|+3(n-|A|-|V'_A|) = 11n+4k+19-2|A|-|V'_A|
$$
where $V'_A = \{ v_i \in V' \setminus N(\delta) : N(v_i) \cap N(\delta) \neq \emptyset\}$.
Notice we have that $|V'_A| \leq n-|A|$, which gives us:
$$
D_\delta \geq 11n+4k+19-2|A|-(n-|A|) = 10n+4k+19-|A|.
$$
Hence, given that $D_{\vs}=10n+3k+20$, we have that $D_\delta < D_{\vs}$ if and only if $|A|=k$ and |$|V'_A| = n-k$, i.e., $\delta$ is connected with $k$ nodes in $V'$ and every other node in $V'$ has a neighbor who is connected with $\delta$.

We will now show that the constructed instance of the Hiding Source by Adding Nodes problem has a solution if and only if the given instance of the Dominating Set problem has a solution.

Assume that there exists a solution to the given instance of the Dominating Set problem, i.e., a subset $V^* \subseteq V'$ of size $k$ such that all other nodes have a neighbor in $V^*$.
After adding to $G$ the set $\Add = \{\delta\} \times V^*$ we have that $|\Add|=k$ and every node in $V' \setminus N(\delta)$ has a neighbor who is connected with $\delta$.
We showed that if there exists a solution to the given instance of the Dominating Set problem, then there also exists a solution to the constructed instance of the Hiding Source by Adding Nodes problem.

Assume that there exists a solution $\Add$ to the constructed instance of the Hiding Source by Modifying Edges problem.
As shown above, we must have $|\Add|=k$ and every node in $V' \setminus N(\delta)$ has a neighbor who is connected with $\delta$.
Therefore $\{v_i \in V' : (\delta,v_i) \in \Add\}$ is a dominating set in $H$ of size exactly $k$.
We showed that if there exists a solution to the constructed instance of the Hiding Source by Adding Nodes problem, then there also exists a solution to the given instance of the Dominating Set problem.

This concludes the proof.
\end{proof}

\begin{theorem}
\label{thrm:npc-subos-betweenness}
The problem of Hiding source by Adding Nodes is NP-complete given the Betweenness source detection algorithm.
\end{theorem}

\begin{proof}
The problem is trivially in NP, since after adding a given set of edges $\Add$, it is possible to compute the betweenness centrality ranking of all nodes in $G^\II$ in polynomial time.

We will now prove that the problem is NP-hard.
To this end, we will show a reduction from the NP-complete \textit{Finding $k$-Clique} problem.
The decision version of this problem is defined by a network, $H=(V',E')$, where $V'=\{v_1,\ldots,v_n\}$, and a constant $k \in \N$, where the goal is to determine whether there exist $k$ nodes forming a clique in $H$.

Let $(H,k)$ be a given instance of the Finding $k$-Clique problem.
Let us assume that $k \geq 3$, all other instances can be easily solved in polynomial time.
We will now construct an instance of the Hiding Source by Modifying Edges problem.

First, let us construct a network $G=(V,E)$ where:
\begin{itemize}
\item $V = V' \cup \{ \vs, \delta, u, w \} \cup \bigcup_{v_i,v_j \in V' : (v_i,v_j) \notin E'} \{\bar{e}_{i,j}\} \cup \bigcup_{i=1}^{k} \{x_i\} \cup \bigcup_{i=1}^{k^3} \{y_i\}$,
\item $E = \{w\} \times (V \setminus \{\delta\}) \cup \bigcup_{\bar{e}_{i,j} \in V} \{(\bar{e}_{i,j},v_i),(\bar{e}_{i,j},v_j)\} \cup \bigcup_{y_i, y_j \in V} \{(y_i,y_j)\} \cup \bigcup_{i=1}^{k} \{(\vs,x_i)\} \cup \{(u,x_1), (u,x_2)\}$.
\end{itemize}

In what follows we denote the set of nodes $x_1,\ldots,x_k$ by $X$, and the set of nodes $y_1,\ldots,y_{k^3}$ by $Y$.
Notice that a node $\bar{e}_{i,j}$ exists in $V$ if and only if $v_i,v_j$ are not connected in $H$.
An example of the construction of the network $G$ is presented in Figure~\ref{fig:npc-subos-betweenness}.

\begin{figure}[tbh!]
\centering
\includegraphics[width=.9\linewidth]{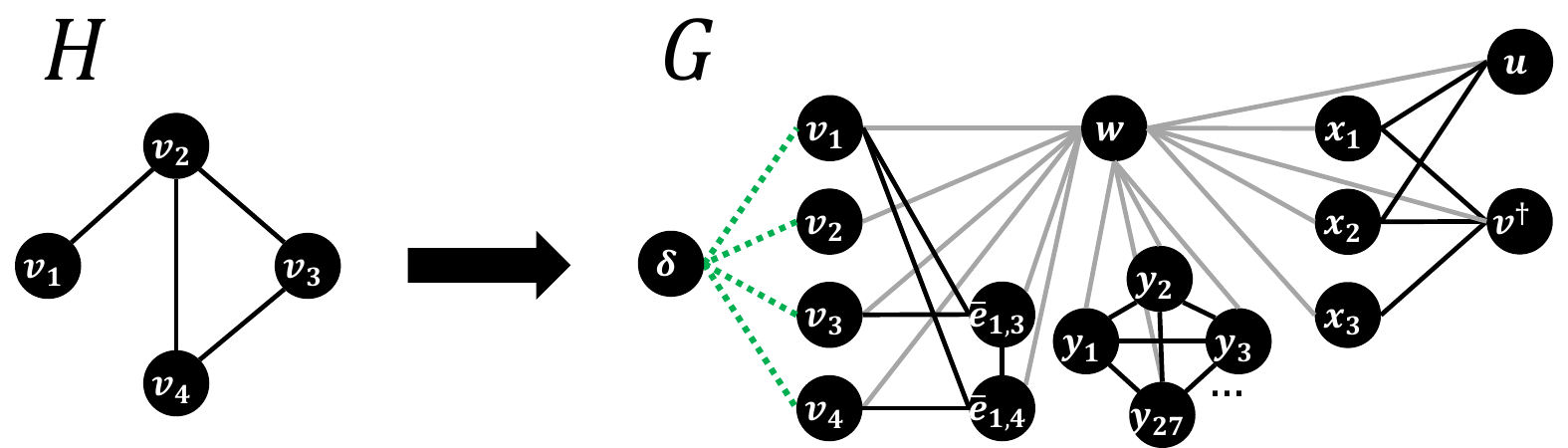}
\caption{
Construction of the network used in the proof of Theorem~\ref{thrm:npc-subos-betweenness}.
Green dotted edges are allowed to be added.
Edges incident with $w$ are printed grey for better readability.
}
\label{fig:npc-subos-betweenness}
\end{figure}

Now, consider the instance $(G,\vs,\II,\sd,\thr,b,\Subo,\FS)$ of the Hiding Source by Adding Nodes problem, where:
\begin{itemize}
\item $G$ is the network we just constructed,
\item $\vs$ is the evader,
\item $\II=V$,
\item $\sd$ is the Betweenness source detection algorithm,
\item $\thr=k+2$ is the safety threshold,
\item $b=k$ is the budget of the evader,
\item $\Subo=\{\delta\}$,
\item $\FS = V'$, i.e., additional node can only be connected with the nodes in $V'$.
\end{itemize}

Let $C_v$ denote the value of $\sdbetw(v,G,\II)$, i.e., the value used to compute determining the source of diffusion.
To remind the reader:
$$
C_v = \sum_{v' \neq v'' : v',v'' \in \II \setminus \{v\}} \frac{|\{\pi \in \Pi(v',v''): v \in \pi\}|}{|\Pi(v',v'')|}
$$
where $\Pi(v',v'')$ is the set of shortest paths between the nodes $v'$ and $v''$.

We will now make the following observations about the values of $C_v$ in $G$ after adding an arbitrary $A \subseteq \{\delta\} \times \FS$:
\begin{itemize}
\item $C_{\vs} = \frac{k(k-1)}{4}-\frac{1}{6}$, as it controls one of three shortest paths between $x_1$ and $x_2$ (the others being controlled by $w$ and $u$), and one of two shortest paths between all other pair of nodes in $X$ (the other being controlled by $w$),
\item $C_u = \frac{1}{3}$, as it controls one of three shortest paths between $x_1$ and $x_2$ (the others being controlled by $w$ and $\vs$),
\item $C_{x_i} = 0$, as it does not control any shortest paths,
\item $C_{a_i} = 0$, as it does not control any shortest paths,
\item $C_w \geq k^3(k+n+4)$, as it controls all shortest paths between nodes in $Y$ and all other nodes,
\item $C_{\bar{e}_{i,j}} \leq \frac{1}{2}$, as it controls one of at least two shortest paths between $v_i$ and $v_j$ (the other being controlled by $w$),
\item $C_{v_i}=0$ if $v_i$ is not connected with $\delta$, as it does not control any shortest paths,
\item $C_{v_i} \geq \frac{k^3}{|A|}$ if $v_i$ is connected with $\delta$, as it controls one of $|A|$ paths between $\delta$ and nodes in $Y$ (the others being controlled by other nodes in $V'$ connected to $\delta$),
\item $C_{\delta}=z_A\frac{1}{3}+\left(\frac{|A|(|A|-1)}{2}-z_A\right)\frac{1}{2}$, where $z_A$ is the number of pairs $v_i,v_j \in V'$ connected with $\delta$ such that $\bar{e}_{i,j} \in E$, as $\delta$ controls one of three shortest paths between such pairs of $v_i,v_j \in V'$ (the others being controlled by $w$ and $\bar{e}_{i,j}$), while $\delta$ controls one of two shortest paths between all other pairs of $v_i,v_j \in V'$ that it is connected two (the other being controlled by $w$).
\end{itemize}

Notice that we have $C_{\vs} < k^2$, and since we assumed that $k \geq 3$ we also have $C_{\vs} \geq 1\frac{1}{3}$.
Hence, the only nodes that can have greater value of $C_v$ (and higher position in the source detection algorithm ranking) are $w$, $\delta$ and nodes in $V'$ connected to $\delta$.
Since the safety threshold is $\thr=k+2$ and the evader's budget is $b=k$, it implies that $\delta$ must be connected with exactly $k$ nodes in $V'$.
Notice also that $w$ and nodes in $V'$ connected to $\delta$ have greater betweenness centrality than $\vs$ no matter the choice of $A$.
Hence, a given set $A$ is a solution to the constructed instance of the Hiding Source by Adding Nodes problem if and only if $\delta$ is connected with exactly $k$ nodes from $V'$ and $\delta$ has greater betweenness centrality than $\vs$.

Let us now analyze the betweenness centrality of $\delta$ when it is connected with $k$ nodes in $V'$ (in which case $|A|=k$):
$$
C_{\delta}=z_A\frac{1}{3}+\left(\frac{|A|(|A|-1)}{2}-z_A\right)\frac{1}{2}=\frac{k(k-1)}{4}-\frac{z_A}{6},
$$
where $z_A$ is the number of pairs $v_i,v_j \in V'$ connected with $\delta$ such that $\bar{e}_{i,j} \in E$.
Notice that if $z_A=0$ then $C_{\delta}=\frac{k(k-1)}{4} > C_{\vs}$.
However, if $z_A \geq 1$ then $C_{\delta} \leq \frac{k(k-1)}{4} - \frac{1}{6} = C_{\vs}$.
Notice that $z_A=0$ only when the nodes connected with $\delta$ form a clique in $H$ (as node $\bar{e}_{i,j} \in E$ is added to $V$ only when nodes $v_i$ and $v_j$ are not neighbors in $H$).
Hence, $\delta$ has greater betweenness centrality than $\vs$ if and only if nodes in $V'$ connected with $\delta$ form a clique in $H$.

We will now show that the constructed instance of the Hiding Source by Adding Nodes problem has a solution if and only if the given instance of the Finding $k$-Clique problem has a solution.

Assume that there exists a solution to the given instance of the Finding $k$-Clique problem, i.e., a subset $V^* \subseteq V'$ forming a $k$-clique in $H$.
Notice that for $\Add = \{\delta\} \times V^*$ we have $|\Add|=k$ and $z_{\Add}=0$.
We showed that if there exists a solution to the given instance of the Finding $k$-Clique problem, then there also exists a solution to the constructed instance of the Hiding Source by Adding Nodes problem.

Assume that there exists a solution $\Add$ to the constructed instance of the Hiding Source by Modifying Edges problem.
As observed above, we must have $|\Add|=k$ and the nodes that $\delta$ is connected with, i.e., nodes $V^* = \{ v_i \in V':(\delta,v_i) \in \Add\}$, must form a clique in $H$.
We showed that if there exists a solution to the constructed instance of the Hiding Source by Adding Nodes problem, then there also exists a solution to the given instance of the Finding $k$-Clique problem.

This concludes the proof.
\end{proof}

\begin{theorem}
\label{thrm:npc-subos-rumor}
The problem of Hiding source by Adding Nodes is NP-complete given the Rumor source detection algorithm.
\end{theorem}

\begin{proof}
The problem is trivially in NP, since after adding a given set of edges $\Add$, it is possible to compute the rumor centrality ranking of all nodes in $G^\II$ in polynomial time.

We will now prove that the problem is NP-hard.
To this end, we will show a reduction from the NP-complete \textit{Exact 3-Set Cover} problem.
The decision version of this problem is defined by a universe, $U=\{u_1,\ldots,u_{3k}\}$, and a collection of sets $S=\{S_1,\ldots,S_{|S|}\}$ such that $\forall_i S_i \subset U$ and $\forall_i |S_i|=3$, where the goal is to determine whether there exist $k$ elements of $S$ the union of which equals $U$.

Let $(U,S)$ be a given instance of the Exact 3-Set Cover problem.
Assume that $k \geq 17$, all other instances can be easily solved in polynomial time.
We will now construct an instance of the Hiding Source by Adding Nodes problem.

First, let us construct a network $G=(V,E)$ where:
\begin{itemize}
\item $V = \{ \vs, w, x, a_1, a_2 \} \cup U \cup S \cup \bigcup_{i=1}^{|S|} \{y_i\} \cup \bigcup_{i=1}^{k} \{Q_i\} \cup \bigcup_{i=1}^{3k} \{z_i\}$,
\item $E = \{(w,x),(w,\vs),(x,a_1),(x,a_2)\} \cup \left(\{w\} \times (Y \cup S \cup Q \cup Z)\right) \cup (Z \times U) \cup \bigcup_{u_i \in S_j} \{(u_i,S_j)\} \cup \bigcup_{u_i} \{(u_i,Q_{\ceil{\frac{i}{3}}})\}$.
\end{itemize}
We denote the set of nodes $\{a_1,a_2\}$ by $A$, the set of nodes $\{y_1,\ldots,y_{|S|}\}$ by $Y$, the set of nodes $\{z_1,\ldots,z_{3k}\}$ by $Z$, and the set of nodes $\{Q_1,\ldots,Q_{k}\}$ by $Q$.
An example of the construction of the network $G$ is presented in Figure~\ref{fig:npc-subos-rumor}.

\begin{figure}[tbh!]
\centering
\includegraphics[width=.9\linewidth]{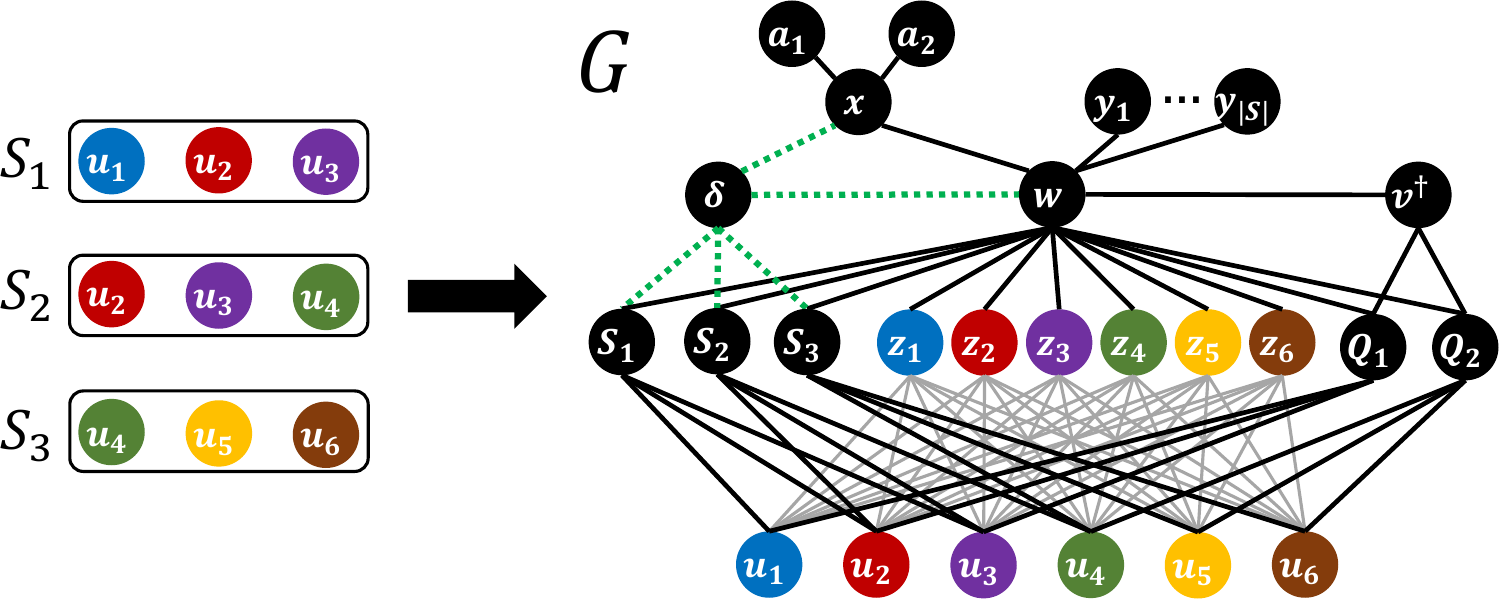}
\caption{
Construction of the network used in the proof of Theorem~\ref{thrm:npc-subos-rumor}.
Green dotted edges are allowed to be added.
}
\label{fig:npc-subos-rumor}
\end{figure}

Now, consider the instance $(G,\vs,\II,\sd,\thr,b,\Subo,\FS)$ of the Hiding Source by Adding Nodes problem, where:
\begin{itemize}
\item $G$ is the network we just constructed,
\item $\vs$ is the evader,
\item $\II=V$, i.e., all nodes in $G$ are infected,
\item $\sd$ is the Rumor source detection algorithm,
\item $\thr=1$,
\item $b=k+2$,
\item $\Subo = \{\delta\}$,
\item $\FS = \{w,x\} \cup S$.
\end{itemize}

To remind the reader, the score assigned to a given node by the Rumor source detection algorithm is $\sdrumor(v,G,\II) = \frac{|\II|!}{\prod_{w \in \II} \Theta^v_w}$ where $\Theta^v_w$ is the size of the subtree of $w$ in the BFS tree of $G^\II$ rooted at $v$.
Let $C_v(A)$ denote $\prod_{w \in \II \setminus \{v\}} \Theta^v_w$ in $G$ after the addition of $A$.
Notice that greater $C_v(A)$ implies lower $\sdrumor(v, G \cup A, \II)$ and vice versa, as we have $\sdrumor(v, G \cup A, \II) = \frac{|\II|!}{|V| C_v(A)}$.

Notice that a BFS tree of a given node can be constructed in many different ways, which makes the Rumor source detection algorithm nondeterministic.

First, let us compute the value of $C_{\vs}(A)$.
Figure~\ref{fig:npc-subos-rumor-vs} presents the BFS tree of $\vs$.
The location of $\vs$ in the tree depends on the connections included in $A$, the three cases are:
\begin{itemize}
\item if $(\delta,w) \in A$ then $\delta$ is in location 1,
\item if $(\delta,w) \notin A \land (\delta,x) \in A$ then $\delta$ is either in location 2 or in location 3,
\item if $(\delta,w) \notin A \land (\delta,x) \notin A$ then $\delta$ is in location 3.
\end{itemize}
The location of node $\delta$ in the BFS tree of $\vs$ determines the value of $C_{\vs}(A)$ as follows:
\begin{itemize}
\item if $\delta$ is in location 1 then $C_{\vs}(A) = 4^k 3 (2|S|+3k+5)$,
\item if $\delta$ is in location 2 then $C_{\vs}(A) = 4^k 4 (2|S|+3k+5)$,
\item if $\delta$ is in location 3 then $C_{\vs}(A) = 4^k 6 (2|S|+3k+5)$.
\end{itemize}

\begin{figure}[tbh!]
\centering
\begin{minipage}{.48\textwidth}
	\centering
	\includegraphics[width=.95\linewidth]{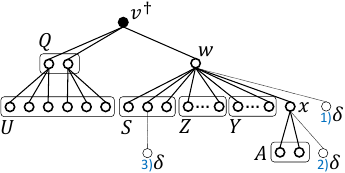}
	\caption{BFS tree of node $\vs$, with three possible locations of node $\delta$.}
	\label{fig:npc-subos-rumor-vs}
\end{minipage}
\hfill
\begin{minipage}{.48\textwidth}
	\centering
	\includegraphics[width=\linewidth]{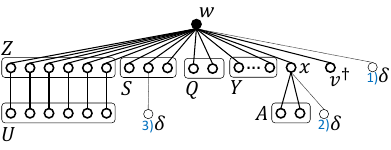}
	\caption{BFS tree of node $w$, with three possible locations of node $\delta$.}
	\label{fig:npc-subos-rumor-w}
\end{minipage}
\end{figure}

Hence, we have that $C_{\vs}(A) \leq 4^k 6 (2|S|+3k+5)$.
We will now show that for every node $v \in V \setminus \{\vs,\delta\}$ and every set of edges $A$ that can be added to $G$, there exists a BFS tree of $v$ such that $C_v(A) \geq 4^k 6 (2|S|+3k+5)$.
Since the Hiding Source by Adding Nodes problem requires the safety threshold to be maintained in every realization of the source detection algorithm, none of such $v$ can contribute to the safety threshold.

\begin{itemize}
\item $C_w(A) \geq C_{\vs}(A)$:
Figure~\ref{fig:npc-subos-rumor-w} presents a possible BFS tree of $w$ with three potential locations of node $\delta$.
We can observe that the value of $C_w(A)$ is minimal when $\delta$ is connected with $w$ (location 1 in Figure~\ref{fig:npc-subos-rumor-w}), which gives us:
$$
C_w(A) \geq 2^{3k} 3.
$$
Since we know that $C_{\vs}(A) \leq 4^k 6 (2|S|+3k+5)$, we have that:
$$
\frac{C_w(A)}{C_{\vs}(A)} \geq \frac{2^{3k} 3}{4^k 6 (2|S|+3k+5)} = \frac{2^k}{4|S|+6k+10}.
$$
Notice that we can assume that the set $S$ contains at most one copy of each 3-element subset of $U$, as a solution to the given instance of the Exact 3-Set Cover problem never contains two instance of the same subset (in fact, since $U$ has $3k$ elements, there is never any overlap between two elements of a solution). Hence, we can assume that $|S| \leq \frac{(3k)!}{(3k-3)!3!} = \frac{3k(3k-1)(3k-2)}{6} < \frac{9}{2}k^3$, which gives us:
$$
\frac{C_w(A)}{C_{\vs}(A)} > \frac{2^k}{18k^3+6k+10}.
$$
Now, notice that for $k \geq 17$, the function $2^k-18k^3-6k-10$ is increasing (its derivative is $2^k \ln(2)-54k^2-6$) and its value for $k=17$ is greater than zero.
Hence, we have that $2^k > 18k^3+6k+10$, which gives us:
$$
\frac{C_w(A)}{C_{\vs}(A)} > 1.
$$

\item $C_x(A) \geq C_{\vs}(A)$:
consider a BFS tree of $x$ obtained by rooting the tree presented in Figure~\ref{fig:npc-subos-rumor-w} in $x$ instead of in $w$.
We can observe that the value of $C_x(A)$ is minimal when $\delta$ is connected with $x$ (location 2 in Figure~\ref{fig:npc-subos-rumor-w}), which gives us:
$$
C_x(A) \geq 2^{3k}(2|S|+7k+2).
$$
Therefore, we have that:
$$
\frac{C_x(A)}{C_{\vs}(A)} \geq \frac{2^{3k}(2|S|+7k+2)}{4^k 6 (2|S|+3k+5)} = \frac{2^k(2|S|+7k+2)}{6 (2|S|+3k+5)} > 1.
$$

\item $C_{a_i}(A) \geq C_{\vs}(A)$:
consider a BFS tree of $a_i$ obtained by rooting the tree presented in Figure~\ref{fig:npc-subos-rumor-w} in $a_i$ instead of in $w$.
We can observe that the value of $C_{a_i}(A)$ is minimal when $\delta$ is connected with $x$ (location 2 in Figure~\ref{fig:npc-subos-rumor-w}), which gives us:
$$
C_{a_i}(A) \geq 2^{3k}(2|S|+7k+5)(2|S|+7k+2).
$$
Therefore, we have that:
$$
\frac{C_{a_i}(A)}{C_{\vs}(A)} \geq \frac{2^{3k}(2|S|+7k+5)(2|S|+7k+2)}{4^k 6 (2|S|+3k+5)} = \frac{2^k(2|S|+7k+5)(2|S|+7k+2)}{6 (2|S|+3k+5)} > 1.
$$

\item $C_{y_i}(A) \geq C_{\vs}(A)$:
consider the BFS tree of $y_i$ obtained by rooting the tree presented in Figure~\ref{fig:npc-subos-rumor-w} in $y_i$ instead of in $w$.
We can observe that the value of $C_{y_i}(A)$ is minimal when $\delta$ is connected with $w$ (location 1 in Figure~\ref{fig:npc-subos-rumor-w}), which gives us:
$$
C_{y_i}(A) \geq 2^{3k}3(2|S|+7k+5).
$$
Therefore, we have that:
$$
\frac{C_{y_i}(A)}{C_{\vs}(A)} \geq \frac{2^{3k}3(2|S|+7k+5)}{4^k 6 (2|S|+3k+5)} = \frac{2^{k-1}(2|S|+7k+5)}{2|S|+3k+5} > 1.
$$

\begin{figure}[tbh!]
\centering
\begin{minipage}{.48\textwidth}
	\centering
	\includegraphics[width=.9\linewidth]{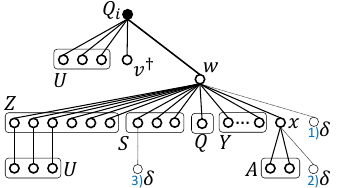}
	\caption{BFS tree of node $Q_i$, with three possible locations of node $\delta$.}
	\label{fig:npc-subos-rumor-q}
\end{minipage}
\hfill
\begin{minipage}{.48\textwidth}
	\centering
	\includegraphics[width=.9\linewidth]{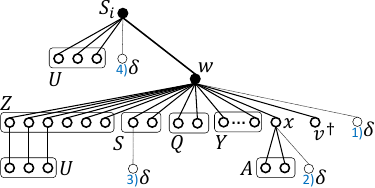}
	\caption{BFS tree of node $S_i$, with four possible locations of node $\delta$.}
	\label{fig:npc-subos-rumor-s}
\end{minipage}
\end{figure}

\item $C_{Q_i}(A) \geq C_{\vs}(A)$:
Figure~\ref{fig:npc-subos-rumor-q} presents a possible BFS tree of $Q_i$ with three potential locations of node $\delta$.
We can observe that the value of $C_{Q_i}(A)$ is minimal when $\delta$ is connected with $w$ (location 1 in Figure~\ref{fig:npc-subos-rumor-q}), which gives us:
$$
C_{Q_i}(A) \geq 2^{3k-3}3(2|S|+7k+1).
$$
Therefore, given the assumption that $k \geq 17$, we have that:
$$
\frac{C_{Q_i}(A)}{C_{\vs}(A)} \geq \frac{2^{3k-3}3(2|S|+7k+1)}{4^k 6 (2|S|+3k+5)} = \frac{2^{k-4}(2|S|+7k+1)}{2|S|+3k+5} > 1.
$$

\item $C_{S_i}(A) \geq C_{\vs}(A)$:
Figure~\ref{fig:npc-subos-rumor-s} presents a possible BFS tree of $S_i$ with four potential locations of node $\delta$.
We can observe that the value of $C_{S_i}(A)$ is minimal when $\delta$ is connected with $S_i$ (location 4 in Figure~\ref{fig:npc-subos-rumor-s}), which gives us:
$$
C_{S_i}(A) \geq 2^{3k-3}3(2|S|+7k+1).
$$
Therefore, given assumption that $k \geq 17$, we have that:
$$
\frac{C_{S_i}(A)}{C_{\vs}(A)} \geq \frac{2^{3k-3}3(2|S|+7k+1)}{4^k 6 (2|S|+3k+5)} = \frac{2^{k-4}(2|S|+7k+1)}{2|S|+3k+5} > 1.
$$

\begin{figure}[tbh!]
\centering
\begin{minipage}{.48\textwidth}
	\centering
	\includegraphics[width=.95\linewidth]{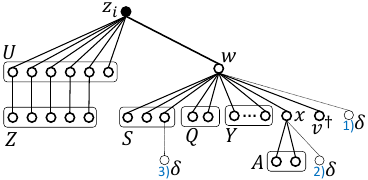}
	\caption{BFS tree of node $z_i$, with three possible locations of node $\delta$.}
	\label{fig:npc-subos-rumor-z}
\end{minipage}
\hfill
\begin{minipage}{.48\textwidth}
	\centering
	\includegraphics[width=.62\linewidth]{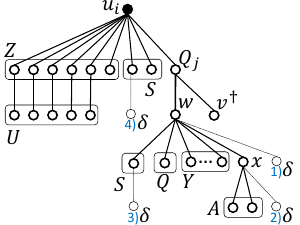}
	\caption{BFS tree of node $u_i$, with four possible locations of node $\delta$.}
	\label{fig:npc-subos-rumor-u}
\end{minipage}
\end{figure}

\item $C_{z_i}(A) \geq C_{\vs}(A)$:
Figure~\ref{fig:npc-subos-rumor-z} presents a possible BFS tree of $z_i$ with four potential locations of node $\delta$.
We can observe that the value of $C_{z_i}(A)$ is minimal when $\delta$ is connected with $w$ (location 1 in Figure~\ref{fig:npc-subos-rumor-z}), which gives us:
$$
C_{z_i}(A) \geq 2^{3k-1}3(2|S|+k+6).
$$
Therefore, we have that:
$$
\frac{C_{z_i}(A)}{C_{\vs}(A)} \geq \frac{2^{3k-1}3(2|S|+k+6)}{4^k 6 (2|S|+3k+5)} = \frac{2^{k-2}(2|S|+k+6)}{2|S|+3k+5} > 1.
$$

\item $C_{u_i}(A) \geq C_{\vs}(A)$:
Figure~\ref{fig:npc-subos-rumor-u} presents a possible BFS tree of $u_i$ with four potential locations of node $\delta$.
We can observe that the value of $C_{u_i}(A)$ is minimal when $\delta$ is connected with $w$ (location 1 in Figure~\ref{fig:npc-subos-rumor-u}), which gives us:
$$
C_{u_i}(A) \geq 2^{3k-1}3(|S|+k+6)(|S|+k+4).
$$
Therefore, we have that:
$$
\frac{C_{u_i}(A)}{C_{\vs}(A)} \geq \frac{2^{3k-1}3(|S|+k+6)(|S|+k+4)}{4^k 6 (2|S|+3k+5)} = \frac{2^{k-2}(|S|+k+6)(|S|+k+4)}{2|S|+3k+5} > 1.
$$

\end{itemize}

We showed that $\delta$ is the only node that can contribute to satisfying the safety threshold.
Therefore, a given $A$ is a solution to the constructed instance of the Hiding Source by Adding Nodes problem if and only if we have that $C_{\vs}(A) > C_\delta(A)$.
We will now show that $C_{\vs}(A) > C_\delta(A)$ for a given $A$ if and only if $(\delta,w) \in A$, and $(\delta,x) \in A$, and for every $u_i \in U$ there exists $S_j \in S$ such that $u_i \in S_j$ and $(\delta,S_j) \in A$.

\begin{figure}[tbh!]
\centering
\begin{minipage}{.48\textwidth}
	\centering
	\includegraphics[width=\linewidth]{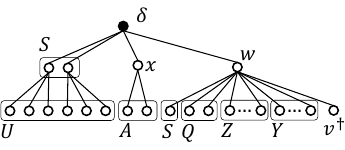}
	\caption{BFS tree of node $\delta$, when it is connected with $w$, $x$, and a cover of $U$.}
	\label{fig:npc-subos-rumor-d-optimal}
\end{minipage}
\hfill
\begin{minipage}{.48\textwidth}
	\centering
	\includegraphics[width=.85\linewidth]{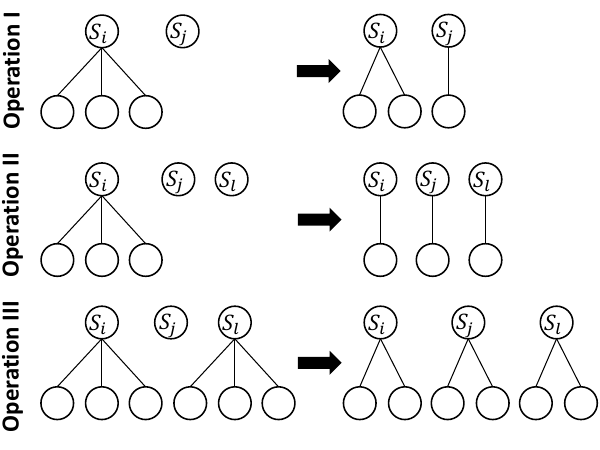}
	\caption{Operations used in the proof of Lemma~\ref{lem:npc-subos-rumor}.}
	\label{fig:npc-subos-rumor-d-operations}
\end{minipage}
\end{figure}

First, we will show that if $(\delta,w) \in A$ and $(\delta,x) \in A$ and $\forall_{u_i \in U} \exists_{S_j \in S} \left( u_i \in S_j \land (\delta,S_j) \in A \right)$, then for every BFS tree of $\delta$ we have $C_{\vs}(A) > C_\delta(A)$.
Figure~\ref{fig:npc-subos-rumor-d-optimal} presents the BFS tree of $\delta$ in this case (notice that in this particular case, this is the only possible structure of the BFS tree).
We have that:
$$
C_\delta(A) = 4^k3(2|S|+3k+2).
$$
At the same time, we have that:
$$
C_{\vs}(A) = 4^k 3 (2|S|+3k+5).
$$
Therefore, we have that $C_{\vs}(A) > C_\delta(A)$.

Before we move on, we will prove a useful lemma.

\begin{lemma}
\label{lem:npc-subos-rumor}
Assume that in a BFS tree of $\delta$ all nodes from $U$ are leaves and are the only children of the nodes in $S$.
The minimum over all possible values of $\prod_{S_i \in S} \Theta^\delta_{S_i}$ is $4^k$.
The second lowest possible value is $4^{k-1}6$.
\end{lemma}

\begin{proof}
In what follows, let $\Theta_S$ denote $\prod_{S_i \in S} \Theta^\delta_{S_i}$.
To remind the reader, there are $3k$ nodes in $U$ and every node $S_i$ is connected with exactly $3$ of them.
The value of $\Theta^\delta_{S_i}$ is therefore between $1$ (if node $S_i$ has no children in the BFS tree) and $4$ (if all three nodes from $U$ connected with $S_i$ are its children in the BFS tree).
Hence, we have that $\prod_{S_i \in S} \Theta^\delta_{S_i} = 4^a 3^b 2^c$ such that $3a+2b+c=3k$.

Assume that $k$ nodes in $S_i$ have three children each.
The value of $\prod_{S_i \in S} \Theta^\delta_{S_i}$ is then $4^k$.
Notice that any possible value $4^a 3^b 2^c$ of $\prod_{S_i \in S} \Theta^\delta_{S_i}$ can be achieved starting from $4^k$ by performing the following operations (presented in Figure~\ref{fig:npc-subos-rumor-d-operations}) in any order:
\begin{itemize}
\item \textbf{operation I:} moving one child of a node in $S$ with three children to another node in $S$ with no children, repeated $\min(b,c)$ times,
\item \textbf{operation II:} moving two children of a node in $S$ with three children to two other nodes in $S$ with no children, repeated $\frac{c-b}{3}$ times if $c > b$, and not executed at all otherwise,
\item \textbf{operation III:} moving one child each of two nodes in $S$ with three children to another node with no children, repeated $\frac{b-c}{3}$ times if $b > c$, and not executed at all otherwise.
\end{itemize}

Let $\Theta_0$ be the value of $\prod_{S_i \in S} \Theta^\delta_{S_i}$ before performing a given operation.
Notice that each of the operations increases the value of $\prod_{S_i \in S} \Theta^\delta_{S_i}$:
\begin{itemize}
\item the value of $\prod_{S_i \in S} \Theta^\delta_{S_i}$ after \textbf{operation I} is $\frac{3}{2}\Theta_0$,
\item the value of $\prod_{S_i \in S} \Theta^\delta_{S_i}$ after \textbf{operation II} is $2\Theta_0$,
\item the value of $\prod_{S_i \in S} \Theta^\delta_{S_i}$ after \textbf{operation III} is $\frac{27}{16}\Theta_0$.
\end{itemize}

Hence, since every possible value of $4^a 3^b 2^c$ can be achieved via performing a sequence of operations starting with $4^k$, and every operation increases the value, then $4^k$ is the minimal possible value.
Moreover, since operation I increases the value the least, the second minimal value is $\frac{3}{2} 4^k = 4^{k-1}6$.
\end{proof}

\begin{figure}[tbh!]
\centering
\begin{minipage}{.48\textwidth}
	\centering
	\includegraphics[width=.7\linewidth]{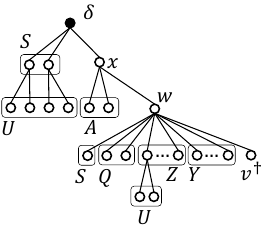}
	\caption{BFS tree of node $\delta$ for Case I.}
	\label{fig:npc-subos-rumor-d-case1}
\end{minipage}
\hfill
\begin{minipage}{.48\textwidth}
	\centering
	\includegraphics[width=.7\linewidth]{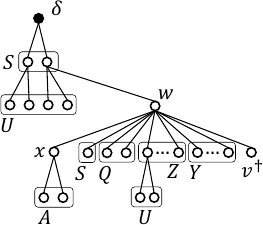}
	\caption{BFS tree of node $\delta$ for Case II.}
	\label{fig:npc-subos-rumor-d-case2}
\end{minipage}
\end{figure}

Now, we will show that if $(\delta,w) \notin A$ or $(\delta,x) \notin A$ or $\exists_{u_i \in U} \forall_{S_j \in S : u_i \in S_j} (\delta,S_j) \notin A$, then there exists a BFS tree of $\delta$ such that $C_\delta(A) \geq C_{\vs}(A)$.
We show the proof for each of the cases separately:
\begin{itemize}

\item \textbf{Case I} $(\delta,w) \notin A \land (\delta,x) \in A$:
Figure~\ref{fig:npc-subos-rumor-d-case1} presents the structure of the BFS tree of $\delta$ for this case.
Notice that, given Lemma~\ref{lem:npc-subos-rumor}, the value of $C_\delta(A)$ is minimal when $k$ nodes from $S$ connected with $\delta$ cover the entire universe (we then have $\prod_{S_i \in S} \Theta^\delta_{S_i}=4^k$).
Notice that, if $\delta$ would be connected with more than $k$ nodes from $S$ (which is possible, since the budget of the evader is $k+2$), then at least one element of $U$ would be a neighbor of two nodes from $S$ connected with $\delta$, and, based on Lemma~\ref{lem:npc-subos-rumor}, we would be able to choose a BFS tree such that $\prod_{S_i \in S} \Theta^\delta_{S_i} \geq 4^{k-1}6$.
We have:
$$
C_\delta(A) \geq 4^k (2|S|+3k+2)(2|S|+3k+5),
$$
as well as:
$$
C_{\vs}(A) \leq 4^k 6(2|S|+3k+5),
$$
which gives us:
$$
\frac{C_\delta(A)}{C_{\vs}(A)} \geq \frac{4^k (2|S|+3k+2)(2|S|+3k+5)}{4^k 6(2|S|+3k+5)} = \frac{(2|S|+3k+2)(2|S|+3k+5)}{6(2|S|+3k+5)} > 1.
$$

\item \textbf{Case II} $(\delta,w) \notin A \land (\delta,x) \notin A$:
Figure~\ref{fig:npc-subos-rumor-d-case2} presents the structure of the BFS tree of $\delta$ for this case.
Notice that, given Lemma~\ref{lem:npc-subos-rumor}, the value of $C_\delta(A)$ is minimal when $k$ nodes from $S$ connected with $\delta$ cover the entire universe (we then have $\prod_{S_i \in S} \Theta^\delta_{S_i}=4^k$).
Notice that, if $\delta$ would be connected with more than $k$ nodes from $S$ (which is possible, since the budget of the evader is $k+2$), then at least one element of $U$ would be a neighbor of two nodes from $S$ connected with $\delta$, and, based on Lemma~\ref{lem:npc-subos-rumor}, we would be able to choose a BFS tree such that $\prod_{S_i \in S} \Theta^\delta_{S_i} \geq 4^{k-1}6$.
We have:
$$
C_\delta(A) \geq 4^{k-1} 3(2|S|+3k+5)(2|S|+3k+8),
$$
as well as:
$$
C_{\vs}(A) = 4^k 6(2|S|+3k+5),
$$
which gives us:
$$
\frac{C_\delta(A)}{C_{\vs}(A)} \geq \frac{4^{k-1} 3(2|S|+3k+5)(2|S|+3k+8)}{4^k 6(2|S|+3k+5)} = \frac{(2|S|+3k+5)(2|S|+3k+8)}{8(2|S|+3k+5)} > 1.
$$

\begin{figure}[tbh!]
\centering
\begin{minipage}{.48\textwidth}
	\centering
	\includegraphics[width=.95\linewidth]{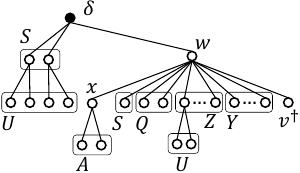}
	\caption{BFS tree of node $\delta$ for Case III.}
	\label{fig:npc-subos-rumor-d-case3}
\end{minipage}
\hfill
\begin{minipage}{.48\textwidth}
	\centering
	\includegraphics[width=.95\linewidth]{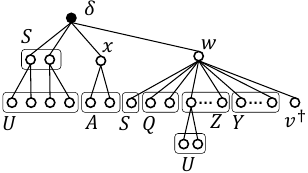}
	\caption{BFS tree of node $\delta$ for Case IV.}
	\label{fig:npc-subos-rumor-d-case4}
\end{minipage}
\end{figure}

\item \textbf{Case III} $(\delta,w) \in A \land (\delta,x) \notin A$:
Figure~\ref{fig:npc-subos-rumor-d-case3} presents the structure of the BFS tree of $\delta$ for this case.
Notice that, given Lemma~\ref{lem:npc-subos-rumor}, the value of $C_\delta(A)$ is minimal when $k$ nodes from $S$ connected with $\delta$ cover the entire universe (we then have $\prod_{S_i \in S} \Theta^\delta_{S_i}=4^k$).
Notice that, if $\delta$ would be connected with more than $k$ nodes from $S$ (which is possible, since the budget of the evader is $k+2$), then at least one element of $U$ would be a neighbor of two nodes from $S$ connected with $\delta$, and, based on Lemma~\ref{lem:npc-subos-rumor}, we would be able to choose a BFS tree such that $\prod_{S_i \in S} \Theta^\delta_{S_i} \geq 4^{k-1}6$.
We have:
$$
C_\delta(A) \geq 4^k 3(2|S|+3k+5),
$$
as well as:
$$
C_{\vs}(A) = 4^k 3(2|S|+3k+5),
$$
which gives us:
$$
\frac{C_\delta(A)}{C_{\vs}(A)} \geq \frac{4^k 3(2|S|+3k+5)}{4^k 3(2|S|+3k+5)} = 1.
$$

\item \textbf{Case IV} $(\delta,w) \in A \land (\delta,x) \in A \land \exists_{u_i \in U} \forall_{S_j \in S : u_i \in S_j} (\delta,S_j) \notin A$:
Figure~\ref{fig:npc-subos-rumor-d-case4} presents the structure of the BFS tree of $\delta$ for this case.
Since nodes in $S$ connected with $\delta$ do not cover the entire universe, there is at least one node from $U$ connected $z_1$ in the BFS tree.
We have:
$$
C_\delta(A) \geq 4^{k-1}18(2|S|+3k+3),
$$
as well as:
$$
C_{\vs}(A) = 4^k 3(2|S|+3k+5),
$$
which gives us:
$$
\frac{C_\delta(A)}{C_{\vs}(A)} \geq \frac{4^{k-1}18(2|S|+3k+3)}{4^k 3(2|S|+3k+5)} = \frac{3(2|S|+3k+3)}{2(2|S|+3k+5)} > 1.
$$

\end{itemize}

We showed that a given $A$ is a solution to the constructed instance of the Hiding Source by Adding Nodes if and only if $(\delta,w) \in A$, and $(\delta,x) \in A$, and for every $u_i \in U$ there exists $S_j \in S$ such that $u_i \in S_j$ and $(\delta,S_j) \in A$.
Finally, we are ready to prove the theorem.

Assume that there exists a solution to the given instance of the Exact 3-Set Cover problem, i.e., a subset $S* \subseteq S$ such that $\bigcup_{S_i \in S*} S_i = U$.
Then the set $\{\delta\} \times (\{w,x\} \cup S^*)$ is a solution to the constructed instance of the Hiding Source by Adding Nodes problem.

Similarly, assume that there exists a solution $\Add$ to the constructed instance of the Hiding Source by Adding Nodes problem.
Hence, for every $u_i \in U$ there exists $S_j \in S$ such that $u_i \in S_j$ and $(\delta,S_j) \in \Add$.
We also have that $\delta$ can be connected with at most $k$ nodes from $S$, as it has to also be connected with $w$ and $x$, and the budget of the evader is $k+2$.
Therefore, the set $\{S_i \in S : (\delta,S_i) \in \Add\}$ is a solution to the given instance of the Exact 3-Set Cover problem.

This concludes the proof.
\end{proof}

\begin{theorem}
\label{thrm:npc-subos-rwalk}
The problem of Hiding source by Adding Nodes is NP-complete given the Random Walk source detection algorithm.
\end{theorem}

\begin{proof}
The problem is trivially in NP, since after adding a given set of edges $\Add$, it is possible to compute the Random Walk source detection algorithm scores of all nodes in $G^\II$ in polynomial time.

We will now prove that the problem is NP-hard.
To this end, we will show a reduction from the NP-complete \textit{$3$-Set Cover} problem.
The decision version of this problem is defined by a universe, $U=\{u_1,\ldots,u_{|U|}\}$, and a collection of sets $S=\{S_1,\ldots,S_{|S|}\}$ such that $\forall_i S_i \subset U$ and $\forall_i |S_i|=3$, where the goal is to determine whether there exist $k$ elements of $S$ the union of which equals $U$.

Let $(U,S)$ be a given instance of the $3$-Set Cover problem.
We will now construct an instance of the Hiding Source by Adding Nodes problem.

First, let us construct a network $G=(V,E)$ where:
\begin{itemize}
\item $V = \{ \vs, w, x, a \} \cup S \cup U$,
\item $E = \{(\vs,w), (\vs,a), (w,x)\} \cup \left( \{\vs\} \times S \right) \cup \bigcup_{S_i \in V} \bigcup_{u_j \in S_i} \{(S_i,u_j)\}$.
\end{itemize}
In what follows we will denote the set of nodes $S_1,\ldots, \S_{|S|}$ by $S$, and we will denote the set of nodes $u_1,\ldots, u_{|U|}$ by $U$.
Notice that having $u_j \in S_i$ in the last union in the formula of $E$ means that we connect a given node $S_i \in S$ only with nodes $u_j \in U$ corresponding to the elements contained in $S_i$ in the given instance of the $3$-Set Cover problem.
An example of the construction of the network $G$ is presented in Figure~\ref{fig:npc-subos-rwalk}.

\begin{figure}[tbh!]
\centering
\includegraphics[width=.7\linewidth]{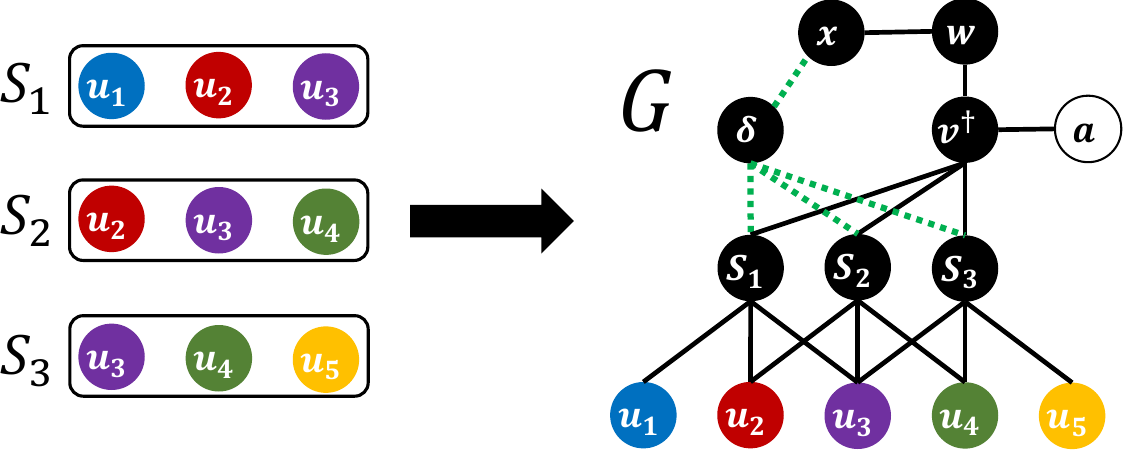}
\caption{
Construction of the network used in the proof of Theorem~\ref{thrm:npc-subos-rwalk}.
Green dotted edges are allowed to be added.
Infected nodes in $G$ are highlighted in black.
}
\label{fig:npc-subos-rwalk}
\end{figure}

Now, consider the instance $(G,\vs,\II,\sd,\thr,b,\Subo,\FS)$ of the Hiding Source by Adding Nodes problem, where:
\begin{itemize}
\item $G$ is the network we just constructed,
\item $\vs$ is the evader,
\item $\II=V \setminus \{a\}$, i.e., all nodes in $G$ other than $a$ are infected,
\item $\sd$ is the Random Walk source detection algorithm,
\item $\thr=|S|+|U|+3$ is the safety threshold,
\item $b=k+1$,
\item $\Subo = \{\delta\}$,
\item $\FS = \{x\} \cup S$.
\end{itemize}

Moreover, the Random Walk source detection algorithm needs to be parameterized with the model of spreading.
Let the algorithm be parameterized with the model in which the probability of propagation is $p=\frac{1}{2}$ and the time of diffusion is $T=3$.

Notice that since the safety threshold is $\thr=|S|+|U|+3$, all nodes other than $a$ (which is not infected) must have greater source detection algorithm scores than $\vs$.
In particular, after adding any non-empty $A$ to the network, the score of $\vs$ is greater than zero (since all infected nodes are then within distance $T=3$ from $\vs$).
Therefore, if $A$ is a solution to the constructed instance of the Hiding Source by Adding Nodes problem, then all nodes in $U$ have to be within distance $T=3$ from the node $x$ (otherwise their scores are all zero).
This is the case if and only if $\delta$ is connected with $x$ and for every node $u_i$ there exists a node $S_j$ such that $u_i \in S_j$ and $S_j$ is connected with $\delta$.

We will now show the implication in the other direction, i.e., if for a given $A$ we have that $\delta$ is connected with $x$ and for every node $u_i$ there exists a node $S_j$ such that $u_i \in S_j$ and $S_j$ is connected with $\delta$, then $A$ is a solution to the constructed instance of the Hiding Source by Adding Nodes problem.

\begin{table}[tbh]
\centering
\begin{tabular}{ l c c c c }
$v$ & $\phi_3(v)$ & $\phi_2(v)$ & $\phi_1(v)$ & $\phi_0(v) = \sdrwalk(v,G \cup A,\II \cup \{\delta\})$ \\
\hline
$\vs$ & $1$ & $1-\frac{p}{|S|+2}$ & $1-\frac{p(2-p)}{|S|+2}$ & $1-\frac{p(80+40|S|+51p^2+25p^2|S| -120p-60p|S|-|A|p^2)}{20(|S|+2)^2}$ \\
$\delta$ & $1$ & $1$ & $1$ & $1-\frac{|A|-1}{|A|}\frac{p^3}{5(|S|+2)}$ \\
$w$ & $1$ & $1$ & $1-\frac{p^2}{2(|S|+2)}$ & $1-\frac{p^2(3-2p)}{2(|S|+2)}$ \\
$x$ & $1$ & $1$ & $1$ & $1-\frac{p^3}{4(|S|+2)}$ \\
$u_i$ & $1$ & $1$ & $1$ & $\geq 1 - \frac{p^3}{4(|S|+2)}$ \\
$S_i \in N(\delta)$ & $1$ & $1$ & $1-\frac{p^2}{5(|S|+2)}$ & $1-\frac{p^2(3-2p)}{5(|S|+2)}$ \\
$S_i \notin N(\delta)$ & $1$ & $1$ & $1-\frac{p^2}{4(|S|+2)}$ & $1-\frac{p^2(3-2p)}{4(|S|+2)}$ \\
\hline
\end{tabular}
\caption{The values of $\phi_t$ after the addition of a given $A$ such that $(\delta,x) \in A$, used in the proof of Theorem~\ref{thrm:npc-subos-rwalk}.}
\label{tab:npc-subos-rwalk}
\end{table}

Table~\ref{tab:npc-subos-rwalk} presents the values of $\phi_t$ for all infected nodes in the network and for different values of $t$, after we add to the network a given $A$ such that $(\delta,x) \in A$.
We will now show that for every node $v$ other $\vs$ we have $\phi_0(v) > \phi_0(\vs)$.
Given that $p=\frac{1}{2}$, we have:

\begin{itemize}
\item $\phi_0(\vs) = 1-\frac{131+65|S|-|A|}{160(|S|+2)^2} < 1-\frac{120+60|S|}{160(|S|+2)^2} = 1-\frac{3}{8(|S|+2)}$;
\item $\phi_0(\delta) = 1-\frac{|A|-1}{|A|}\frac{1}{40(|S|+2)} > 1-\frac{1}{40(|S|+2)} > \phi_0(\vs)$;
\item $\phi_0(w) = 1-\frac{1}{4(|S|+2)} > \phi_0(\vs)$;
\item $\phi_0(x) = 1-\frac{1}{32(|S|+2)} > \phi_0(\vs)$;
\item $\phi_0(u_i) \geq 1 - \frac{1}{32(|S|+2)} > \phi_0(\vs)$;
\item if $S_i \in N(\delta)$ then $\phi_0(S_i) = 1-\frac{1}{10(|S|+2)} > \phi_0(\vs)$;
\item if $S_i \notin N(\delta)$ then $\phi_0(S_i) = 1-\frac{1}{8(|S|+2)} > \phi_0(\vs)$.
\end{itemize}

We showed that a given $A$ is a solution to the constructed instance of the Hiding Source by Adding Nodes problem if and only if $\delta$ is connected with $x$ and for every node $u_i$ there exists a node $S_j$ such that $u_i \in S_j$ and $S_j$ is connected with $\delta$.
We will now show that the constructed instance of the Hiding Source by Adding Nodes problem has a solution if and only if the given instance of the $3$-Set Cover problem has a solution.

Assume that there exists a solution to the given instance of the $3$-Set Cover problem, i.e., a subset $S^* \subseteq S$ of size $k$ the union of which is the universe $U$.
In that case the set $\{\delta\} \times (\{x\} \cup S^*)$ is a solution to the constructed instance of the Hiding Source by Adding Nodes problem.

Assume that there exists a solution $\Add$ to the constructed instance of the Hiding Source by Adding Nodes problem.
In that case $S^* = \{S_i \in S : (\delta,S_i) \in \Add\}$ is a a solution to the given instance of the $3$-Set Cover problem.

This concludes the proof.
\end{proof}

\begin{theorem}
\label{thrm:npc-subos-mcarlo}
The problem of Hiding Source by Adding Nodes is NP-complete given the Monte Carlo source detection algorithm.
\end{theorem}

\begin{proof}
The problem is trivially in NP, since after adding a given set of edges $\Add$, it is possible to generate Monte Carlo samples and compute the ranking of all nodes in $G^\II$ in polynomial time.

We will now prove that the problem is NP-hard.
To this end, we will show a reduction from the NP-complete \textit{Dominating Set} problem.
The decision version of this problem is defined by a network, $H=(V',E')$, where $V'=\{v_1,\ldots,v_n\}$, and a constant $k \in \N$, where the goal is to determine whether there exist $V^* \subseteq V'$ such that $|V^*|=k$ and every node outside $V^*$ has at least one neighbor in $V^*$, i.e., $\forall_{v \in V' \setminus V^*} N_H(v) \cap V^* \neq \emptyset$.

Let $(H,k)$ be a given instance of the Dominating Set problem.
Let us assume that there is no solution of size one, i.e., $\forall_{v_i \in V'}d_H(v_i) < n-1$.
This can be easily checked in polynomial time.
We will now construct an instance of the Hiding Source by Adding Nodes problem.

First, let us construct a network $G=(V,E)$ where:
\begin{itemize}
\item $V = V' \cup \{ \vs, u, w, x \} \cup \bigcup_{i=1}^{n-1} \{a_i\}$,
\item $E = E' \cup \{(\vs,u), (u,w), (w,x)\} \cup \bigcup_{i=1}^{n-1} \{(a_i,x)\}$.
\end{itemize}
An example of the construction of the network $G$ is presented in Figure~\ref{fig:npc-subos-mcarlo}.

\begin{figure}[tbh!]
\centering
\includegraphics[width=.8\linewidth]{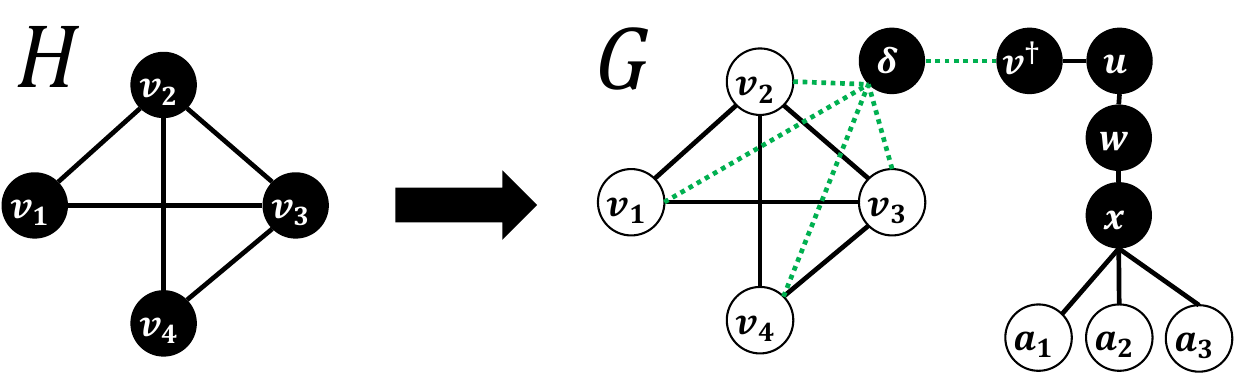}
\caption{
Construction of the network used in the proof of Theorem~\ref{thrm:npc-subos-mcarlo}.
Green dotted edges are allowed to be added.
Infected nodes in $G$ are highlighted in black.
}
\label{fig:npc-subos-mcarlo}
\end{figure}

Now, consider the instance $(G,\vs,\II,\sd,\thr,b,\Subo,\FS)$ of the Hiding Source by Adding Nodes problem, where:
\begin{itemize}
\item $G$ is the network we just constructed,
\item $\vs$ is the evader,
\item $\II=\{\vs,u,w,x\}$,
\item $\sd$ is the Monte Carlo source detection algorithm,
\item $\thr=2$ is the safety threshold,
\item $b=k+1$, where $k$ is the size of the dominating set from the Dominating Set problem instance,
\item $\Subo = \{\delta\}$,
\item $\FS = V' \cup \{\vs\}$.
\end{itemize}

Moreover, the Monte Carlo source detection algorithm needs to be parameterized with the model of spreading.
Let the algorithm be parameterized with the model in which the probability of propagation is $p=1$ and the time of diffusion is $T=3$.
Notice that since the model with $p=1$ is deterministic, all Monte Carlo samples starting from a given node always give the same result, hence the formula of the Monte Carlo source detection algorithm is in this case:
$$
\sdmcarlo(v,G,\II) = \exp\left( \frac{-\left( \frac{|\II \cap \II_v|}{|\II \cup \II_v|}- 1 \right)^2}{a^2} \right)
$$
where $\II_v$ is the set of infected nodes in the Monte Carlo sample starting the  diffusion at $v$.
Moreover, notice that a position of a given node $v \in \II$ in the ranking generated by the Monte Carlo source detection algorithm depends solely on the value of $C_v= \frac{|\II \cap \II_v|}{|\II \cup \II_v|}$ (in particular, the node with the greatest value of $C_v$ is selected as the source of diffusion).

Notice that the problem of Hiding Source by Adding Nodes requires the network induced by the infected nodes to be connected, hence $\delta$ must be connected with $\vs$, and it may be connected with at most $k$ nodes from $V'$.
Let $V'_{A,y}$ denote the set of nodes from $V'$ in distance at most $y$ from $\delta$ after the addition of $A$, i.e., $V'_{A,y} = \{v \in V': d_{(V \cup \Subo, E \cup A)}(\vs,\delta) \leq y\}$.
Notice that for any $y$ we have $|V'_{A,y}| \leq n$.
Let us compute the values of $C_v$ for the nodes in $\II$ (notice that these are the only nodes that can be selected as the source of diffusion by the Monte Carlo algorithm) after the addition of an arbitrary $A$:
\begin{itemize}
\item $C_{\vs} = \frac{5}{|V'_{A,2}|+5}$,
\item $C_\delta = \frac{4}{|V'_{A,3}|+5} \leq \frac{4}{|V'_{A,2}|+5} < \frac{5}{|V'_{A,2}|+5} = C_{\vs}$,
\item $C_u = \frac{5}{n+|V'_{A,1}|+4}$,
\item $C_w = \frac{5}{n+4}$,
\item $C_x = \frac{4}{n+4} < \frac{5}{n+5} \leq \frac{5}{|V'_{A,2}|+5} = C_{\vs}$ (as $|V'_{A,2}| \leq n$).
\end{itemize}

Since the safety threshold is $\thr=2$, both $u$ and $w$ need to have a greater value of $C_v$ than $\vs$.
It is easy to see that $C_w > C_{\vs}$ if and only if $|V'_{A,2}|=n$.
Notice also that if $|V'_{A,2}|=n$ then also $|V'_{A,1}| > 1$ (and, consequently, $C_u = \frac{5}{n+|V'_{A,1}|+4} < \frac{5}{n+5} = C_{\vs}$), as we assumed that there is no node with degree $n-1$ in $H$ (and subnetwork induced by $V'$ in $G$ is the same as $H$).
Therefore, the safety threshold is met if and only if $|V'_{A,2}|=n$.

We will now show that the constructed instance of the Hiding Source by Modifying Edges problem has a solution if and only if the given instance of the Dominating Set problem has a solution.

Assume that there exists a solution to the given instance of the Dominating Set problem, i.e., a subset $V^* \subseteq V'$ of size $k$ such that all other nodes have a neighbor in $V^*$.
After adding to $G$ set $\Add = \{(\delta,\vs)\} \cup \{\delta\} \times V^*$ we have that $|\Add|=k+1$ and $|V'_{\Add,2}|=n$, which implies that the safety threshold is met.
We showed that if there exists a solution to the given instance of the Dominating Set problem, then there also exists a solution to the constructed instance of the Hiding Source by Adding Nodes problem.

Assume that there exists a solution $\Add$ to the constructed instance of the Hiding Source by Modifying Edges problem.
As shown above, we must have $|V'_{\Add,2}|=n$, i.e., every node in $V'$ is either connected with $\delta$ or has a neighbor who is connected with $\delta$.
Moreover, $\{(\delta,\vs)\}$ must be a part of $\Add$, and, since the budget of the evader is $b=k+1$, there are at most $k$ edges connecting $\delta$ with the nodes in $V'$.
Therefore $V^*=\{v \in V' : (\delta,v) \in \Add\}$ is a dominating set in $H$ of size at most $k$ (we can add $k-|\Add|$ arbitrarily chosen elements to obtain set of size exactly $k$).
We showed that if there exists a solution to the constructed instance of the Hiding Source by Adding Nodes problem, then there also exists a solution to the given instance of the Dominating Set problem.

This concludes the proof.
\end{proof}

\begin{theorem}
\label{thrm:npc-rewiring-degree}
The problem of Hiding Source by Modifying Edges is NP-complete given the Degree source detection algorithm.
\end{theorem}

\begin{proof}
The problem is trivially in NP, since after adding and removing the given sets of edges $\Add$ and $\Rem$, it is possible to compute the degree centrality ranking of all nodes in $G^\II$ in polynomial time.

We will now prove that the problem is NP-hard.
To this end, we will show a reduction from the NP-complete \textit{Finding $k$-Clique} problem.
The decision version of this problem is defined by a network, $H=(V',E')$, where $V'=\{v_1,\ldots,v_n\}$, and a constant $k \in \N$, where the goal is to determine whether there exist $k$ nodes forming a clique in $H$.

Let $(H,k)$ be a given instance of the Finding $k$-Clique problem.
Let us assume that $n=2$, i.e., network $H$ has at least two nodes.
We will now construct an instance of the Hiding Source by Modifying Edges problem.

First, let us construct a network $G=(V,E)$ where:
\begin{itemize}
\item $V = V' \cup \{ \vs \} \cup \bigcup_{i=1}^{n} \bigcup_{j=1}^{n-k+1} \{a_{i,j}\}$,
\item $E = \bigcup_{v_i \in V} \{(\vs,v_i)\} \cup \bigcup_{v_i \in V} \bigcup_{a_{i,j} \in V} \{(v_i,a_{i,j})\}$.
\end{itemize}
An example of the construction of the network $G$ is presented in Figure~\ref{fig:npc-rewiring-degree}.

\begin{figure}[tbh!]
\centering
\includegraphics[width=.65\linewidth]{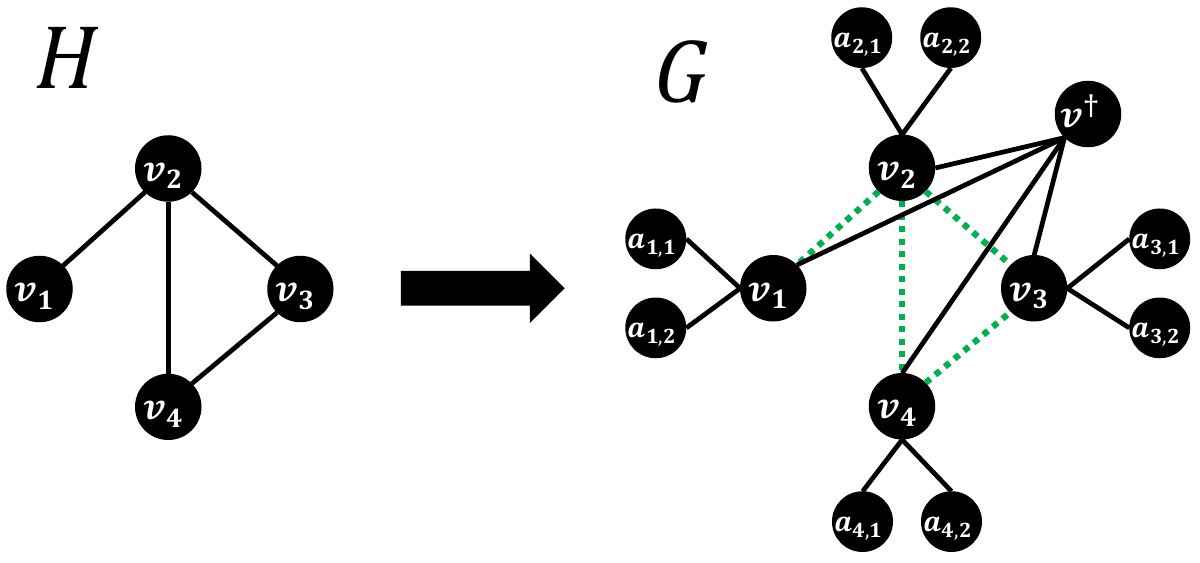}
\caption{
Construction of the network used in the proof of Theorem~\ref{thrm:npc-rewiring-degree}.
Green dotted edges are allowed to be added.
}
\label{fig:npc-rewiring-degree}
\end{figure}

Now, consider the instance $(G,\vs,\II,\sd,\thr,b,\FA,\FR)$ of the Hiding Source by Modifying Edges problem, where:
\begin{itemize}
\item $G$ is the network we just constructed,
\item $\vs$ is the evader,
\item $\II=V$, i.e., all nodes in $G$ are infected,
\item $\sd$ is the Degree source detection algorithm,
\item $\thr=k$ is the safety threshold,
\item $b=\frac{k(k-1)}{2}$,
\item $\FA = E'$, i.e., only edges existing in $H$ can be added to $G$,
\item $\FR = \emptyset$, i.e., none of the edges can be removed.
\end{itemize}

Since $\FR = \emptyset$, for any solution to the constructed instance of the Hiding Source by Modifying Edges problem we must have $\Rem = \emptyset$.
Hence, we will omit mentions of $\Rem$ in the remainder of the proof, and we will assume that a solution consists just of $\Add$.

Notice that the degree of the evader $\vs$ in $G$ is $n$, and it does not change after the addition of any $A \subseteq \FA$ (as we can only add edges between the members of $V'$).
Notice also that the degree of every node $a_{i,j}$ is $1$, and it cannot be increased.
Therefore, the only nodes that can contribute to satisfying the safety threshold (by increasing their degree to a value greater than the degree of the evader) are the nodes in $V'$.
The degree of any $v_i$ in $G$ is $n-k+2$ (as it is connected with $n-k+1$ nodes $a_{i,j}$ and with the node $\vs$).
Therefore, in order for a given $v_i$ to have greater degree than $\vs$, we have to add to $G$ at least $k-1$ edges incident with $v_i$.

We will now show that the constructed instance of the Hiding Source by Modifying Edges problem has a solution if and only if the given instance of the Finding $k$-Clique problem has a solution.

Assume that there exists a solution to the given instance of the Finding $k$-Clique problem, i.e., a subset $V^* \subseteq V'$ forming a $k$-clique in $H$.
We will show that $\Add = V^* \times V^*$ is a solution to the constructed instance of the Hiding Source by Modifying Edges.
First, notice that indeed $\Add \subseteq \FA$, as $\FA$ contains all edges from $H$, and $V^* \times V^*$ is a clique in $H$.
Notice also that adding $\Add$ to $G$ increases the degree of $k$ nodes in $V^*$ to $n+1$, i.e., to a value greater than the degree of the evader.
Hence, there now exist $k$ nodes with degree greater than the evader.
We showed that if there exists a solution to the given instance of the Finding $k$-Clique problem, then there also exists a solution to the constructed instance of the Hiding Source by Modifying Edges problem.

Assume that there exists a solution $\Add$ to the constructed instance of the Hiding Source by Modifying Edges problem.
We will show that $V^* = \bigcup_{(v,w) \in \Add} \{v,w\}$ forms a $k$-clique in $H$.
Notice that since $\Add$ is a solution, it increases the degree of at least $k$ nodes in $V'$ (since the safety threshold is $\thr=k$) by at least $k-1$.
However, since the budget is $b=\frac{k(k-1)}{2}$, adding $\Add$ must increase the degree of \textit{exactly} $k$ nodes in $V'$ by \textit{exactly} $k-1$.
If such a choice is available, the nodes in $V^*$ form a clique in $\FA$, therefore, they also form a clique in $H$.
We showed that if there exists a solution to the constructed instance of the Hiding Source by Modifying Edges problem, then there also exists a solution to the given instance of the Finding $k$-Clique problem.

This concludes the proof.
\end{proof}

\begin{theorem}
\label{thrm:npc-rewiring-closeness}
The problem of Hiding source by Modifying Edges is NP-complete given the Closeness source detection algorithm.
\end{theorem}

\begin{proof}
The problem is trivially in NP, since after adding and removing the given sets of edges $\Add$ and $\Rem$, it is possible to compute the closeness centrality ranking of all nodes in $G^\II$ in polynomial time.

We will now prove that the problem is NP-hard.
To this end, we will show a reduction from the NP-complete \textit{$3$-Set Cover} problem.
The decision version of this problem is defined by a universe, $U=\{u_1,\ldots,u_{|U|}\}$, and a collection of sets $S=\{S_1,\ldots,S_{|S|}\}$ such that $\forall_i S_i \subset U$ and $\forall_i |S_i|=3$, where the goal is to determine whether there exist $k$ elements of $S$ the union of which equals $U$.

Let $(U,S)$ be a given instance of the $3$-Set Cover problem.
Let us assume that $|S|+|U| \geq 8$, note that all instances where $|S|+|U| < 8$ can be solved in polynomial time.
We will now construct an instance of the Hiding Source by Modifying Edges problem.

First, let us construct a network $G=(V,E)$ where:
\begin{itemize}
\item $V = \{ \vs, w \} \cup S \cup U \cup \bigcup_{i=1}^{|S|-k+1} \{a_i\}$,
\item $E = \{(\vs,w)\} \cup \bigcup_{a_i \in V} \{(w,a_i)\} \cup \bigcup_{S_i \in V} \{(\vs,S_i)\} \cup \bigcup_{S_i \in V} \bigcup_{u_j \in S_i} \{(S_i,u_j)\}$.
\end{itemize}
In what follows we will denote the set of nodes $S_1,\ldots, \S_{|S|}$ by $S$, and we will denote the set of nodes $u_1,\ldots, u_{|U|}$ by $U$.
Notice that $u_j \in S_i$ in the last union in the formula of $E$ means that we connect a given node $S_i \in S$ only with nodes $u_j \in U$ corresponding to the elements contained in $S_i$ in the given instance of the $3$-Set Cover problem.
An example of the construction of the network $G$ is presented in Figure~\ref{fig:npc-rewiring-closeness}.

\begin{figure}[tbh!]
\centering
\includegraphics[width=.75\linewidth]{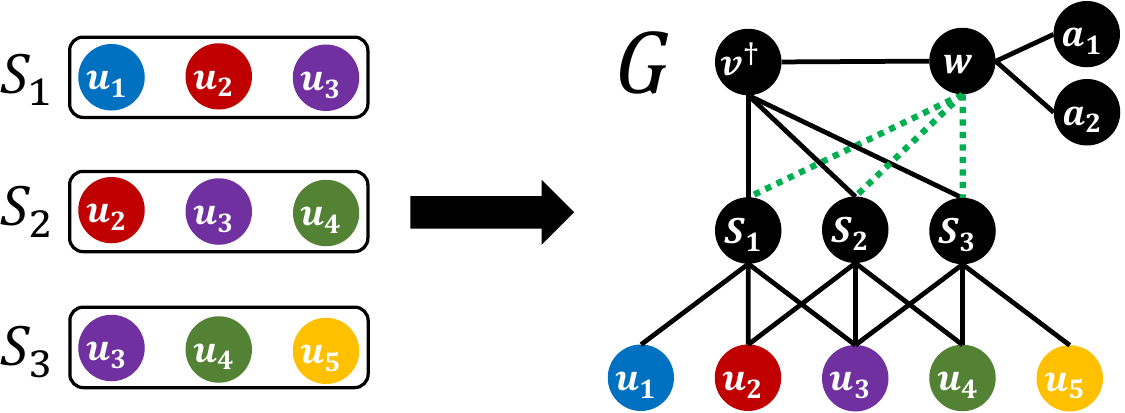}
\caption{
Construction of the network used in the proof of Theorem~\ref{thrm:npc-rewiring-closeness}.
Green dotted edges are allowed to be added.
Colors of the nodes $u_i$ express correspondence with the elements of the universe in the $3$-Set Cover problem instance.
}
\label{fig:npc-rewiring-closeness}
\end{figure}

Now, consider the instance $(G,\vs,\II,\sd,\thr,b,\FA,\FR)$ of the Hiding Source by Modifying Edges problem, where:
\begin{itemize}
\item $G$ is the network we just constructed,
\item $\vs$ is the evader,
\item $\II=V$, i.e., all nodes in $G$ are infected,
\item $\sd$ is the Closeness source detection algorithm,
\item $\thr=1$ is the safety threshold,
\item $b=k$, where $k$ is the size of solution of the given instance of the $3$-Set Cover problem,
\item $\FA = \{w\} \times S$, i.e., we allow to add any edges between $w$ and nodes in $S$,
\item $\FR = \emptyset$, i.e., none of the edges can be removed.
\end{itemize}

Since $\FR = \emptyset$, for any solution to the constructed instance of the Hiding Source by Modifying Edges problem we must have $\Rem = \emptyset$.
Hence, we will omit mentions of $\Rem$ in the remainder of the proof, and we will assume that a solution consists just of $\Add$.

First, let us analyze the closeness centrality values of the nodes in $G$ after the addition of any $A \subseteq \FA$.
Let $x_A$ denote the number of nodes $u_j \in U$ such that $d(w,u_j)=2$, i.e., the number of nodes $u_j \in U$ such that $w$ is connected (via the edges from $A$) with at least one node $S_i$ connected with $u_j$ (notice that there are no other possibilities for $w$ to be in distance $2$ from a node in $U$).
Moreover, let $D_v$ denote the sum of distances from $v$ to all nodes in the network, i.e., $D_v = \sum_{w \in V} d(v,w)$.
Notice that $\sdclos(v,G,\II) = \frac{1}{\sum_{w \in \II} D_v}$, i.e., greater value of $D_v$ implies lower position of the ranking of nodes according to Closeness source detection algorithm.
Table~\ref{tab:npc-rewiring-closeness} presents the computation of $D_v$ for every node $v \in V$ after the addition of a given $A \subseteq \FA$.

\begin{table}[tbh]
\centering
\begin{tabular}{ l c c c c c c }
$v$ & $d(v,\vs)$ & $d(v,w)$ & $\sum_j d(v,a_j)$ & $\sum_j d(v,S_j)$ & $\sum_j d(v,u_j)$ & $D_v$ \\
\hline
$\vs$ & $0$ & $1$ & $2(|S|-k+1)$ & $|S|$ & $2|U|$ & $3|S|+2|U|-2k+3$ \\
$w$ & $1$ & $0$ & $|S|-k+1$ & $2|S|-|A|$ & $3|U|-x_A$ & $3|S|+3|U|-k-|A|-x_A+2$ \\
$a_i$ & $2$ & $1$ & $2(|S|-k)$ & $3|S|-|A|$ & $4|U|-x_A$ & $5|S|+4|U|-2k-|A|-x_A+3$ \\
$S_i$ & $1$ & $\geq 1$ & $\geq 2(|S|-k+1)$ & $2(|S|-1)$ & $3+3(|U|-3)$ & $\geq 4|S|+3|U|-2k-4$ \\
$u_i$ & $2$ & $\geq 2$ & $\geq 3(|S|-k+1)$ & $\geq |S|$ & $\geq 2(|U|-1)$ & $\geq 4|S|+2|U|-3k+5$ \\
\hline
\end{tabular}
\caption{Sums of distances between nodes of the network after the addition of $A \subseteq \FA$, used in the proof of Theorem~\ref{thrm:npc-rewiring-closeness}.}
\label{tab:npc-rewiring-closeness}
\end{table}

Since the safety threshold is $\thr=1$, only one node have to have greater closeness centrality than $\vs$ after he addition of a given $A \subseteq \FA$ in order for said $A$ to be solution to the constructed instance of the Hiding Source by Modifying Edges problem.
However, notice that after adding any $A$ to the network we have $D_{\vs} < D_{a_i}$, $D_{\vs} < D_{u_i}$, and $D_{\vs} < D_{S_i}$ (the last inequality holds when we use the assumption that $|S|+|U| \geq 8$).
Hence, the only node that can have greater closeness centrality than $\vs$ (and therefore higher position in the ranking used to determine the source of diffusion) is $w$.
In other words, a given $A \subseteq \FA$ is a solution to the constructed instance of the Hiding Source by Modifying Edges problem if and only if we have $D_{\vs} > D_w$ after the addition of $A$.

We will now show that the constructed instance of the Hiding Source by Modifying Edges problem has a solution if and only if the given instance of the $3$-Set Cover problem has a solution.

Assume that there exists a solution to the given instance of the $3$-Set Cover problem, i.e., a subset $S^* \subseteq S$ of size $k$ the union of which is the universe $U$.
We will show that $\Add = \{w\} \times S^*$ (i.e., connecting $w$ with nodes corresponding to all sets in $S^*$) is a solution to the constructed instance of the Hiding Source by Modifying Edges problem.
First, notice that after the addition of $\Add$ there exists a path of length two from $w$ to every node $u_j \in U$, leading through the node corresponding to an element $S_i \in S^*$ containing $u_j$, with which $w$ is now connected.
Hence, we have that $x_{\Add} = |U|$ and $|\Add|=k$, which gives us:
$$
D_w = 3|S|+3|U|-k-|\Add|-x_{\Add}+2 = 3|S|+2|U|-2k+2 < 3|S|+2|U|-2k+3 = D_{\vs}.
$$
Therefore, after the addition of $\Add$ node $w$ has greater closeness centrality than the evader.
We showed that if there exists a solution to the given instance of the $3$-Set Cover problem, then there also exists a solution to the constructed instance of the Hiding Source by Modifying Edges problem.

Assume that there exists a solution $\Add$ to the constructed instance of the Hiding Source by Modifying Edges problem.
We will show that $S^* = \{S_i \in S : (w,S_i) \in \Add\}$ is a cover of $U$.
Let us compute the difference between $D_{\vs}$ and $D_w$ after the addition of $\Add$ (based on the values from Table~\ref{tab:npc-rewiring-closeness}):
$$
D_{\vs} - D_w = |\Add| + x_{\Add}+1-k-|U|.
$$
Since $\Add$ is a solution, $w$ must have greater centrality measure than $\vs$, implying $D_{\vs} - D_w > 0$, which gives us:
$$
|\Add| + x_{\Add}+1 > k+|U|.
$$
Notice however that $|\Add| \leq k$ (since the evader's budget is $b = k$) and that $x_{\Add} \leq |U|$ (since $x_{\Add}$ is the number of nodes in $U$ at distance $2$ from $w$).
Therefore, we must have $|\Add|=k$ and $x_{\Add}=|U|$, which means that after the addition of $\Add$ for every $u_j \in U$ we have a node $S_i \in S$ connected to both $w$ and $u_j$ (there is no other way of forming a path of length $2$ between $w$ and $u_j$).
Consequently, every element of the universe $u_j \in U$ is covered by at least one set $S_i \in S^*$, as $w$ is connected with a node $S_i$ only if it belongs to $S^*$, and $S_i$ is connected with $u_j$ only if it contains $u_j$ in the $3$-Set Cover problem instance.
We showed that if there exists a solution to the constructed instance of the Hiding Source by Modifying Edges problem, then there also exists a solution to the given instance of the $3$-Set Cover problem.

This concludes the proof.
\end{proof}

\begin{theorem}
\label{thrm:npc-rewiring-betweenness}
The problem of Hiding source by Modifying Edges is NP-complete given the Betweenness source detection algorithm.
\end{theorem}

\begin{proof}
The problem is trivially in NP, since after adding and removing the given sets of edges $\Add$ and $\Rem$, it is possible to compute the betweenness centrality ranking of all nodes in $G^\II$ in polynomial time.

We will now prove that the problem is NP-hard.
To this end, we will show a reduction from the NP-complete \textit{Finding $k$-Clique} problem.
The decision version of this problem is defined by a network, $H=(V',E')$, where $V'=\{v_1,\ldots,v_n\}$, and a constant $k \in \N$, where the goal is to determine whether there exist $k$ nodes forming a clique in $H$.

Let $(H,k)$ be a given instance of the Finding $k$-Clique problem.
Let us assume that $n \geq 3$, i.e., network $H$ has at least three nodes.
We will now construct an instance of the Hiding Source by Modifying Edges problem.

First, let us construct a network $G=(V,E)$ where:
\begin{itemize}
\item $V = V' \cup \{ \vs, w_1, w_2 \} \cup \bigcup_{i=1}^{n} \{a_i\} \cup \bigcup_{i=1}^{n-2} \{b_i\}$,
\item $E = E' \cup \{(v_1,w_1),(v_1,w_2)\} \cup \bigcup_{v_i \in V} \{(v_i,\vs), (v_i, a_i)\} \cup \bigcup_{b_i \in V} \{(b_i,w_1),(b_i,w_2)\}$.
\end{itemize}
In what follows we will denote the set of nodes $a_1,\ldots, a_n$ by $A$, we will denote the set of nodes $b_1,\ldots, b_{n-2}$ by $B$, and we will denote the set of nodes $w_1,w_2$ by $W$.
Notice that two nodes $v_i,v_j \in V'$ are connected in $G$ if and only if the are connected in $H$ (as $E'$ is included in $E$, and no other edges are added between the members of $V'$).
An example of the construction of the network $G$ is presented in Figure~\ref{fig:npc-rewiring-betweenness}.

\begin{figure}[tbh!]
\centering
\includegraphics[width=.75\linewidth]{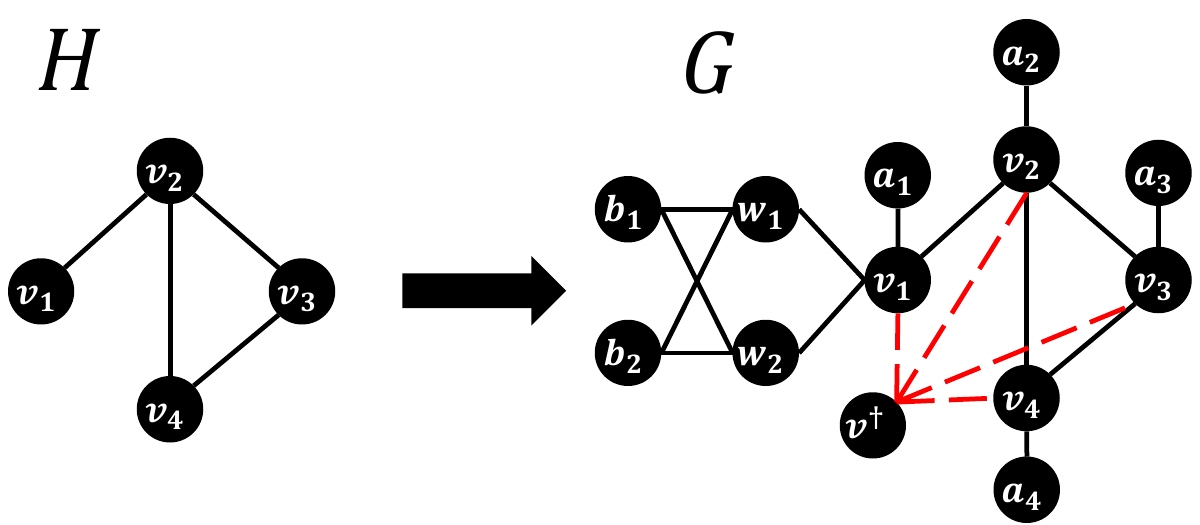}
\caption{
Construction of the network used in the proof of Theorem~\ref{thrm:npc-rewiring-betweenness}.
Red dashed edges are allowed to be removed.
}
\label{fig:npc-rewiring-betweenness}
\end{figure}

Now, consider the instance $(G,\vs,\II,\sd,\thr,b,\FA,\FR)$ of the Hiding Source by Modifying Edges problem, where:
\begin{itemize}
\item $G$ is the network we just constructed,
\item $\vs$ is the evader,
\item $\II=V$, i.e., all nodes in $G$ are infected,
\item $\sd$ is the Betweenness source detection algorithm,
\item $\thr=2n$ is the safety threshold,
\item $b=n-k$, where $k$ is the size of clique in the given instance of Finding $k$-Clique problem,
\item $\FA = \emptyset$, i.e., none of the edges can be removed,
\item $\FR = \{\vs\} \times V'$, i.e., only edges between $\vs$ and nodes in $V'$ can be removed.
\end{itemize}

Since $\FA = \emptyset$, for any solution to the constructed instance of the Hiding Source by Modifying Edges problem we must have $\Add = \emptyset$.
Hence, we will omit mentions of $\Add$ in the remainder of the proof, and we will assume that a solution consists just of $\Rem$.

Let $C_v$ denote the value of $\sdbetw(v,G,\II)$, i.e., the value used to determine the source of diffusion.
To remind the reader:
$$
C_v = \sum_{u \neq w : u,w \in \II \setminus \{v\}} \frac{|\{\pi \in \Pi(u,w): v \in \pi\}|}{|\Pi(u,w)|}
$$
where $\Pi(u,w)$ is the set of shortest paths between the nodes $u$ and $w$.

We will now make the following observations about the values of $C_v$ in $G$ after removal of an arbitrary $R \subseteq \FR$:
\begin{itemize}
\item $C_{v_i} \geq 3n-1$ for any $v_i \in V'$, as it controls all shortest paths between $a_i$ an all other $3n-1$ nodes,
\item $C_{w_i} = \frac{(n-2)(2n+1)}{2}$ for any $w_i \in W$, as it controls half of all the shortest paths between $n-2$ nodes in $B$ and all other $2n+1$ nodes in $\{\vs\} \cup V' \cup A$ (the other half is controlled by the other node in $W$) and it does not control any other shortest paths,
\item $C_{a_i} = 0$ for any $a_i \in A$, as it does not control any shortest paths,
\item $C_{b_i} = \frac{1}{n-1}$ for any $b_i \in B$, as it controls one of $n-1$ shortest paths between $w_1$ and $w_2$ (the other paths are controlled by other nodes in $B$ and by the node $v_1$) and it does not control any other shortest paths.
\end{itemize}

Notice that if after the removal of a given $R \subseteq \FR$ the evader $\vs$ is connected with at least two nodes in $v_i,v_j \in V'$ that do not have an edge between them, then we have $C_{\vs} \geq \frac{1}{n-1}$, as $\vs$ controls one shortest paths between $v_i$ and $v_j$, where other paths can only be controlled by the other $n-2$ nodes in $V'$.
The only nodes with the values of $C_v$ greater than $\frac{1}{n-1}$ are the $n+2$ nodes in $V'\cup W$.
However, the safety threshold is $\thr=2n$, so we also need $n-2$ nodes in $B$ to have greater value of $C_v$ than $\vs$ (notice that since $C_{a_i}=0$ for all $a_i \in A$, they will never contribute to the safety threshold).
Therefore, the safety margin requirement is fulfilled only if all neighbors of $\vs$ form a clique (in which case $C_{\vs}=0$).

We will now show that the constructed instance of the Hiding Source by Modifying Edges problem has a solution if and only if the given instance of the Finding $k$-Clique problem has a solution.

Assume that there exists a solution to the given instance of the Finding $k$-Clique problem, i.e., a subset $V^* \subseteq V'$ forming a $k$-clique in $H$.
Notice that after the removal of $\Rem = \{\vs\} \times (V' \setminus V^*)$, the evader $\vs$ is only connected with the nodes in $V^*$.
Since they form a clique in $H$, they also form a clique in $G$, hence according to the above observation the safety threshold is met.
Notice also that $|\Rem|=n-k$, hence the solution is within the evader's budget.
We showed that if there exists a solution to the given instance of the Finding $k$-Clique problem, then there also exists a solution to the constructed instance of the Hiding Source by Modifying Edges problem.

Assume that there exists a solution $\Rem$ to the constructed instance of the Hiding Source by Modifying Edges problem.
As observed above, the remaining nodes with which $\vs$ is connected, i.e., nodes $V^* = \{ v_i \in V':(vs,v_i) \notin \Rem\}$ must form a clique in $G$, hence they also must form a clique in $H$.
Since the budget of the evader is $b=n-k$, there are at least $k$ nodes in $V^*$.
We showed that if there exists a solution to the constructed instance of the Hiding Source by Modifying Edges problem, then there also exists a solution to the given instance of the Finding $k$-Clique problem.

This concludes the proof.
\end{proof}

\begin{theorem}
\label{thrm:npc-rewiring-rumor}
The problem of Hiding source by Modifying Edges is NP-complete given the Rumor source detection algorithm.
\end{theorem}

\begin{proof}
The problem is trivially in NP, since after adding and removing the given sets of edges $\Add$ and $\Rem$, it is possible to compute the rumor centrality ranking of all nodes in $G^\II$ in polynomial time.

We will now prove that the problem is NP-hard.
To this end, we will show a reduction from the NP-complete \textit{Finding a Hamiltonian Cycle} problem.
The decision version of this problem is defined by a network, $H=(V',E')$, where $V'=\{v_1,\ldots,v_n\}$, and where the goal is to determine whether there exists a Hamiltonian cycle in $H$, i.e., a cycle that visits each node exactly once.

Let $(H)$ be a given instance of the problem of Finding a Hamiltonian Cycle.
We will now construct an instance of the problem of Hiding Source by Modifying Edges.

First, let us construct a network $G=(V,E)$ where:
\begin{itemize}
\item $V = V' \cup \{ \vs, w \} \cup \bigcup_{i=1}^{3} \bigcup_{j=1}^{n} \{a_{i,j}\}$,
\item $E = E' \cup \{(\vs,v_1)\} \cup \bigcup_{v_i \in N_H(v_1)} \{(v_i,w)\} \cup \bigcup_{a_{i,n} \in V} \{(a_{i,n},w)\} \cup \bigcup_{i=1}^{3} \bigcup_{j=1}^{n-1} \{(a_{i,j},a_{i,j+1})\}$.
\end{itemize}
Notice that the evader $\vs$ is connected only with $v_1$, while $w$ is connected with all neighbors of $v_1$ in $H$ (we can assume that $v_1$ has at least two neighbors in $H$, as otherwise $H$ definitely does not have a Hamiltonian cycle).
An example of the construction of the network $G$ is presented in Figure~\ref{fig:npc-rewiring-rumor}.

\begin{figure}[tbh!]
\centering
\includegraphics[width=.9\linewidth]{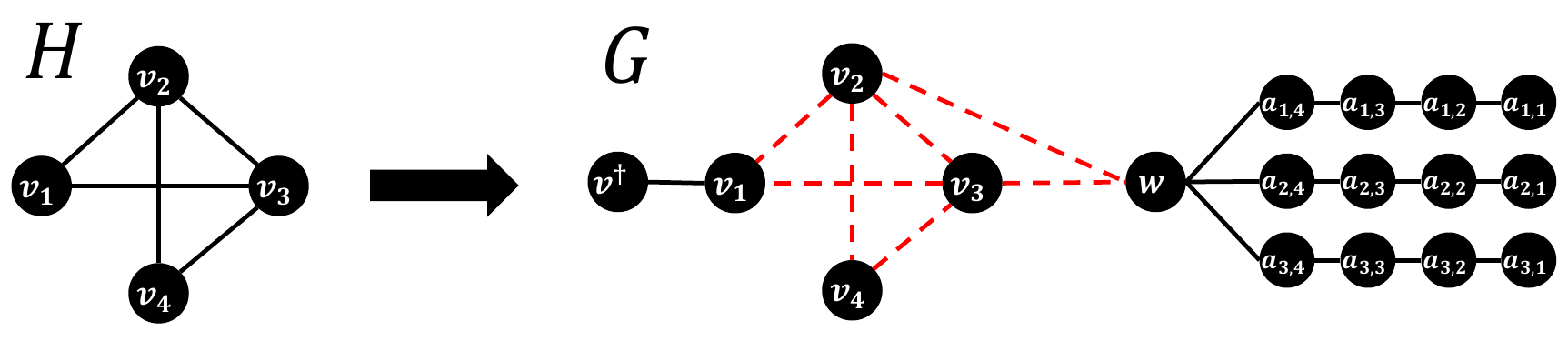}
\caption{
Construction of the network used in the proof of Theorem~\ref{thrm:npc-rewiring-rumor}.
Red dashed edges are allowed to be removed.
}
\label{fig:npc-rewiring-rumor}
\end{figure}

Now, consider the instance $(G,\vs,\II,\sd,\thr,b,\FA,\FR)$ of the problem of Hiding Source by Modifying Edges, where:
\begin{itemize}
\item $G$ is the network we just constructed,
\item $\vs$ is the evader,
\item $\II=V$, i.e., all nodes in $G$ are infected,
\item $\sd$ is the Rumor source detection algorithm,
\item $\thr=4n+1$ is the safety threshold,
\item $b=|E'|+|N_H(v_1)|-n$,
\item $\FA = \emptyset$, i.e., none of the edges can be added,
\item $\FR = E' \cup (\{w\} \times N_H(v_1))$, i.e., only edges belonging to the original set of edges in $H$ and edges between $w$ and neighbors of $v_1$ in $H$ can be removed.
\end{itemize}

Since $\FA = \emptyset$, for any solution to the constructed instance of the problem of Hiding Source by Modifying Edges, we must have $\Add = \emptyset$.
Hence, we will omit mentions of $\Add$ in the remainder of the proof, and we will assume that a solution consists of just $\Rem$.

Let $d_R$ denote distance between $v_1$ and $w$ in $G$ after the removal of $R$, i.e., $d_R=d_{(V,E \setminus R)}(v_1,w)$.
We will first prove the following lemma.

\begin{lemma}
\label{lem:npc-rewiring-rumor}
A given $R$ is a solution to the constructed instance of the problem of Hiding Source by Modifying Edges if and only if $d_R=n$.
\end{lemma}

\begin{proof}[Proof of Lemma~\ref{lem:npc-rewiring-rumor}]
We first show that if $d_R=n$ then $R$ is a solution to the constructed instance of the Hiding Source by Modifying Edges problem.
Notice that $d_R=n$ implies that all nodes in $V'$ form a path between $\vs$ and $w$ (as presented in the example in Figure~\ref{fig:npc-rewiring-rumor-example}), as there are exactly $n$ nodes in $V'$, no other nodes can be part of the shortest path between $\vs$ and $w$, and if at least one additional edge was added to this path, the distance between $v_1$ and $w$ would be smaller than $n$.

\begin{figure}[tbh!]
\centering
\includegraphics[width=.6\linewidth]{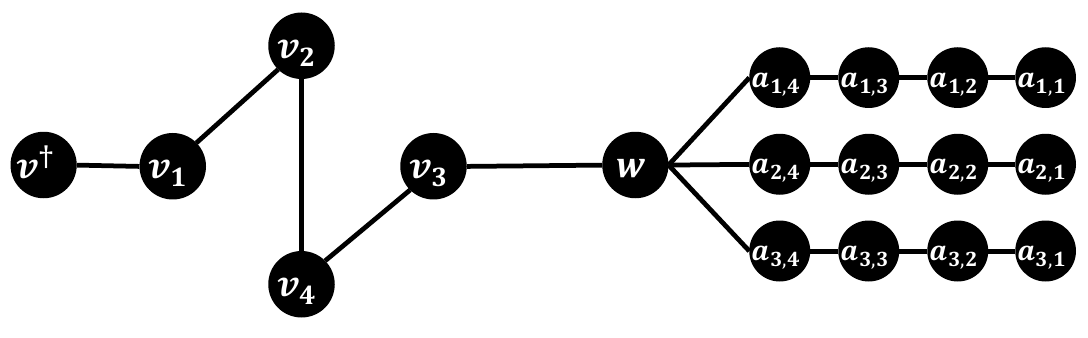}
\caption{
Network $G$ from Figure~\ref{fig:npc-rewiring-rumor} after the removal of $R$ such that nodes in $V'$ form a path between $\vs$ and $w$.
}
\label{fig:npc-rewiring-rumor-example}
\end{figure}

To remind the reader, the score assigned to a given node by the Rumor source detection algorithm is $\sdrumor(v,G,\II) = \frac{|\II|!}{\prod_{w \in \II} \Theta^v_w}$ where $\Theta^v_w$ is the size of the subtree of $w$ in the BFS tree of $G^\II$ rooted at $v$.
Let $C_v(d_R)$ denote $\prod_{w \in \II \setminus \{v\}} \Theta^v_w$ in $G$ where $R$ was removed.
Notice that greater $C_v(d_R)$ implies lower $\sdrumor(v,(V,E \setminus R),\II)$ and vice versa, as we have $\sdrumor(v,(V,E \setminus R),\II) = \frac{|\II|!}{|V| C_v(d_R)}$.

Let us now compute the values of $C_v(n)$ of all nodes in $G$, i.e., the value of $C_v$ when the removal of $R$ caused the nodes in $V'$ the form a path between $\vs$ and $w$:

\begin{itemize}
\item $C_{\vs}(n) = (4n+1)(4n)\ldots(3n+1) n!^3 = \frac{(4n+1)!n!^3}{(3n)!},$
\item $C_{v_i}(n) = j! (4n-j+1)(4n-j)\ldots(3n+1) n!^3 = \frac{j!(4n-j+1)!n!^3}{(3n)!}$, where $v_i$ is $j$-th node on the path from $\vs$ to $w$,
\item $C_w(n) = (n+1)!n!^3,$
\item $C_{a_{i,j}}(n) = (j-1)!(4n-j+2)(4n-j+1)\ldots(3n+2)(n+1)!n!^2 = \frac{(j-1)!(4n-j+2)!(n+1)!n!^2}{(3n+1)!}.$
\end{itemize}

To give the reader a better understanding of how the values of $C_v(n)$ are computed, we will describe in more detail the computation for $C_{\vs}(n)$; values for other nodes are computed analogically.
We root the BFS tree of network $G$ after the removal of $R$ in node $\vs$ (since the network is already a tree, all of its edges are part of the BFS tree).
The subtree rooted at $v_1$ is then of size $4n+1$, as it contain all nodes of the network other than $\vs$.
The subtree of the next node from $V'$ on the path from $\vs$ to $w$ (in case of the example in Figure~\ref{fig:npc-rewiring-rumor-example} it is $v_2$) is of size $4n$.
The size of the subtree of each subsequent node on the path from $\vs$ to $w$ is lower by one, until the subtree rooted at $w$, which is of size $3n+1$.
Based on similar reasoning, the product of sizes of subtrees on each branch consisting of nodes $a_{i,j}$ for a fixed $i$ is $n!$, and there are three such branches, which gives us $n!^3$.

To prove that $R$ is a solution to the constructed instance of the problem of Hiding Source by Modifying Edges, we now have to show that the safety threshold is met, i.e., that the evader has lower rumor centrality (greater $C_v(n)$) than all other nodes in $G$.
Indeed, we have:
\begin{itemize}
\item $C_{\vs}(n) / C_{v_i}(n) = \frac{(4n+1)!}{j!(4n-j+1)!} > 1$, given that $1 \leq j \leq n$ as $v_i$ is the $j$-th node on the path from $\vs$ to $w$,
\item $C_{\vs}(n) / C_w(n) = \frac{(4n+1)!}{(3n)!(n+1)!} > 1,$
\item $C_{\vs}(n) / C_{a_{i,j}}(n) = \frac{(4n+1)!n!^3}{(3n)!} \frac{(3n+1)!}{(j-1)!(4n-j+2)!(n+1)!n!^2} = \frac{(4n+1)!(3n+1)}{(j-1)!(4n-j+2)!(n+1)} > 1$, given that $1 \leq j \leq n$.
\end{itemize}
We have shown that if $d_R=n$ then $R$ is a solution to the constructed instance of the problem of Hiding Source by Modifying Edges.

We will now show that if $R$ is a solution to the constructed instance of the problem of Hiding Source by Modifying Edges then $d_R=n$.
Assume to the contrary, that $d_R<n$ and $R$ is a solution.
Since $R$ is a solution, then all nodes must have greater rumor centrality (and lower $C_v$) than $\vs$, in particular we must have $C_{a_{1,1}} < C_{\vs}$.

We have that:
$$
C_{\vs}(d_R) \leq (4n+1)(4n)\ldots(4n+1-d_R)(n-d_R)!n!^3 = \frac{(4n+1)!(n-d_R)!n!^3}{(4n-d_R)!}
$$
because there are $d_R$ nodes from $V'$ on the path from $\vs$ to $w$, while the other $n-d_R$ nodes have contribution at most $(n-d_R)!$ to the product (notice that $k!$ is the maximal value for the product of sizes of $k$ subtrees).
Similarly, we also have that:
$$
C_{a_{1,1}}(d_R) \geq (4n+1)(4n)\ldots(3n+2)(d_R+1)!n!^2 = \frac{(4n+1)!(d_R+1)!n!^2}{(3n+1)!}
$$
because the contribution of the nodes from $V'$ on the shortest path from $w$ to $\vs$ is at least $(d_R+1)!$, while the contribution of the other $n-d_R$ nodes in $V'$ is at least $1$.

To complete the proof of the lemma, we will now show that if $d_R < n$ then $C_{\vs}(d_R) \leq C_{a_{1,1}}(d_R)$ (hence, our assumption was false and $R$ cannot be a solution).
Let $UL(d_R)$ be the upper limit on $\frac{C_{\vs}(d_R)}{C_{a_{1,1}}(d_R)}$, i.e.:
$$
\frac{C_{\vs}(d_R)}{C_{a_{1,1}}(d_R)} \leq UL(d_R).
$$
The formula of $UL(d_R)$, based on the equations in the two previous paragraphs, is:
$$
UL(d_R) = \frac{(4n+1)!(n-d_R)!n!^3}{(4n-d_R)!} \frac{(3n+1)!}{(4n+1)!(d_R+1)!n!^2} = \frac{(n-d_R)!n!(3n+1)!}{(4n-d_R)!(d_R+1)!}.
$$
In particular:
\begin{itemize}
\item For $d_R=n-1$ we have that:
$$
UL(n-1) = \frac{1!n!(3n+1)!}{(3n+1)!n!} = 1.
$$
Hence $C_{\vs}(n-1) \leq C_{a_{1,1}}(n-1)$, i.e., $R$ is not a solution when $d_R=n-1$.
\item We also have that:
$$
\frac{UL(d_R+1)}{UL(d_R)} = \frac{(n-d_R-1)!n!(3n+1)!}{(4n-d_R-1)!(d_R+2)!} \frac{(4n-d_R)!(d_R+1)!}{(n-d_R)!n!(3n+1)!} = \frac{4n-d_R}{(d_R+2)(n-d_R)}.
$$
Notice that since $d_R<n$ and $d_R \geq 1$ we have that:
$$
\frac{UL(d_R+1)}{UL(d_R)} > \frac{3n}{3(n-1)} > 1.
$$
Hence $UL(d_R)$ is increasing with $d_R$, and for $d_R < n-1$ we have that
$$
UL(d_R) < UL(n-1) = 1,
$$
which implies that $C_{\vs}(d_R) < C_{a_{1,1}}(d_R)$ for $d_R < n-1$, i.e., $R$ is not a solution when $d_R<n-1$
\end{itemize}
We showed that if $R$ is a solution to the constructed instance of the problem of Hiding Source by Modifying Edges then $d_R=n$.

This concludes the proof of Lemma~\ref{lem:npc-rewiring-rumor}.
\end{proof}

Having proved Lemma~\ref{lem:npc-rewiring-rumor}, we now move back to proving Theorem~\ref{thrm:npc-rewiring-rumor}. To this end, we will show that the constructed instance of the problem of Hiding Source by Modifying Edges has a solution if and only if the given instance of the problem of Finding a Hamiltonian Cycle has a solution.

Assume that there exists a solution to the given instance of the Finding a Hamiltonian Cycle problem, i.e., a set of edges $E^* \subseteq E'$ inducing a Hamiltonian cycle in $H$.
Let $v^*$ be one of the neighbors of $v_1$ on the cycle.
We will show that $\Rem = (E' \setminus E^*) \cup \{(v_1,v^*)\} \cup \left( \{w\} \times (N_H(v_1) \setminus \{v^*\}) \right)$ is a solution to the constructed instance of the Hiding Source by Modifying Edges.
By removing edges $(E' \setminus E^*) \cup \{(v_1,v^*)\}$ from $G$, nodes in $V'$ now induce a path in $G$.
Node $v_1$ is connected with $\vs$ and with the neighbor on the cycle other than $v^*$, whereas $v^*$ is connected with $w$.
Notice also that since we removed the edges in $\{w\} \times (N_H(v_1) \setminus \{v^*\})$ from $G$, the node $v^*$ is the only node in $V'$ connected to $w$.
Finally, notice that the number of removed edges is within the evader's budget.
Hence, we have that $\vs$ is connected to $w$ with a path of nodes from $V'$ (a situation presented in Figure~\ref{fig:npc-rewiring-rumor-example}), which implies that $d_{\Rem}=n$.
Based on Lemma~\ref{lem:npc-rewiring-rumor}, we showed that $\Rem$ is a solution to the constructed instance of the Hiding Source by Modifying Edges problem.

Assume that there exists a solution $\Rem$ to the constructed instance of the problem of Hiding Source by Modifying Edges.
Based on Lemma~\ref{lem:npc-rewiring-rumor} it implies that $d_{\Rem}=n$, i.e., $\vs$ is connected to $w$ with a path of nodes from $V'$.
Let $v^*$ be the node from $V'$ directly connected to $w$, and let $E^*$ be the set of edges induced by $V'$ in $E \setminus \Rem$ (edges that form a path from $\vs$ to $w$ after the removal of $\Rem$).
Since all edges in $G$ between the nodes in $V'$ exist also in $H$, the set $E^*$ induces a path connecting all the nodes in $H$.
Moreover, since $v_1$ is the only node in $V'$ connected to $\vs$ in $G$, it has to be the first node on the path formed by $E^*$, with $v^*$ being the last node.
However, because of the way we constructed $G$, $v^*$ is a neighbor of $v_1$ in $H$.
Therefore, the set $E^* \cup \{(v_1,v^*)\}$ induces a Hamiltonian cycle in $H$.
We showed that if there exists a solution to the constructed instance of the problem of Hiding Source by Modifying Edges, then there also exists a solution to the given instance of the problem of Finding a Hamiltonian Cycle.

This concludes the proof.
\end{proof}

\begin{theorem}
\label{thrm:npc-rewiring-rwalk}
The problem of Hiding source by Modifying Edges is NP-complete given the Random Walk source detection algorithm.
\end{theorem}

\begin{proof}
The problem is trivially in NP, since after adding and removing the given sets of edges $\Add$ and $\Rem$, it is possible to compute the Random Walk source detection algorithm scores of all nodes in $G^\II$ in polynomial time.

We will now prove that the problem is NP-hard.
To this end, we will show a reduction from the NP-complete \textit{Finding a Hamiltonian Cycle} problem.
The decision version of this problem is defined by a network, $H=(V',E')$, where $V'=\{v_1,\ldots,v_n\}$, and where the goal is to determine whether there exists a Hamiltonian cycle in $H$, i.e., a cycle that visits each node exactly once.

Let $(H)$ be a given instance of the problem of Finding a Hamiltonian Cycle.
We will now construct an instance of the problem of Hiding Source by Modifying Edges.

First, let us construct a network $G=(V,E)$ where:
\begin{itemize}
\item $V = V' \cup \{ \vs, w, u \}$,
\item $E = E' \cup \{(\vs,v_1), (w,u)\} \cup \bigcup_{v_i \in N_H(v_1)} \{(v_i,w)\}$.
\end{itemize}
Notice that the evader $\vs$ is connected only with $v_1$, while $w$ is connected with all neighbors of $v_1$ in $H$ (we can assume that $v_1$ has at least two neighbors in $H$, as otherwise $H$ definitely does not have a Hamiltonian cycle).
An example of the construction of the network $G$ is presented in Figure~\ref{fig:npc-rewiring-rwalk}.

\begin{figure}[tbh!]
\centering
\includegraphics[width=.7\linewidth]{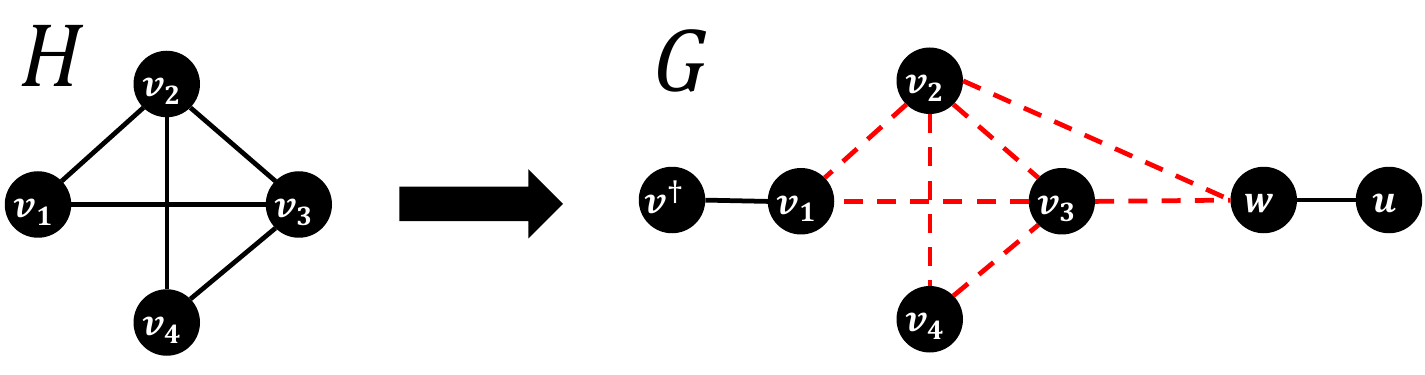}
\caption{
Construction of the network used in the proof of Theorem~\ref{thrm:npc-rewiring-rwalk}.
Red dashed edges are allowed to be removed.
}
\label{fig:npc-rewiring-rwalk}
\end{figure}

Now, consider the instance $(G,\vs,\II,\sd,\thr,b,\FA,\FR)$ of the problem of Hiding Source by Modifying Edges, where:
\begin{itemize}
\item $G$ is the network we just constructed,
\item $\vs$ is the evader,
\item $\II=V$, i.e., all nodes in $G$ are infected,
\item $\sd$ is the Random Walk source detection algorithm,
\item $\thr=n+1$ is the safety threshold,
\item $b=|E'|+|N_H(v_1)|-n$,
\item $\FA = \emptyset$, i.e., none of the edges can be added,
\item $\FR = E' \cup (\{w\} \times N_H(v_1))$, i.e., only edges belonging to the original set of edges in $H$ and edges between $w$ and neighbors of $v_1$ in $H$ can be removed.
\end{itemize}

Moreover, the Random Walk source detection algorithm needs to be parameterized with the model of spreading.
Let the algorithm be parameterized with the model in which the probability of propagation is $p=\frac{1}{2}$ and the time of diffusion is $T=n+1$.

Since $\FA = \emptyset$, for any solution to the constructed instance of the problem of Hiding Source by Modifying Edges, we must have $\Add = \emptyset$.
Hence, we will omit the mentions of $\Add$ in the remainder of the proof, and we will assume that a solution consists just of $\Rem$.

We will first prove the following lemma.

\begin{lemma}
\label{lem:npc-rewiring-rwalk}
If all nodes of the network are infected then the Random Walk source detection algorithm assigns to a given node $v \in V$ score $1$ if $\forall_{w \in \II} d_G(v,w) \leq T$ and $0$ otherwise.
\end{lemma}

\begin{proof}[Proof of Lemma~\ref{lem:npc-rewiring-rwalk}]
To remind the reader, the formula of the Random Walk source detection algorithm is:
$$
\sdrwalk(v,G,\II) =
\begin{cases}
	\phi_0(v) & \mbox{if } \forall_{w \in \II} d_G(v,w) \leq T \\
	0 & \mbox{otherwise}
\end{cases}
$$
where $T$ is the number of rounds in the SI model and $\phi$ is defined using the formula:
$$
\phi_t(v) =
\begin{cases}
	1 & \mbox{if } t=T \\
	(1-p) \phi_{t+1}(v) + \sum_{w \in N(v) \cap \II} \frac{p}{|N(v)|} \phi_{t+1}(w)  & \mbox{otherwise}
\end{cases}
$$
where $p$ is the probability of infection in the SI model.

Hence, the fact that the node $v$ gets assigned a score of $0$ if $\neg \forall_{w \in \II} d_G(v,w) \leq T$ follows from the formula of $\sdrwalk$.
We have to prove that if all nodes are infected and $\forall_{w \in \II} d_G(v,w) \leq T$ then $v$ is assigned a score of $1$.

Assume that $\II=V$ and $\forall_{v \in \II} \phi_{t+1}(v)=1$.
We have that:
$$
\phi_t(v) = (1-p) + |N(v)| \frac{p}{|N(v)|} = (1-p) + p = 1.
$$
We have shown that if $\II=V$ and $\forall_{v \in \II} \phi_{t+1}(v)=1$ then $\forall_{v \in \II} \phi_t(v)=1$.
Since we also have that $\forall_{v \in \II} \phi_T(v)=1$ then by induction we have that $\forall_{v \in \II} \phi_0(v)=1$

This concludes the proof of Lemma~\ref{lem:npc-rewiring-rwalk}.
\end{proof}

Since all nodes in $G$ are infected, the only way some of them can have a greater Random Walk algorithm score than the evader (and thus contribute to the safety threshold) is if $\exists_{w \in \II} d_G(\vs,w) > T$.
Notice that since $T=n+1$ and the network $G$ has $n+3$ nodes, the only way in which the distance between $\vs$ and some node is $n+2$ or greater is if the nodes in $G$ form a path with $\vs$ on one of its ends.
In that case $\vs$ and the node on the other end of the path have scores $0$, while all other $n+1$ nodes have scores $1$, hence the safety threshold $\thr=n+1$ is satisfied.

We will now show that the constructed instance of the problem of Hiding Source by Modifying Edges has a solution if and only if the given instance of the problem of Finding a Hamiltonian Cycle has a solution.

Assume that there exists a solution to the given instance of the problem of Finding a Hamiltonian Cycle, i.e., a set of edges $E^* \subseteq E'$ inducing a Hamiltonian cycle in $H$.
Let $v^*$ be one of the neighbors of $v_1$ on the cycle.
We will show that $\Rem = (E' \setminus E^*) \cup \{(v_1,v^*)\} \cup \left( \{w\} \times (N_H(v_1) \setminus \{v^*\}) \right)$ (notice that this set has exactly $b=|E'|+|N_H(v_1)|-n$ edges) is a solution to the constructed instance of the problem of Hiding Source by Modifying Edges.
By removing the edges $(E' \setminus E^*) \cup \{(v_1,v^*)\}$ from $G$, the nodes in $V'$ now induce a path in $G$.
Node $v_1$ is connected with $\vs$ and with the neighbor on the cycle other than $v^*$, whereas $v^*$ is connected with $w$.
Notice also that since we removed the edges in $\{w\} \times (N_H(v_1) \setminus \{v^*\})$ from $G$, $v^*$ is the only node in $V'$ connected with $w$.
Finally, notice the number of removed edges is within the evader's budget.
Hence, we have that after the removal of $\Rem$ the network $G$ is a path with $\vs$ on one of the ends (and $u$ on the other end).
Based on Lemma~\ref{lem:npc-rewiring-rwalk} and the previous observations, we showed that $\Rem$ is a solution to the constructed instance of the problem of Hiding Source by Modifying Edges.

Assume that there exists a solution $\Rem$ to the constructed instance of the problem of Hiding Source by Modifying Edges.
Based on Lemma~\ref{lem:npc-rewiring-rwalk} and the previous observations, it implies that after the removal of $\Rem$ the network $G$ is a path with $\vs$ on one of the ends.
Notice that $u$ has to be the other end of the path, as it has only one neighbor, namely $w$.
Hence, the nodes in the path between $\vs$ and $w$ must be all the nodes in $V'$, connected with the edges from $E'$ (as we did not add any other edges between the nodes in $V'$).
Moreover, the neighbor of $\vs$ on the path must be $v_1$, while the neighbor of $w$ (other than $u$) on the path has to be one of the neighbors of $v_1$ in $H$.
If we denote this neighbor of $w$ by $v^*$, then the set of edges induced by $V'$ in $(V, E\setminus \Rem)$ with the addition of $(v_1,v^*)$ induces a Hamiltonian cycle in $H$.
We showed that if there exists a solution to the constructed instance of the problem of Hiding Source by Modifying Edges, then there also exists a solution to the given instance of the problem of Finding a Hamiltonian Cycle.

This concludes the proof.
\end{proof}

\begin{theorem}
\label{thrm:npc-rewiring-mcarlo}
The problem of Hiding Source by Modifying Edges is NP-complete given the Monte Carlo source detection algorithm.
\end{theorem}

\begin{proof}
The problem is trivially in NP, since after adding and removing the given sets of edges $\Add$ and $\Rem$, it is possible to generate Monte Carlo samples and compute the ranking of all nodes in $G^\II$ in polynomial time.

We will now prove that the problem is NP-hard.
To this end, we will show a reduction from the NP-complete \textit{Dominating Set} problem.
The decision version of this problem is defined by a network, $H=(V',E')$, where $V'=\{v_1,\ldots,v_n\}$, and a constant $k \in \N$, where the goal is to determine whether there exist $V^* \subseteq V'$ such that $|V^*|=k$ and every node outside $V^*$ has at least one neighbor in $V^*$, i.e., $\forall_{v \in V' \setminus V^*} N_H(v) \cap V^* \neq \emptyset$.

Let $(H,k)$ be a given instance of the Dominating Set problem.
Let us assume that $k < n-1$, all other instances can be easily solved in polynomial time.
We will now construct an instance of the Hiding Source by Modifying Edges problem.

First, let us construct a network $G=(V,E)$ where:
\begin{itemize}
\item $V = V' \cup \{ \vs, u, w, x \} \cup \bigcup_{i=1}^{n-2} \{a_i\}$,
\item $E = E' \cup \{(\vs,u), (u,w), (w,x)\} \cup \bigcup_{i=1}^{n-2} \{(a_i,x)\}$.
\end{itemize}
An example of the construction of the network $G$ is presented in Figure~\ref{fig:npc-rewiring-mcarlo}.

\begin{figure}[tbh!]
\centering
\includegraphics[width=.75\linewidth]{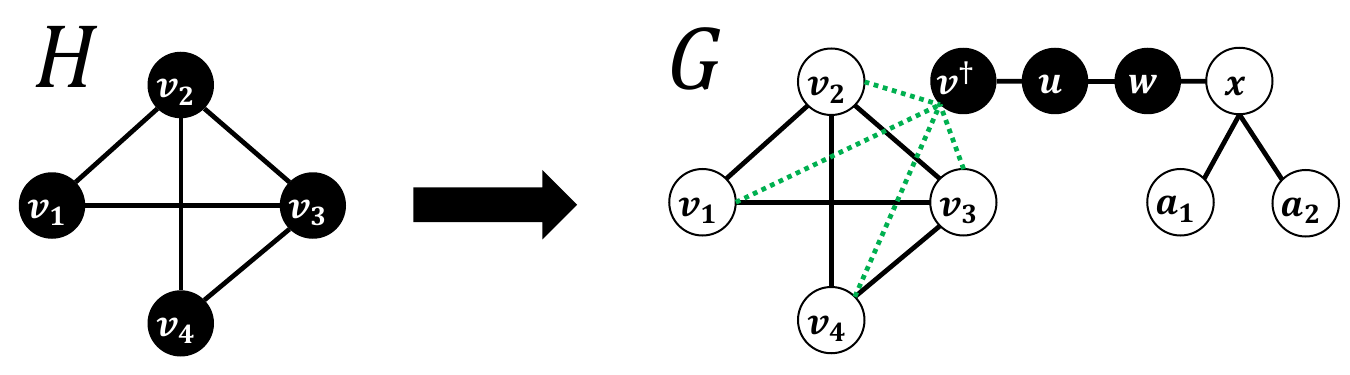}
\caption{
Construction of the network used in the proof of Theorem~\ref{thrm:npc-rewiring-mcarlo}.
Green dotted edges are allowed to be added.
Infected nodes in $G$ are highlighted in black.
}
\label{fig:npc-rewiring-mcarlo}
\end{figure}

Now, consider the instance $(G,\vs,\II,\sd,\thr,b,\FA,\FR)$ of the problem of Hiding Source by Modifying Edges, where:
\begin{itemize}
\item $G$ is the network we just constructed,
\item $\vs$ is the evader,
\item $\II=\{\vs,u,w\}$,
\item $\sd$ is the Monte Carlo source detection algorithm,
\item $\thr=2$ is the safety threshold,
\item $b=k$, where $k$ is the size of the dominating set from the Dominating Set problem instance,
\item $\FA = \{\vs\} \times V'$, i.e., only edges between $\vs$ and the members of $V'$ can be added,
\item $\FR = \emptyset$, i.e., none of the edges can be removed.
\end{itemize}

Since $\FR = \emptyset$, for any solution to the constructed instance of the Hiding Source by Modifying Edges problem we must have $\Rem = \emptyset$.
Hence, we will omit any mentions of $\Rem$ in the remainder of the proof, and we will assume that a solution consists just of $\Add$.

Moreover, the Monte Carlo source detection algorithm needs to be parameterized with the model of spreading.
Let the algorithm be parameterized with the model where the probability of propagation is $p=1$ and the time of diffusion is $T=2$.
Notice that since the model with $p=1$ is deterministic, all Monte Carlo samples starting from a given node always give the same result, hence the formula of the Monte Carlo source detection algorithm is in this case:
$$
\sdmcarlo(v,G,\II) = \exp\left( \frac{-\left( \frac{|\II \cap \II_v|}{|\II \cup \II_v|}- 1 \right)^2}{a^2} \right)
$$
where $\II_v$ is the set of infected nodes in the Monte Carlo sample starting the diffusion at $v$.
Moreover, notice that the position of a given node $v \in \II$ in the ranking generated by the Monte Carlo source detection algorithm depends solely on the value of $C_v= \frac{|\II \cap \II_v|}{|\II \cup \II_v|}$ (in particular, the node with the greatest value of $C_v$ is selected as the source of diffusion).

Let $V'_A$ denote the set of nodes from $V'$ connected with $\vs$ after the addition of $A$, i.e., $V'_A = \{v \in V': (\vs,v) \in A\}$.
Let us compute the values of $C_v$ for the nodes in $\II$ (notice that these are the only nodes that can be selected as the source of diffusion by the Monte Carlo algorithm) after the addition of an arbitrary $A$:
\begin{itemize}
\item $C_{\vs} = \frac{3}{|A|+3+m_A}$ where $m_A$ is the number of nodes in $V'/V'_A$ with at least one neighbor in $V'_A$,
\item $C_u = \frac{3}{|A|+4},$
\item $C_w = \frac{3}{n+2}.$
\end{itemize}

Since the safety threshold is $\thr=2$, both $u$ and $w$ need to have greater value of $C_v$ than $\vs$.
It is easy to see that $C_w > C_{\vs}$ if and only if $|A|+m_A=n$.
Notice also that if $|A|+m_A=n$ then also $m_A > 2$, since $|A| \leq k$ (as the budget of the evader is $k$) and since we assumed that $k < n -1$.
Therefore, the safety threshold is met if and only if $|A|+m_A=n$.

We will now show that the constructed instance of the Hiding Source by Modifying Edges problem has a solution if and only if the given instance of the Dominating Set problem has a solution.

Assume that there exists a solution to the given instance of the Dominating Set problem, i.e., a subset $V^* \subseteq V'$ of size $k$ such that all other nodes have a neighbor in $V^*$.
After adding to $G$ the set $\Add = \{\vs\} \times V^*$ we have that $|\Add|=k$ and $m_{\Add} = n-k$, which implies that the safety threshold is met.
We showed that if there exists a solution to the given instance of the Dominating Set problem, then there also exists a solution to the constructed instance of the Hiding Source by Modifying Edges problem.

Assume that there exists a solution $\Add$ to the constructed instance of the Hiding Source by Modifying Edges problem.
As shown above, we must have $|\Add|+m_{\Add}=n$.
Since the budget of the evader is $k$, the size of $\Add$ is at most $k$.
Therefore $V'_{\Add}$ is a dominating set in $H$ of size at most $k$ (we can add $k-|\Add|$ arbitrarily chosen elements to obtain a set of size exactly $k$).
We showed that if there exists a solution to the constructed instance of the Hiding Source by Modifying Edges problem, then there also exists a solution to the given instance of the Dominating Set problem.

This concludes the proof.
\end{proof}

\clearpage
\section{Hiding With and Without Strategic Manipulation}
\label{app:hiding-profiles}

In the main article, we discussed the concept of disentangling the two different notions of hiding---one provided by the structure of the network itself, and the other being the result of strategic manipulation.
Figure~2 in the main article shows the results of this analysis for the Eigenvector source detection algorithm.
Here, we present analogical results for all source detection algorithms considered in the study; see Figures~\ref{fig:profile-1} to~\ref{fig:profile-6}.
For an evaluation of the source detection algorithms themselves, without any attempts of hiding, see Figure~\ref{fig:standard-before-bars}.

As can be seen, the scale-free structure of the \BAn networks tends to keep the evader well-hidden (and even more so when the size and density of the network increase), whereas the identity of the evader is much more exposed in the other network structures. Next, we comment on the susceptibility of different network structures to strategic manipulation, starting with the addition of nodes, following by the modification of edges. In particular, when hiding by adding nodes, \ERn and Watts-Strogatz networks demonstrate less resilience to manipulation compared to their \BAn counterpart. Moreover, the effectiveness of hiding seems to be affected by the average degree of the network in a much more significant way than by the number of nodes in the network; the effect can be either positive or negative, depending on the source detection algorithm. Finally, let us comment on the susceptibility to strategic manipulation when the evader modifies edges. In this case, for most source detection algorithms, hiding the source of diffusion is more effective in scale-free networks, and the effectiveness increases in larger and denser scale-free networks. In contrast, no clear patterns were found for the other two types of network structure.

\begin{figure}[tbh]
\centering
\includegraphics[width=\linewidth]{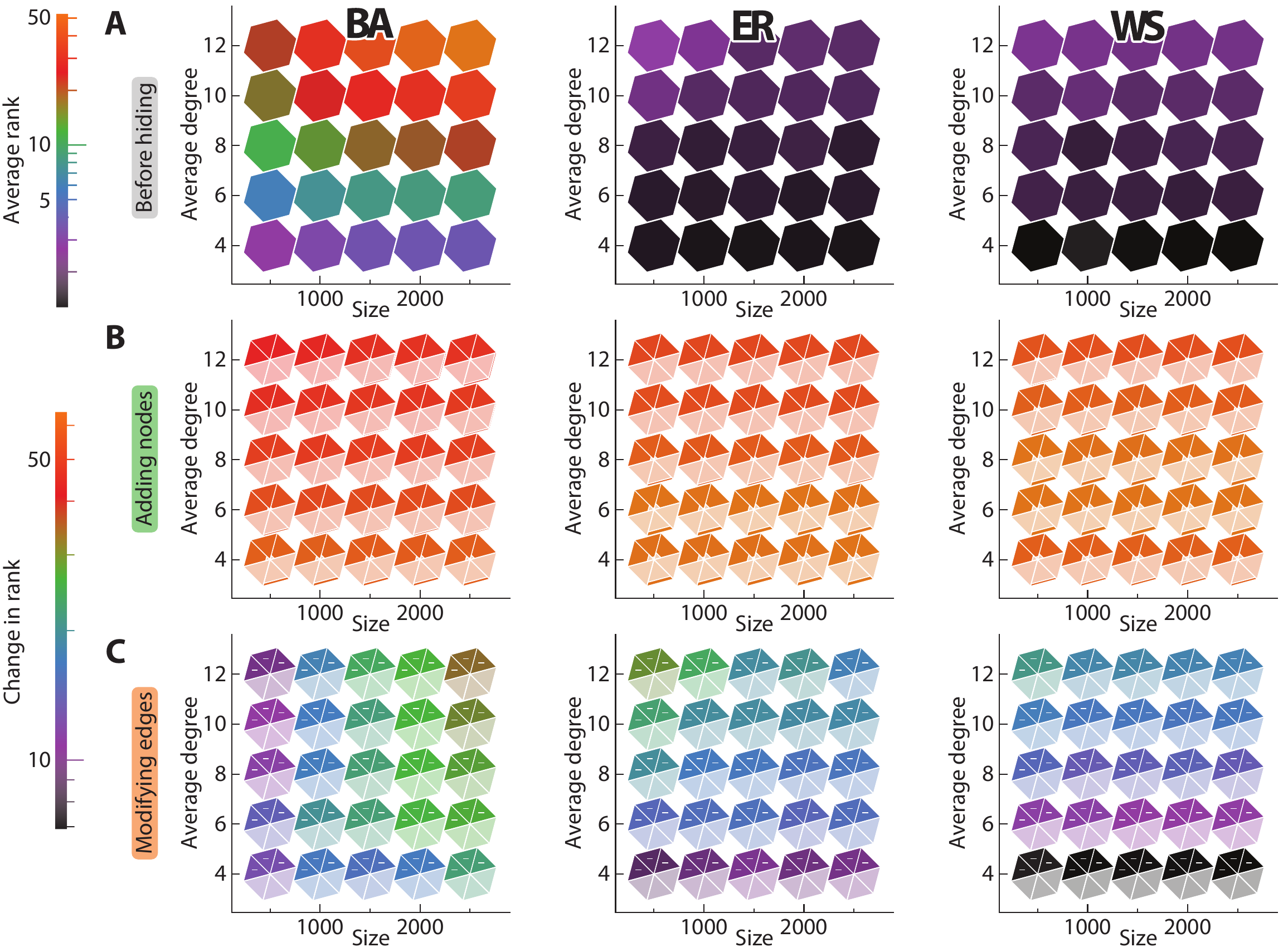}
\caption{
The same as Figure~2 in the main article, but for Degree source detection algorithm instead of Eigenvector.
}
\label{fig:profile-1}
\end{figure}

\begin{figure}[tbh]
\centering
\includegraphics[width=\linewidth]{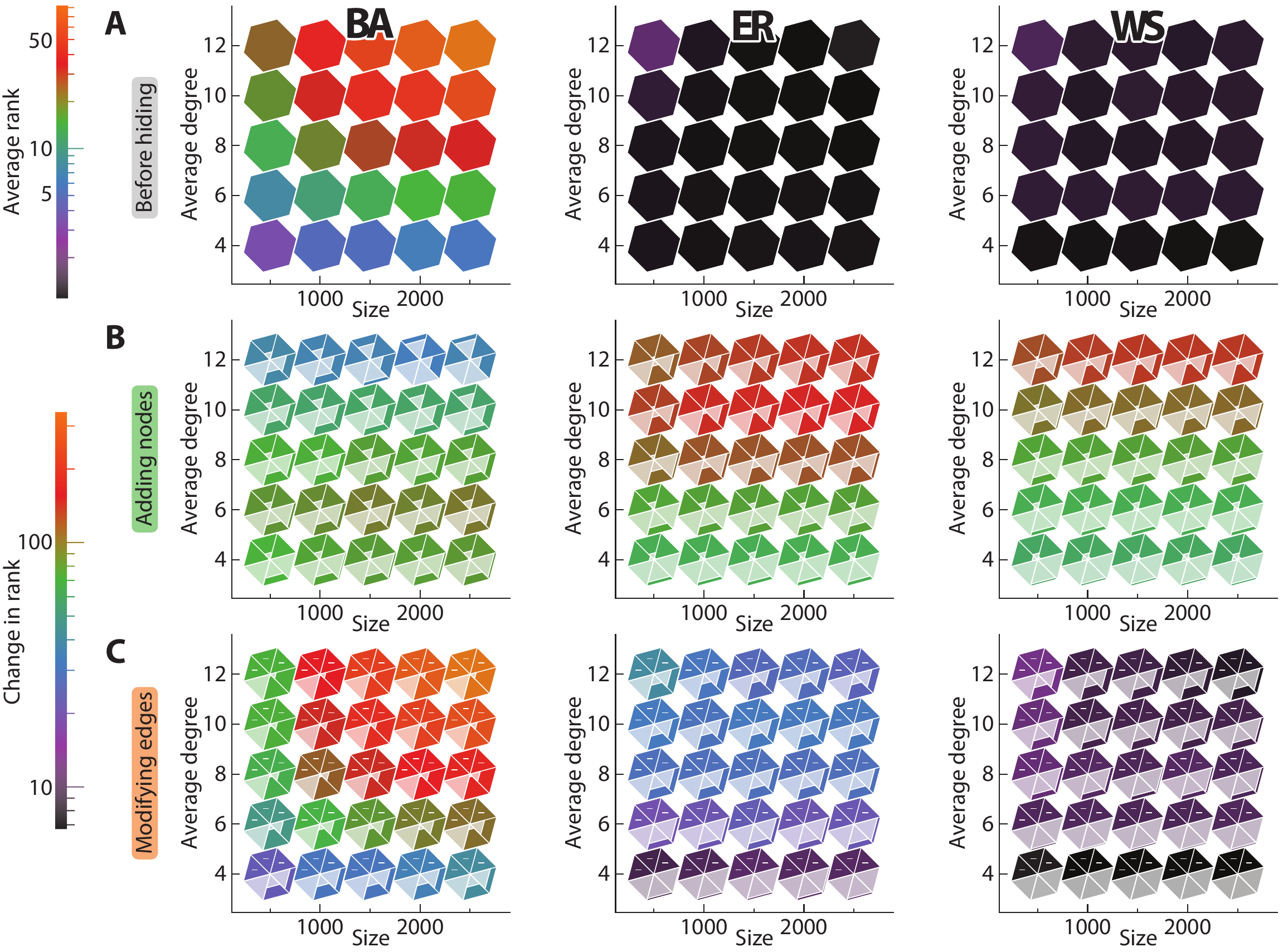}
\caption{
The same as Figure~2 in the main article, but for Closeness source detection algorithm instead of Eigenvector.
}
\label{fig:profile-2}
\end{figure}

\begin{figure}[tbh]
\centering
\includegraphics[width=\linewidth]{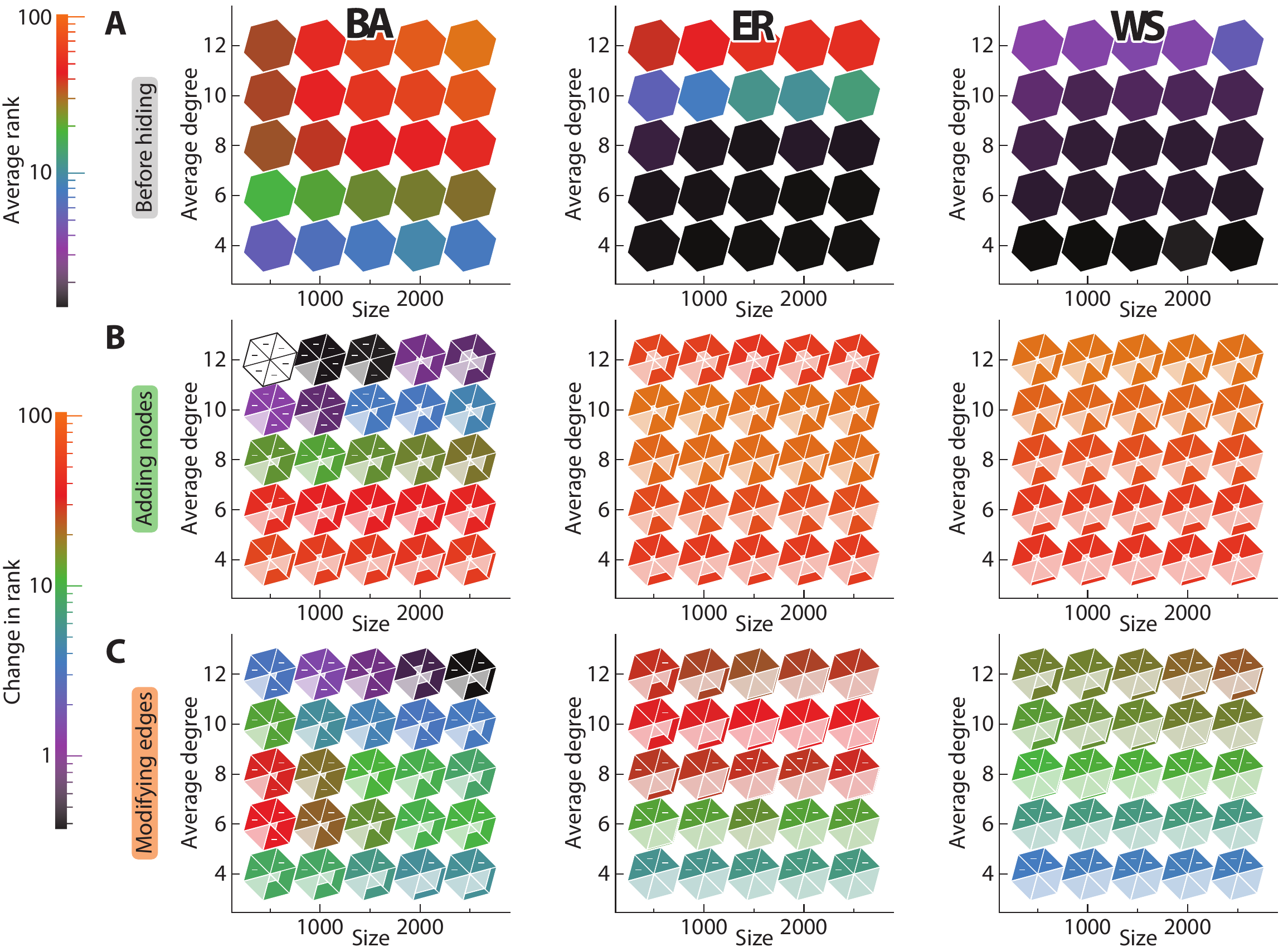}
\caption{
The same as Figure~2 in the main article, but for Rumor source detection algorithm instead of Eigenvector. White hexagons imply that all hiding strategies increase the visibility of the evader.
}
\label{fig:profile-3}
\end{figure}

\begin{figure}[tbh]
\centering
\includegraphics[width=\linewidth]{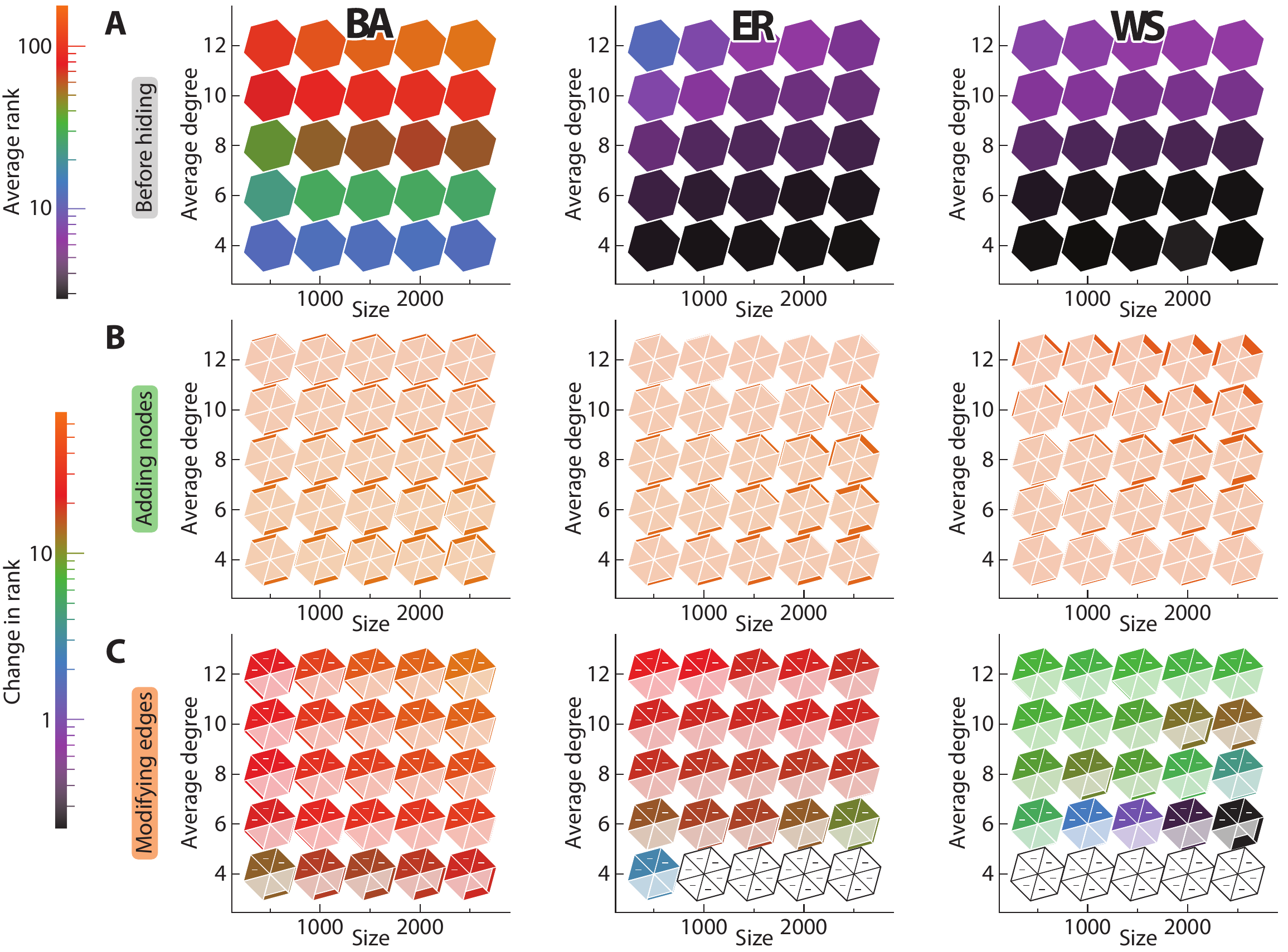}
\caption{
The same as Figure~2 in the main article, but for Random walk source detection algorithm instead of Eigenvector. White hexagons imply that all hiding strategies increase the visibility of the evader.
}
\label{fig:profile-4}
\end{figure}

\begin{figure}[tbh]
\centering
\includegraphics[width=\linewidth]{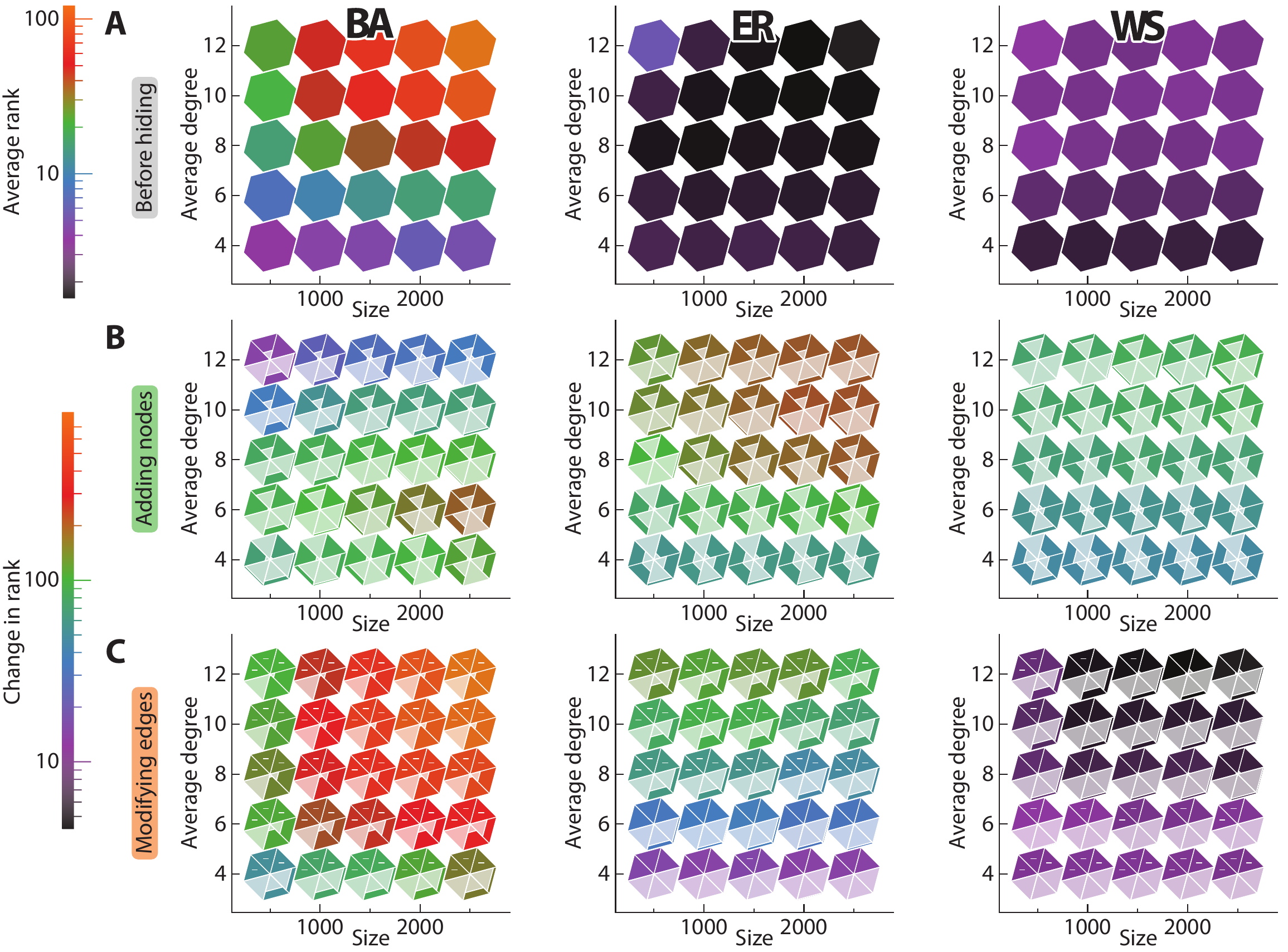}
\caption{
The same as Figure~2 in the main article, but for Monte Carlo source detection algorithm instead of Eigenvector.
}
\label{fig:profile-5}
\end{figure}

\begin{figure}[tbh]
\centering
\includegraphics[width=\linewidth]{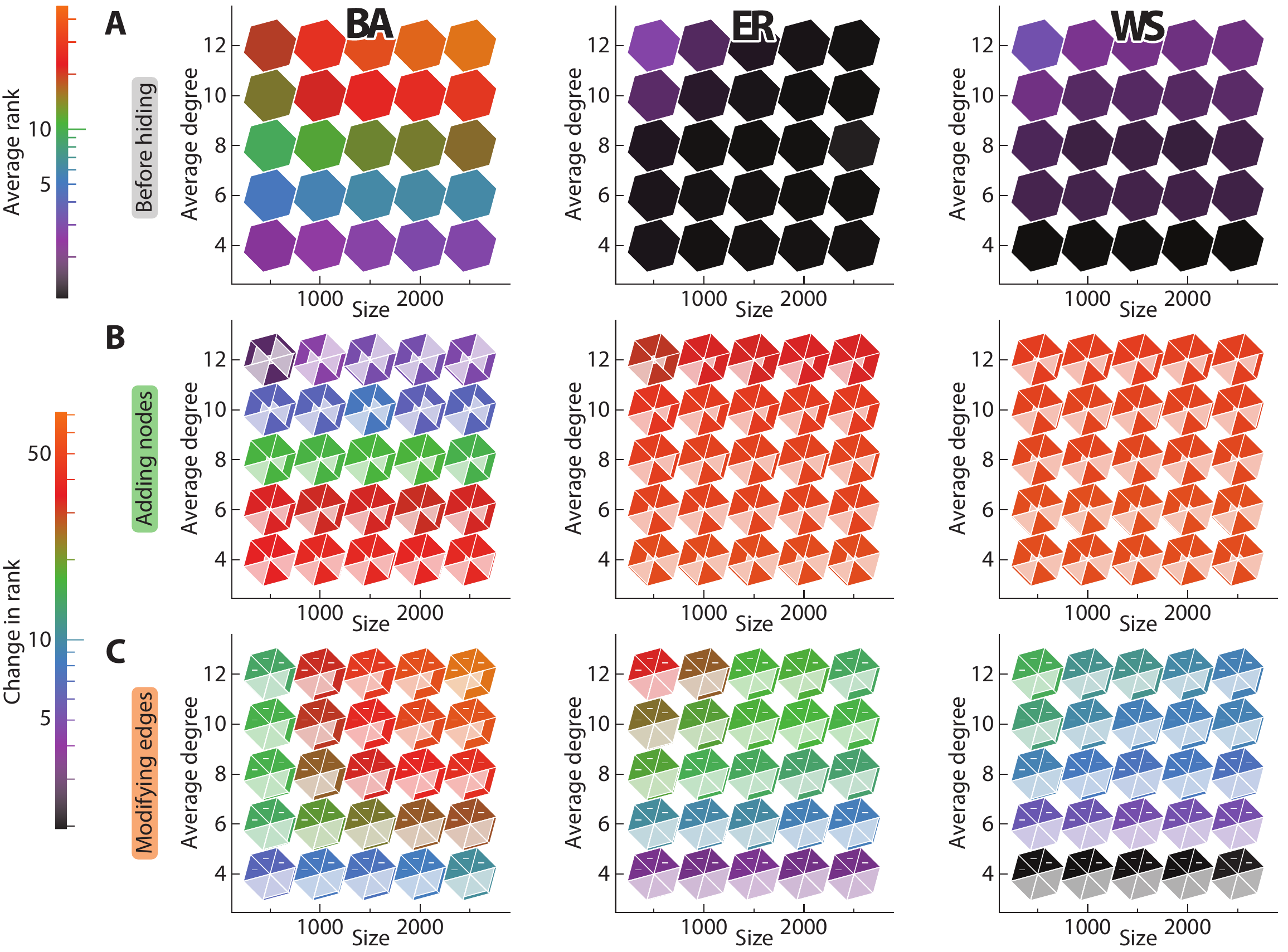}
\caption{
The same as Figure~2 in the main article, but for Betweenness source detection algorithm instead of Eigenvector.
}
\label{fig:profile-6}
\end{figure}

\begin{figure}[tbh]
\centering
\includegraphics[width=\linewidth]{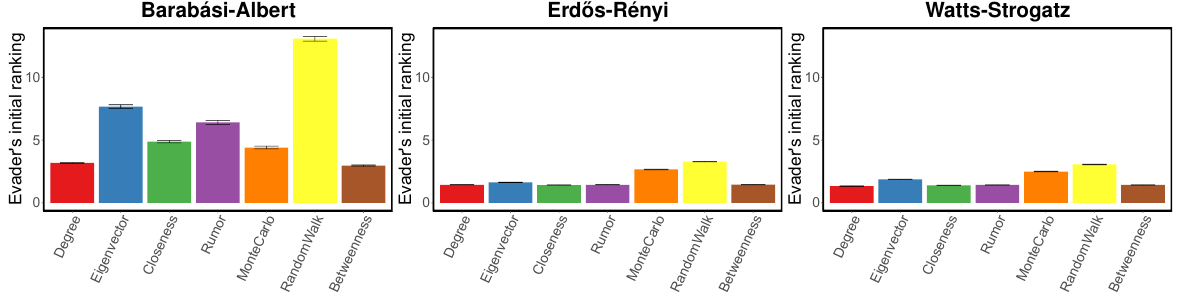}
\caption{
Comparison of the effectiveness of different source detection algorithms before the hiding process in random networks consisting of 1,000 nodes.
The x-axis corresponds to different source detection algorithms, while the y-axis corresponds to the initial ranking of the evader.
The error bars represent $95\%$ confidence intervals.
}
\label{fig:standard-before-bars}
\end{figure}

\clearpage
\section{Hiding by Adding Nodes---Additional Results}
\label{app:simulation-subos}

Examples of applying different heuristics that add nodes to the network are presented in Figure~\ref{fig:subo-heuristic-examples}.
In the examples each confederate is being connected with only one supporter.

\begin{figure}[tbh]
\centering
\includegraphics[width=\linewidth]{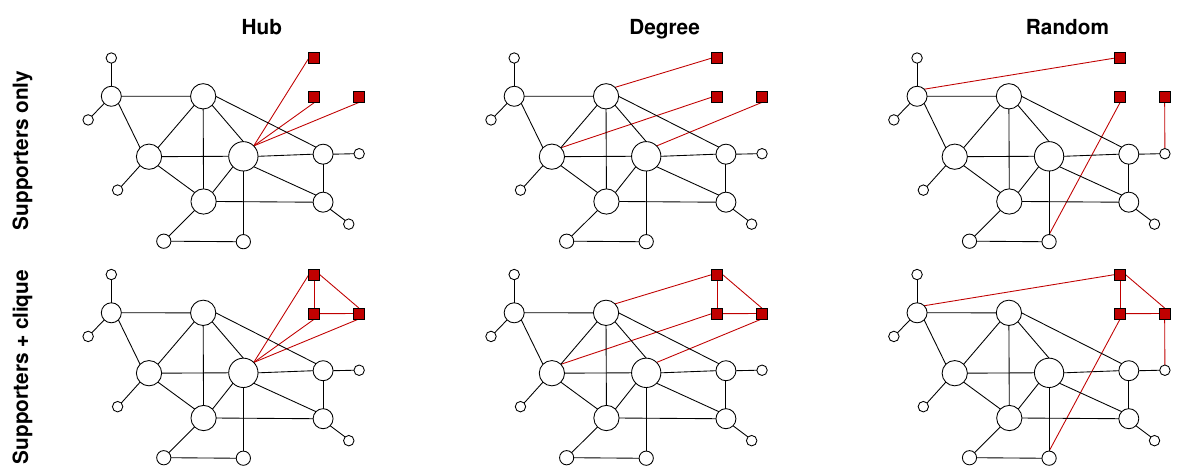}
\caption{
Examples of heuristics that add new nodes to the network.
Red squared nodes represent confederates (new nodes), while rounded nodes represent network members.
The size of each rounded node corresponds to its degree.
Each column corresponds to a different way of selecting supporters to be connected with the confederates.
The first row shows examples of heuristics that connects confederates with supporters only, while the second row shows examples of heuristics that additionally connect all confederates into a clique.
}
\label{fig:subo-heuristic-examples}
\end{figure}

Next, Figures~\ref{fig:subo-simulation-line} and~\ref{fig:subo-simulation-heat} present additional results; the first of these figures shows the change in the evader's ranking \textit{during} the hiding process, while the second figure shows the evader's ranking \textit{after} the hiding process. Several observations can be made based on these results. Regarding the different network structures, the source detection algorithms are on average less effective in the preferential-attachment networks generated using the \BAn model (i.e., the evader's position in the ranking is usually low to begin with). Regarding the different source detection algorithms, Betweenness and Rumor are on average the most resilient to hiding attempts (i.e., the average drop in the evader's ranking is the smallest), whereas Monte Carlo, Eigenvector, and Random walk are usually the least resilient. Regarding the selection of supporters to connect the confederates with, it is usually beneficial to link confederates with many different supporters (using either the random or degree heuristics), rather than connecting them all to the same set of supporters (using the hub heuristic), with the only exception being the Monte Carlo source detection algorithm. As for the question of whether confederates should be connected into a clique, our simulations show that this would indeed result in superior hiding of the evader. Moreover, in the majority of cases, connecting each confederate to three supporters proves to be more effective than connecting the confederate to just a single supporter.

\begin{figure}[tbh]
\centering
\includegraphics[width=.76\linewidth]{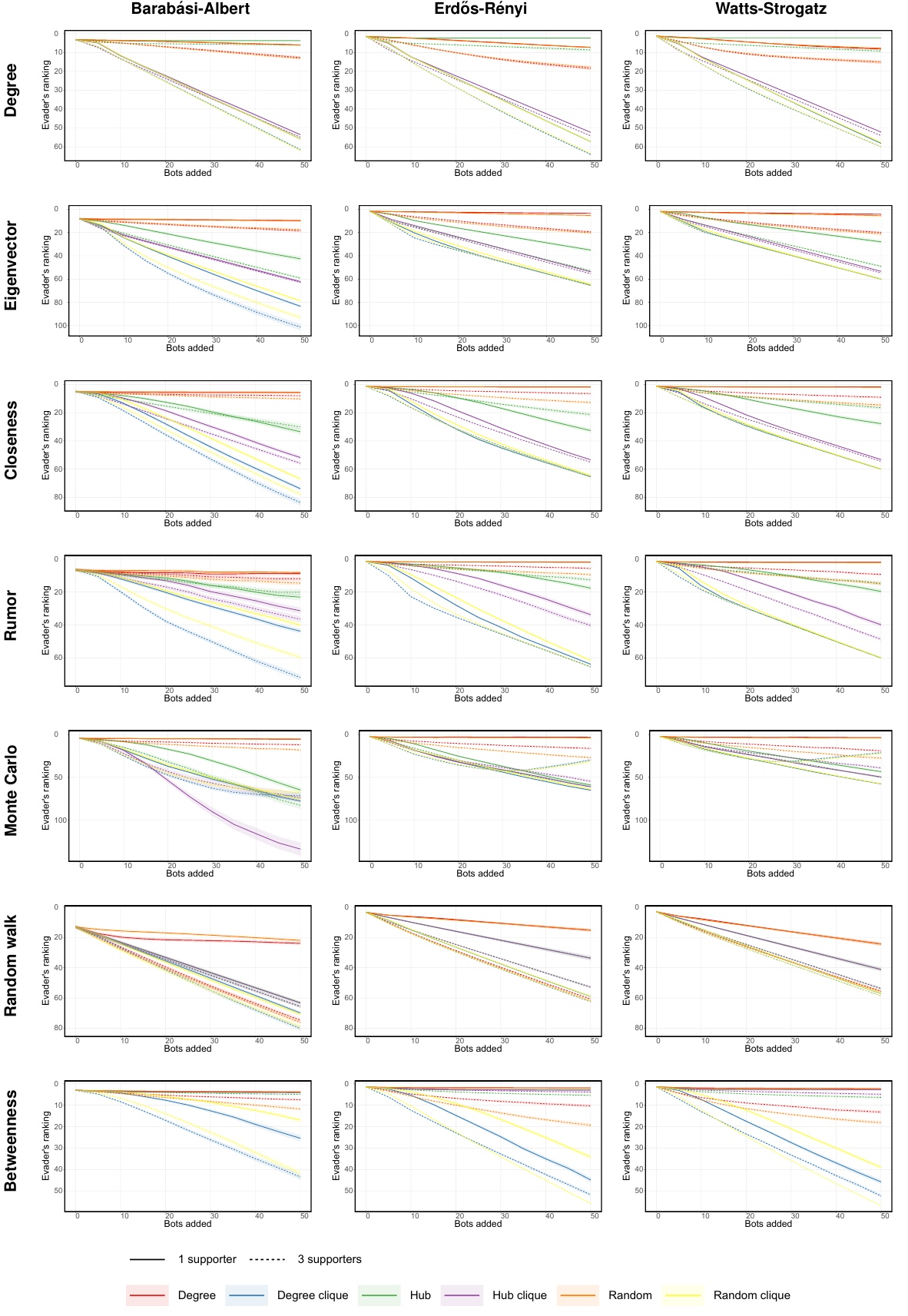}
\caption{
Results of simulations with hiding the source of diffusion by adding confederates to randomly generated networks with $1,000$ nodes and an average degree of $4$. The y-axis corresponds to the ranking of the evader according to the source detection algorithm (greater values indicate more efficient hiding), while the x-axis corresponds to the number of nodes added to the network. Each color corresponds to a different heuristic. Solid lines correspond to cases where each confederate is connected to a single supporter, while dashed lines correspond to cases where each confederate is connected to three supporters. Shaded areas represent $95\%$ confidence intervals.
}
\label{fig:subo-simulation-line}
\end{figure}

\begin{figure}[tbh]
\centering
\includegraphics[width=\linewidth]{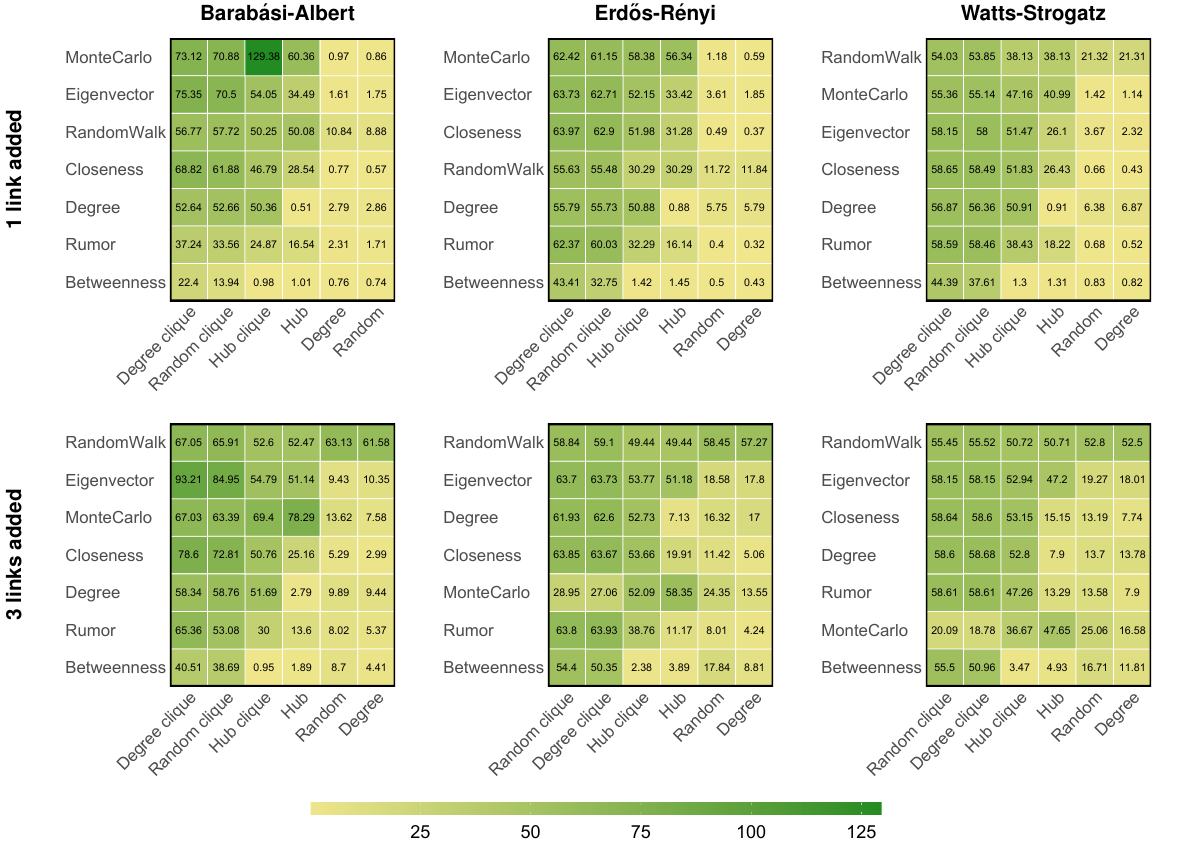}
\caption{
Results of simulations with hiding the source of diffusion by adding nodes to randomly generated networks with $1,000$ nodes with an average degree of $4$. The y-axis of each heatmap corresponds to different source detection algorithms, whereas the x-axis corresponds to different heuristics. The value in each cell indicates the change in the evader's ranking according to the source detection algorithm after adding $50$ confederates to the network using the heuristic. Rows and columns are sorted by average value.
}
\label{fig:subo-simulation-heat}
\end{figure}

\clearpage
\section{Hiding by Modifying Edges---Additional Results}
\label{app:simulation-edges}

Examples of applying different heuristics that modify edges of the network are presented in Figure~\ref{fig:edges-heuristic-examples}.

\begin{figure}[tbh]
\centering
\includegraphics[width=\linewidth]{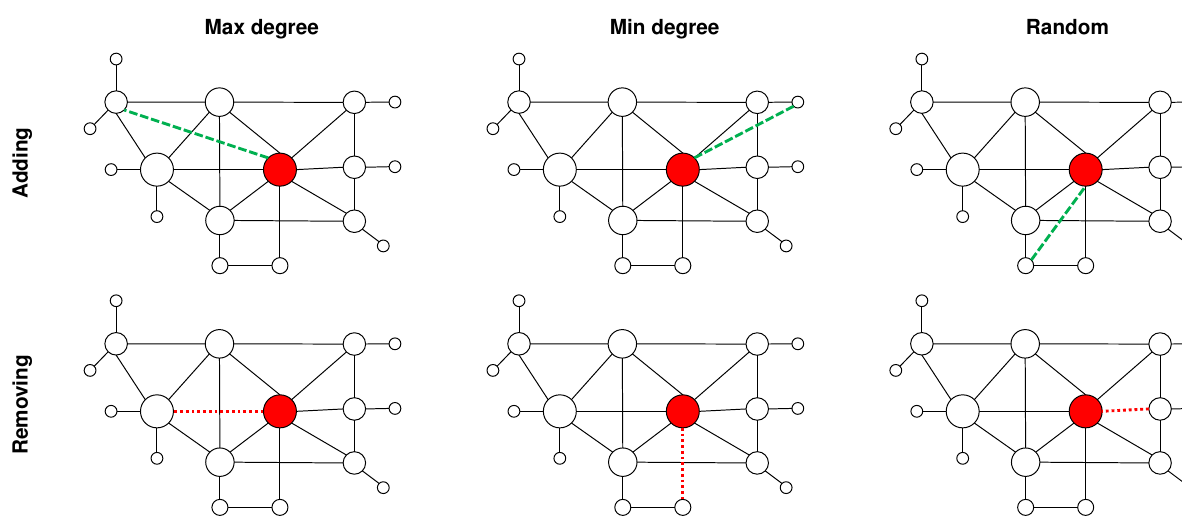}
\caption{
Examples of heuristics that modify edges of the network. The red node in each network represents the evader. The size of each node corresponds to its degree. Green dashed edges represent newly added connections, while dotted red edges represent newly removed connections. The first row shows examples of heuristics that add edges to the network, while the second row shows examples of heuristics that remove edges from the network. Each column corresponds to a different way of selecting the nodes that will be connected to, or disconnected from, the evader.
}
\label{fig:edges-heuristic-examples}
\end{figure}

Figures~\ref{fig:edge-simulation-line} to~\ref{fig:edge-simulation-mix-2} present the results of our simulations on networks with $1,000$ nodes, with Figures~\ref{fig:edge-simulation-line} and~\ref{fig:edge-simulation-heat} showing the results for the pure strategies, and Figures~\ref{fig:edge-simulation-mix-1} and~\ref{fig:edge-simulation-mix-2} showing the results for the mixed strategies. Our results indicate that removing existing edges from the network is significantly more effective in hiding the source of diffusion than adding new edges. In fact, adding new connections incident with the evader exposes the evader even more to the source detection algorithms in the vast majority of cases. Mixing between the two types of strategies typically provides worse performance in terms of hiding than simply running the removal component of the strategy. Regarding the choice of the neighbor to disconnect the evader from, selecting the neighbors with the greatest degree provides the best hiding to the evader, followed by neighbors selected at random, with low-degree neighbors being the worst choice. Nevertheless, in many cases all three removal heuristics provide very similar performance.

\begin{figure}[tbh]
\centering
\includegraphics[width=.78\linewidth]{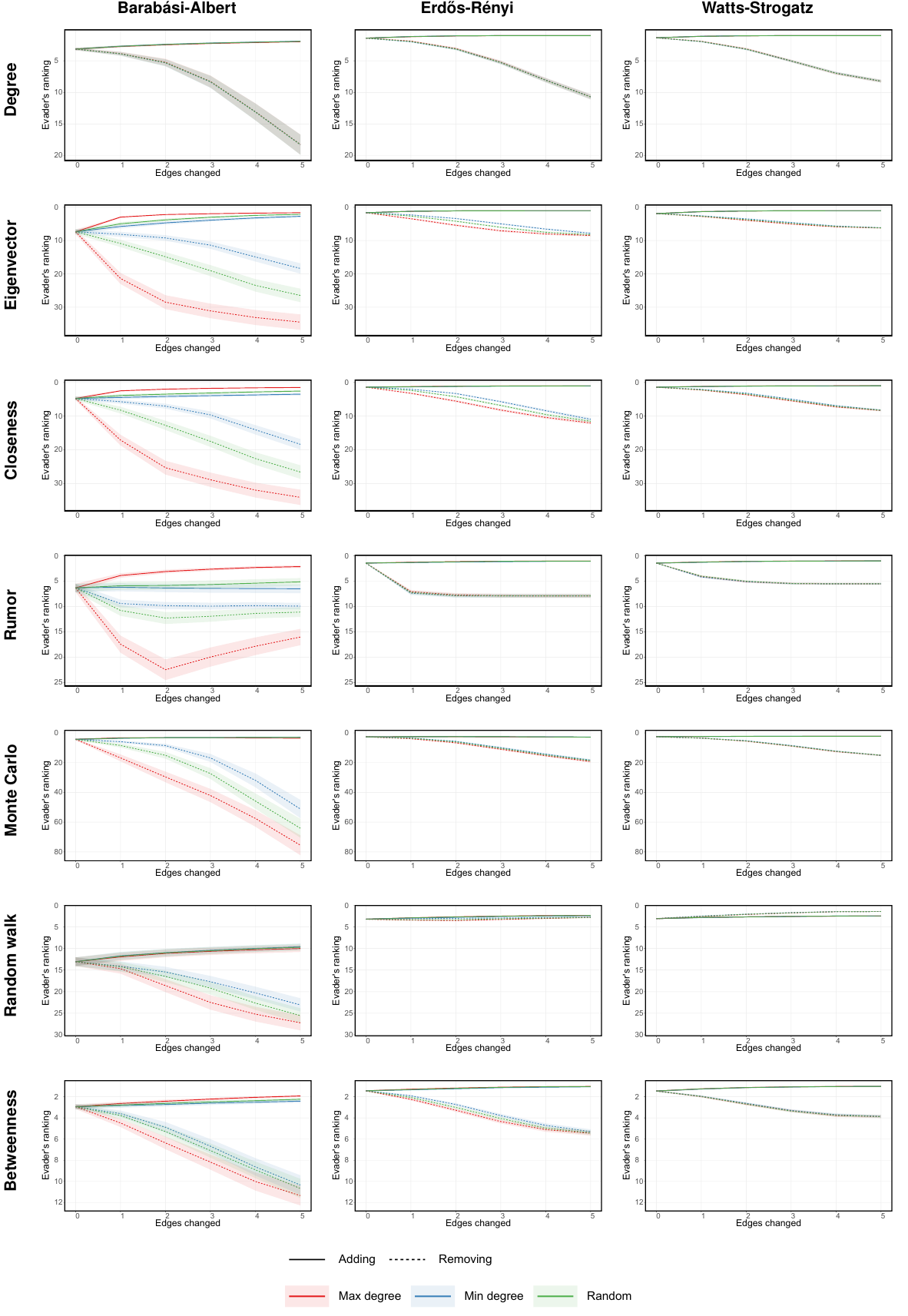}
\caption{
Results of hiding the source of diffusion by modifying edges in networks consisting of $1,000$ nodes with an average degree of $4$. The y-axis represents the evader's ranking according to the source detection algorithm (greater value indicates more effective hiding); the x-axis corresponds to the number of edges added to, or removed from, the network. Each color corresponds to a different way of choosing edges, while each line type (dashed or solid) corresponds to either adding or removing. Shaded areas represent $95\%$ confidence intervals.
}
\label{fig:edge-simulation-line}
\end{figure}

\begin{figure}[tbh]
\centering
\includegraphics[width=\linewidth]{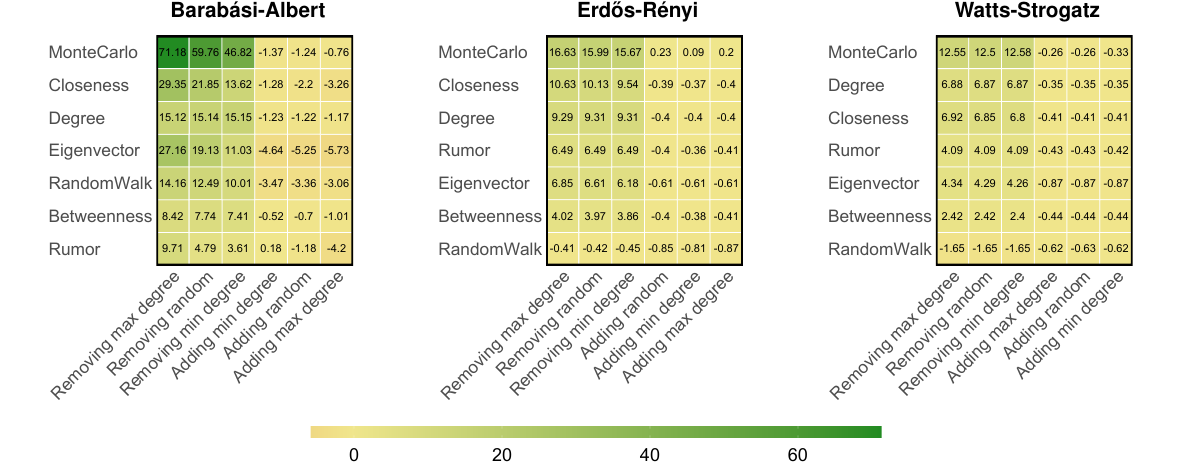}
\caption{
Results of hiding the source of diffusion by modifying edges in networks consisting of $1,000$ nodes with an average degree of $4$. In each heatmap, rows correspond to different source detection algorithms, while columns correspond to different heuristics. The value in each cell indicates the change in the evader's ranking according to the source detection algorithm as a result of adding or removing $5$ edges to the network, depending on the heuristic. Positive values indicate that the evader became more hidden, with greater values indicated a more effective disguise. In contrast, negative values indicate that the evader became less hidden. Rows and columns are sorted by average value.
}
\label{fig:edge-simulation-heat}
\end{figure}

\begin{figure}[tbh]
\centering
\includegraphics[width=\linewidth]{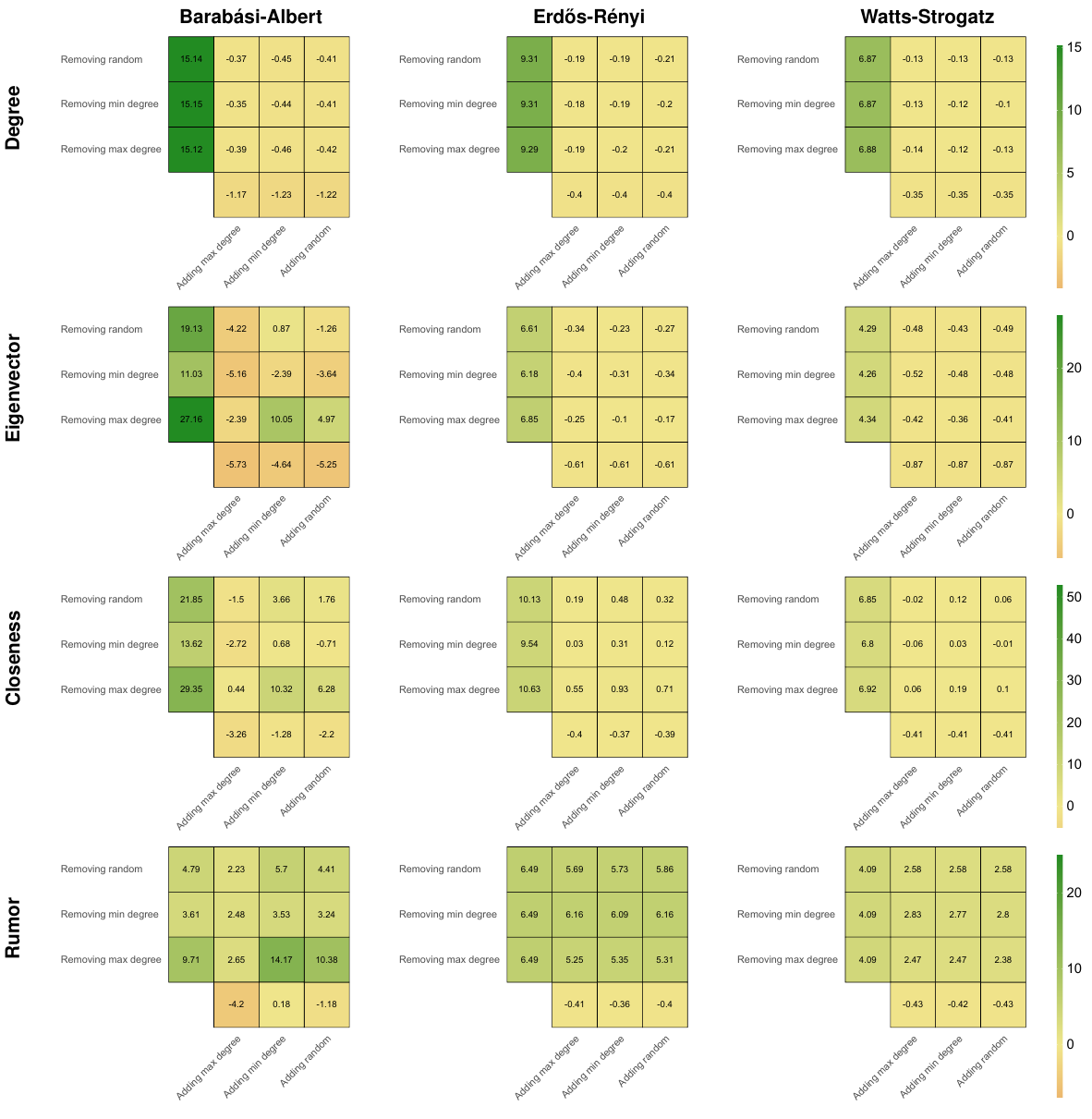}
\caption{
Results of hiding the source of diffusion by modifying edges in networks consisting of $1,000$ nodes with an average degree of $4$. In each heatmap, rows correspond to different removal heuristics, while columns correspond to different addition heuristics. The first column and the last row represent pure strategies, while the remaining cells represent mixed strategies. The value in each cell indicates the change in the evader's ranking according to the source detection algorithm after adding or removing $5$ edges, depending on the heuristic. Positive values indicate that the evader became more hidden, with greater values indicated a more effective disguise. In contrast, negative values indicate that the evader became less hidden. Rows and columns are sorted by average value.
}
\label{fig:edge-simulation-mix-1}
\end{figure}

\begin{figure}[tbh]
\includegraphics[width=\linewidth]{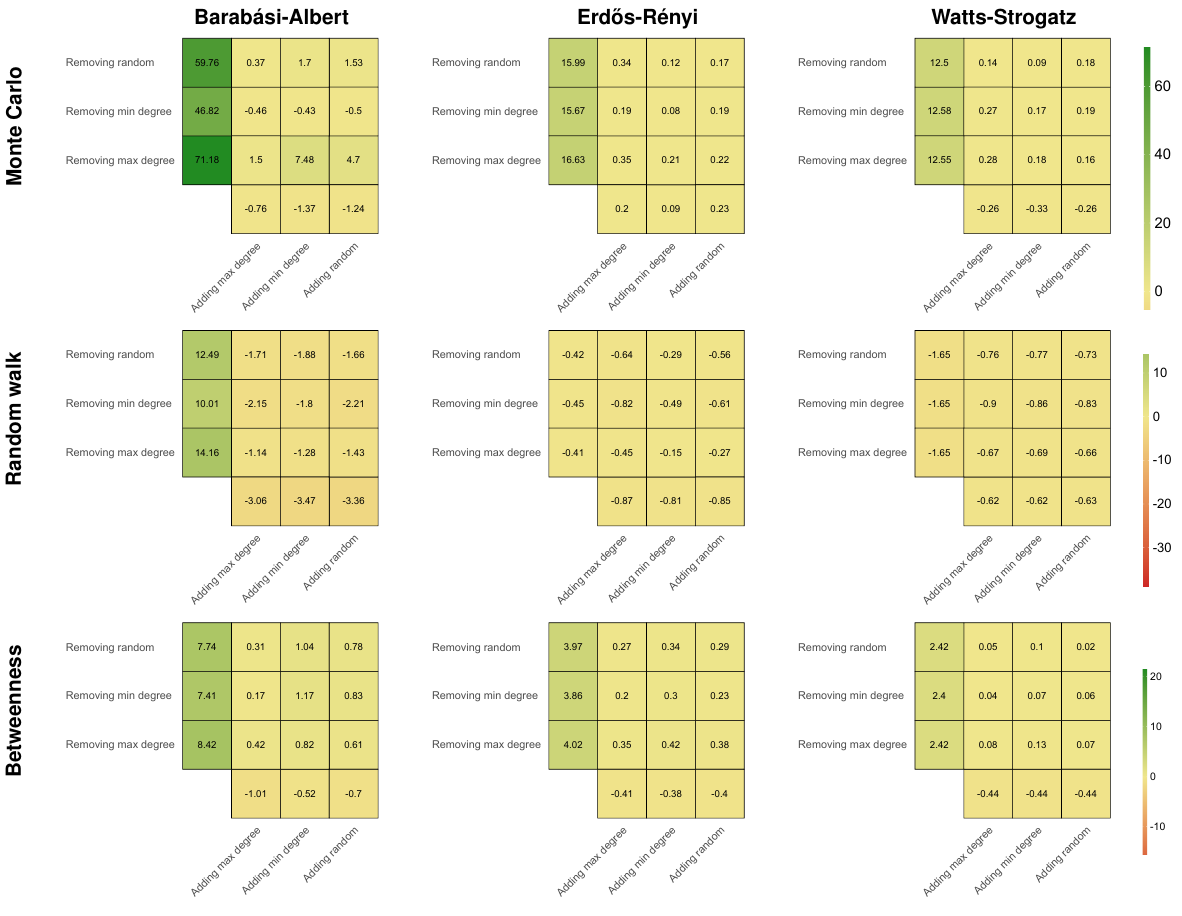}
\caption{
The same as Figure~\ref{fig:edge-simulation-mix-1}, except for the difference in the source detection algorithms being considered in the analysis.
}
\label{fig:edge-simulation-mix-2}
\end{figure}

\clearpage
\section{Comparison of Effectiveness of Adding Nodes and Modifying Edges}
\label{app:exchange}

In the main article, we compared the effectiveness of hiding the source of diffusion using two different types of network modifications---adding nodes and modifying edges. Figure~3 in the main article depicts the results of this analysis for the Eigenvector source detection algorithm. Here, we present analogical results for all other source detection algorithms considered in the study; see Figure~\ref{fig:exchange-appendix}.

As can be seen, the results for most source detection algorithm exhibit trends that are similar to those observed given the Eigenvector algorithm in the main article. More specifically, in the vast majority of cases, modifying one edge is significantly more effective than adding a single confederate to the network (as indicated by values greater than $1$ in the heatmaps). This tendency is particularly strong in the networks generated using the \BAn model, where for some networks it takes more than a hundred confederates to achieve the same effect as one edge. In contrast, in many networks generated using the Watts-Strogatz model, the addition of confederates is more effective than modifying edges.

\begin{figure}[tbh]
\centering
\setlength\tabcolsep{0pt}
\begin{tabular}{m{.05\textwidth}m{.77\textwidth}}
\rotatebox{90}{\hspace{5mm} \textbf{Degree}} &
\includegraphics[width=0.85\linewidth]{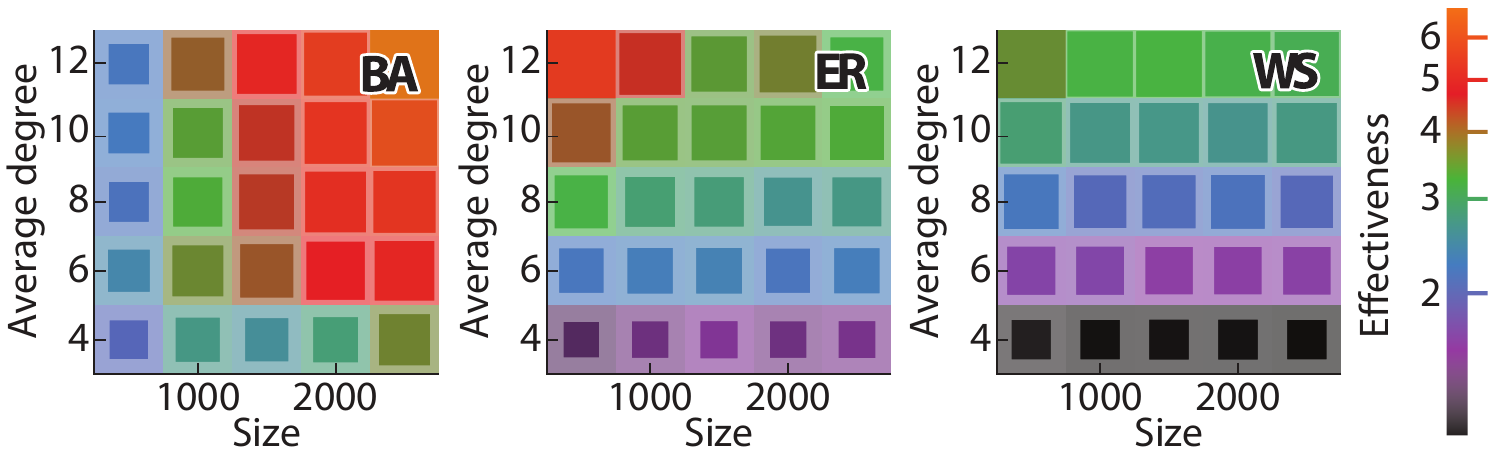} \\
\rotatebox{90}{\hspace{5mm} \textbf{Closeness}} &
\includegraphics[width=0.85\linewidth]{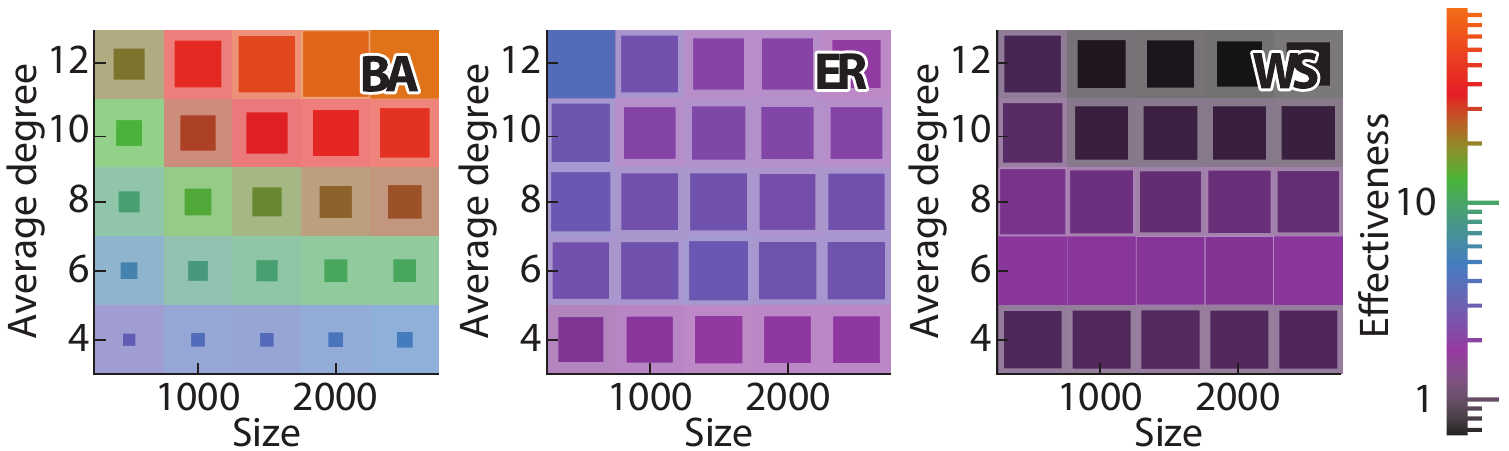} \\
\rotatebox{90}{\hspace{5mm} \textbf{Rumor}} &
\includegraphics[width=0.85\linewidth]{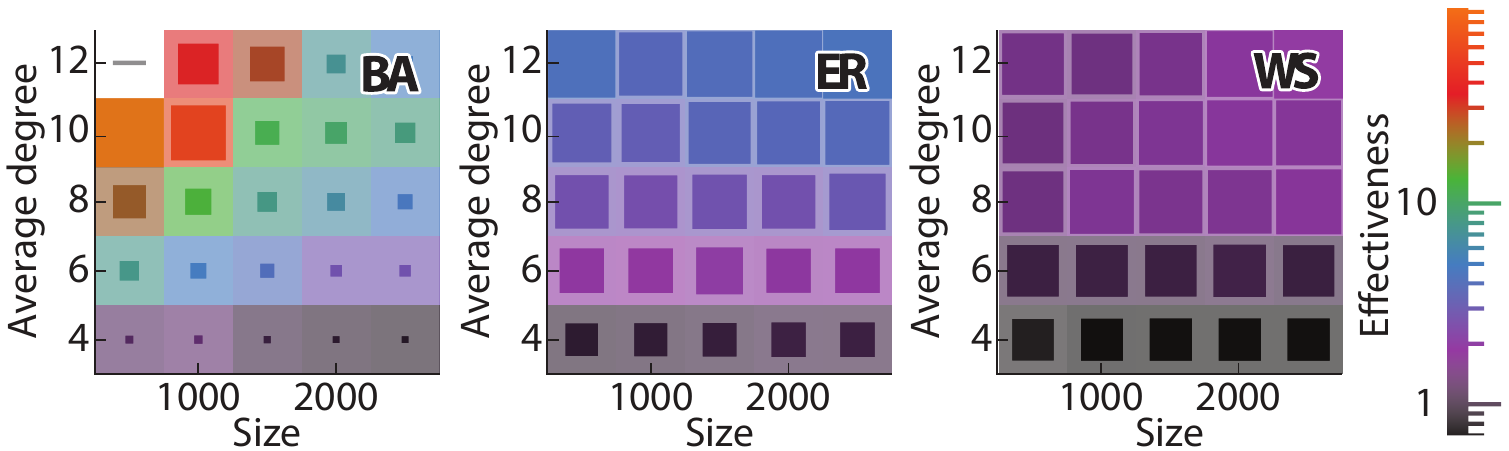} \\
\rotatebox{90}{\hspace{5mm} \textbf{Monte Carlo}} &
\includegraphics[width=0.85\linewidth]{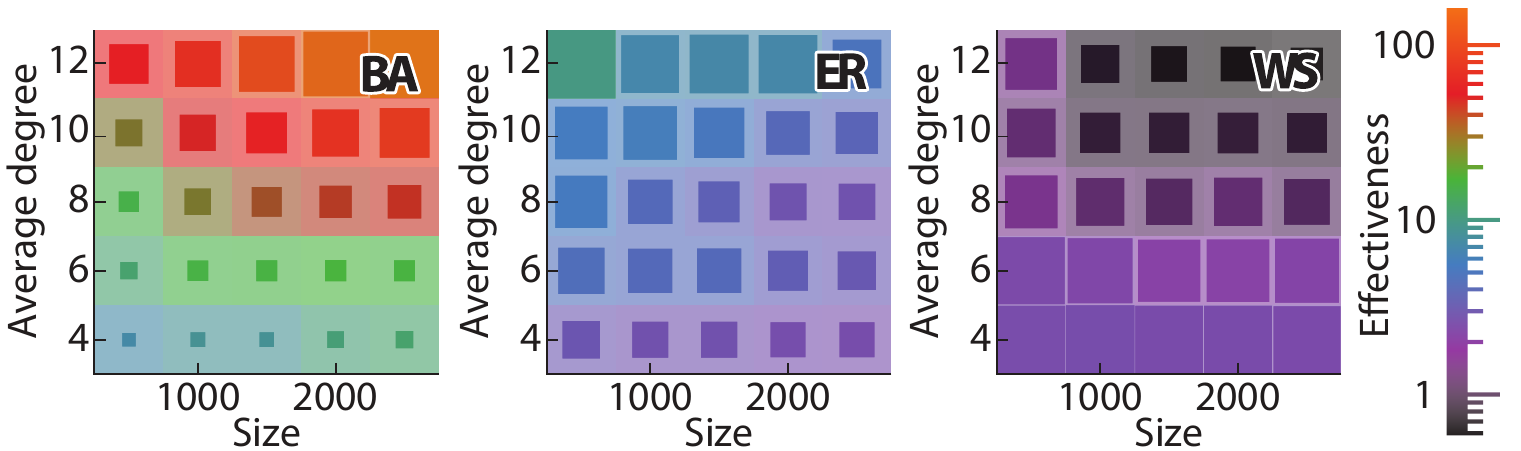} \\
\rotatebox{90}{\hspace{5mm} \textbf{Random walk}} &
\includegraphics[width=0.85\linewidth]{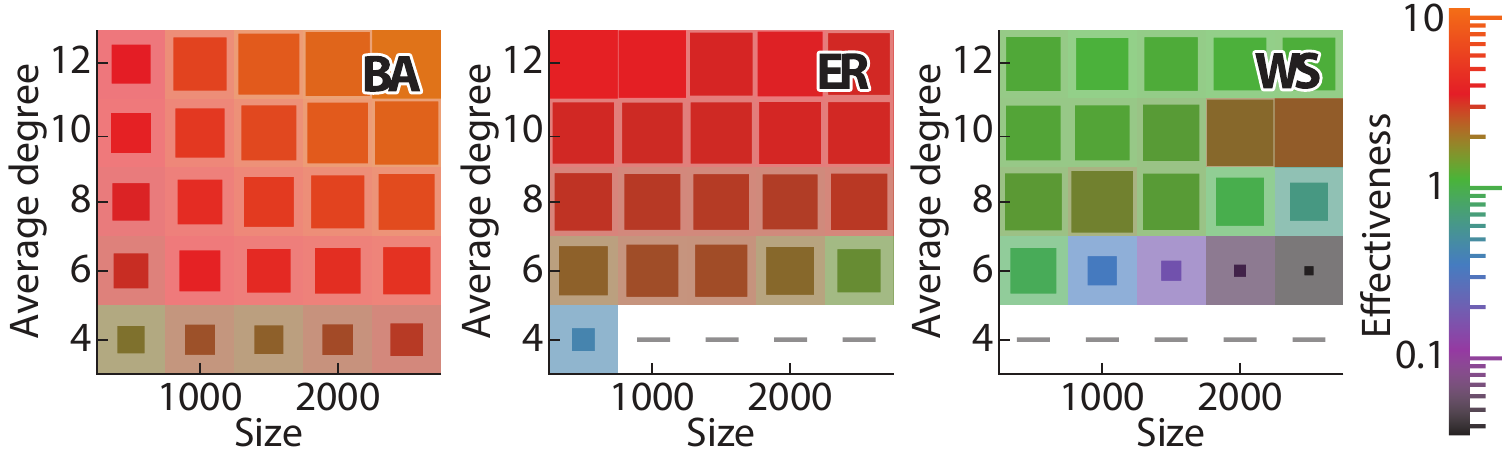} \\
\rotatebox{90}{\hspace{5mm} \textbf{Betweenness}} &
\includegraphics[width=0.85\linewidth]{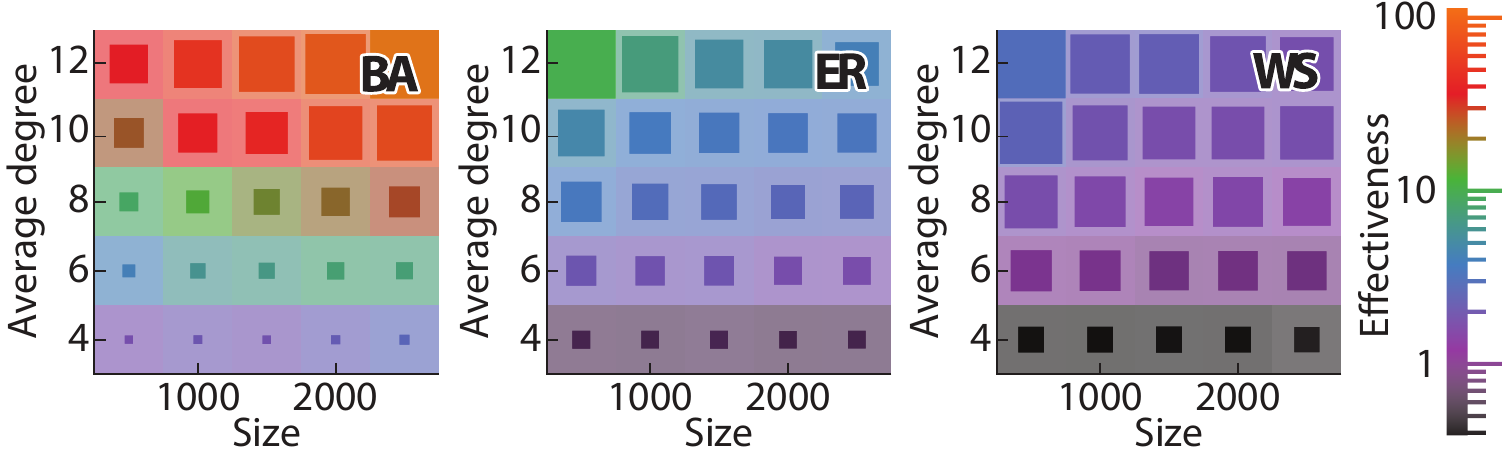} \\
\end{tabular}
\caption{
The same as Figure~3 in the main article, except for the difference in the source detection algorithms being considered in the analysis. For cells marked with minuses, at least one type of heuristic (adding confederates or modifying edges) does not reduce the evader's ranking.
}
\label{fig:exchange-appendix}
\end{figure}

\clearpage
\section{Hiding in Large Networks}
\label{app:simulation-large}

All the experiments presented thus far in the Appendices consider only $1,000$ nodes. Next, we evaluate the effectiveness of hiding the source of diffusion in much larger networks. Specifically, we consider random networks with $100,000$ nodes generated using the \BAn, the \ERn and the Watts-Strogatz models. As mentioned in the main article, due to the excessive amount of computations needed to handle such networks, we approximate the evader's ranking. The approximation is done by computing the ranking of the evader not among all nodes, but rather among $10,000$ nodes, consisting of the $5,000$ infected nodes with the greatest degrees and another $5,000$ infected nodes chosen uniformly at random from the remaining ones. Furthermore, in our approximation we do not consider the Betweenness and Random walk source detection algorithms, since their ranking cannot be efficiently computed for just a selected subset of nodes.

Figures from~\ref{fig:subo-simulation-line-large} to~\ref{fig:edge-simulation-heat-large} present the results of our simulations, while Figure~\ref{fig:large-before-bars} compares the effectiveness of the source detection algorithms themselves, before any attempts of hiding.
Generally, results show similar patterns to these computed for smaller networks. For networks generated using the \BAn model, the evader is relatively well-hidden even without executing any heuristics, whereas in other types of networks the evader is usually exposed at the beginning of the process, usually occupying the top position in the ranking. What is more, unlike the case with the \BAn model, applying the heuristics given the other network models has a negligible effect, with the evader's ranking decreasing by only a few positions in most cases.

\begin{figure}[tbh]
\centering
\includegraphics[width=.98\linewidth]{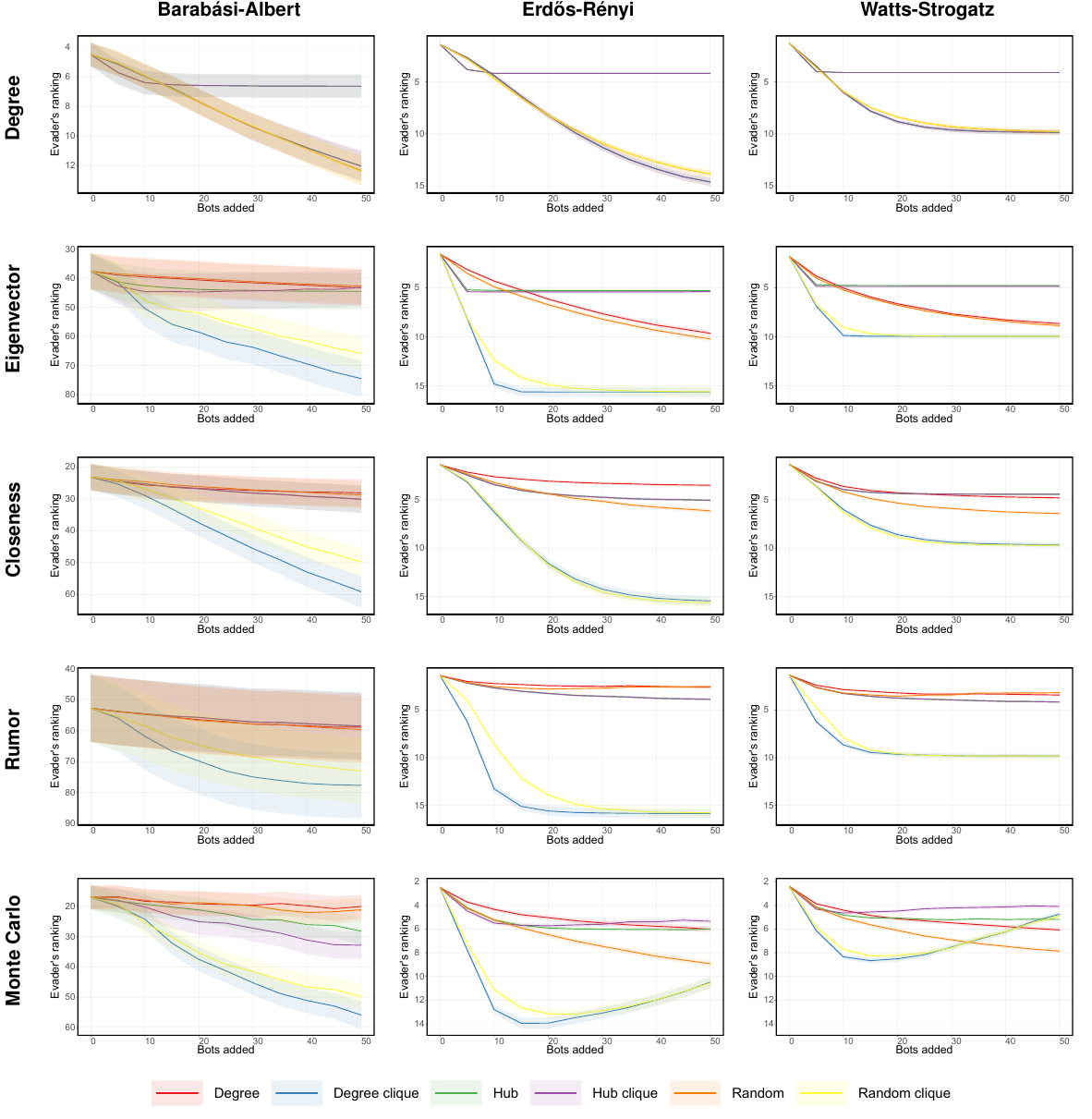}
\caption{
The same as Figure~\ref{fig:subo-simulation-line}, but for networks with $100,000$ nodes (instead of $1,000$) and only for the case where $3$ supporters are connected to each confederate.
}
\label{fig:subo-simulation-line-large}
\end{figure}

\begin{figure}[tbh]
\centering
\includegraphics[width=.98\linewidth]{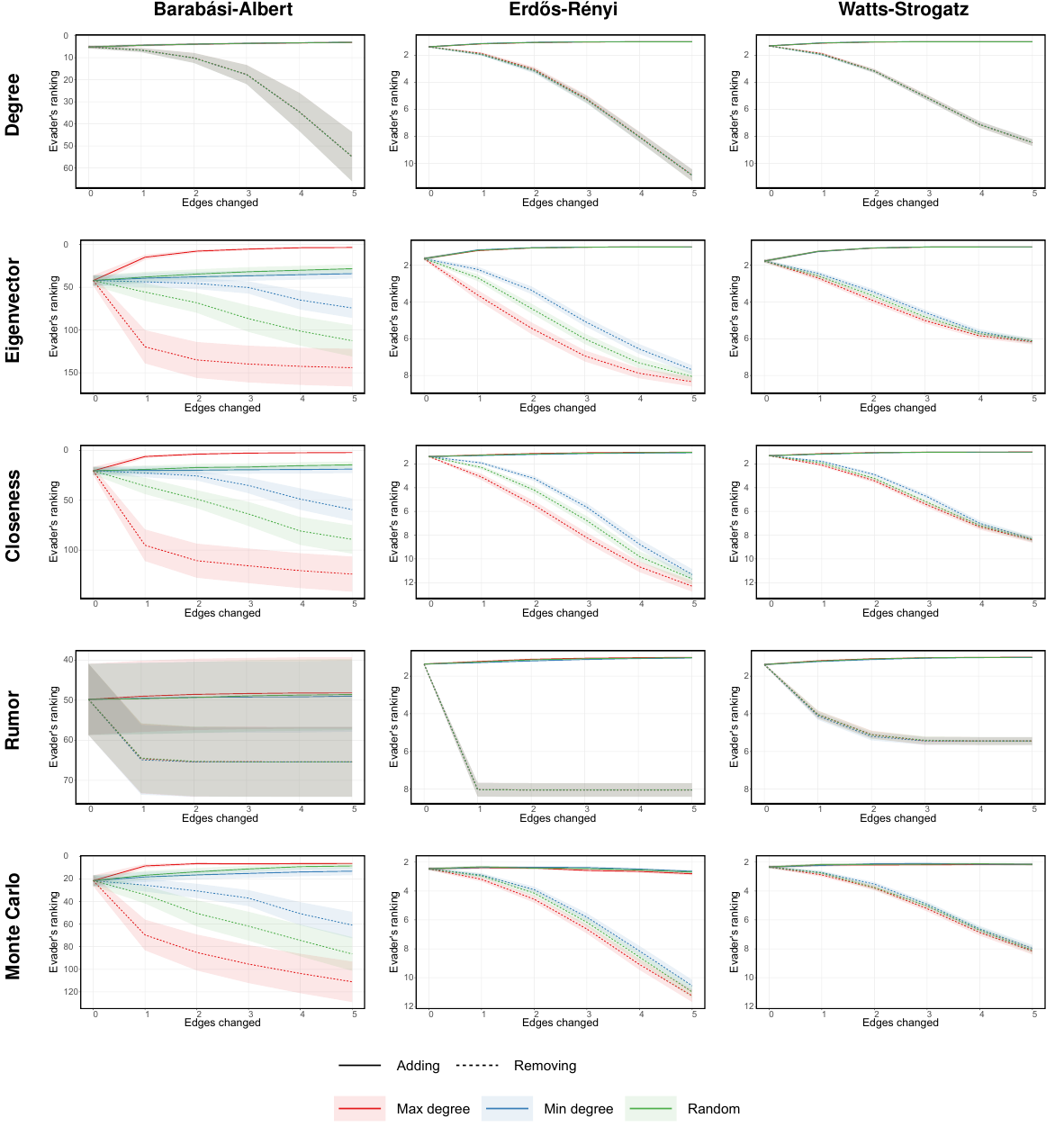}
\caption{
The same as Figure~\ref{fig:edge-simulation-line}, but for networks with $100,000$ nodes instead of $1,000$ nodes.
}
\label{fig:edge-simulation-line-large}
\end{figure}

\begin{figure}[tbh]
\centering
\includegraphics[width=\linewidth]{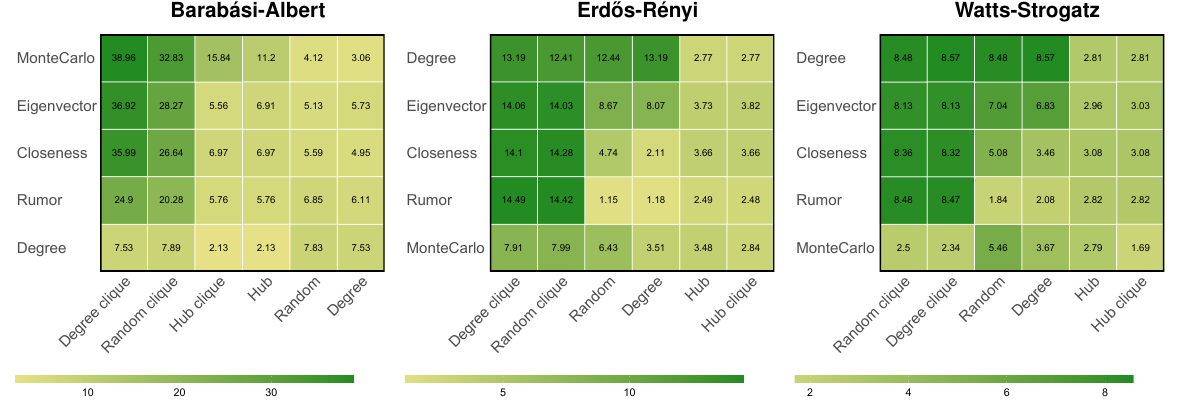}
\caption{
The same as Figure~\ref{fig:subo-simulation-heat}, but for networks with $100,000$ nodes (instead of $1,000$ nodes) and only for the case in which $3$ supporters are connected to each confederate.
}
\label{fig:subo-simulation-heat-large}
\end{figure}

\begin{figure}[tbh]
\centering
\includegraphics[width=\linewidth]{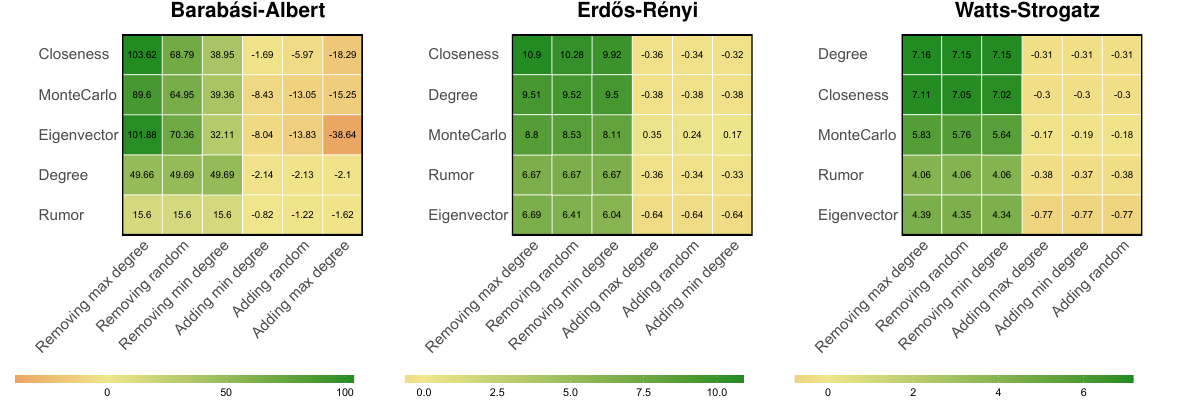}
\caption{
The same as Figure~\ref{fig:edge-simulation-heat}, but for networks with $100,000$ nodes instead of $1,000$ nodes.
}
\label{fig:edge-simulation-heat-large}
\end{figure}

\begin{figure}[tbh]
\centering
\includegraphics[width=\linewidth]{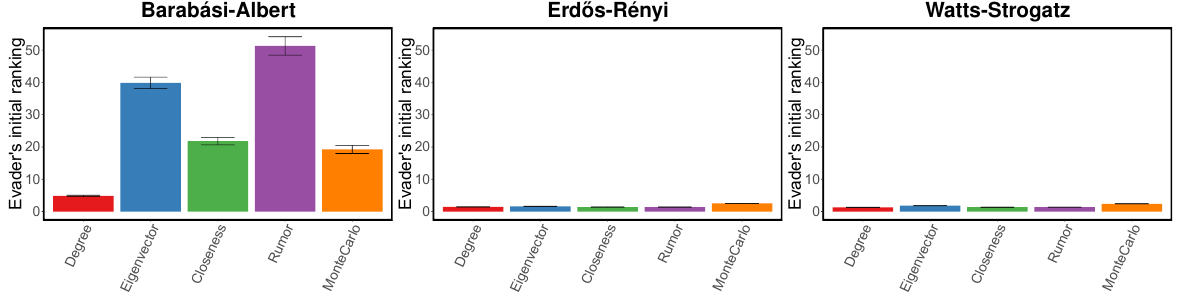}
\caption{
Comparison of the effectiveness of different source detection algorithms before the hiding process in random networks consisting of 100,000 nodes.
The x-axis corresponds to different source detection algorithms, while the y-axis corresponds to the initial ranking of the evader.
The error bars represent $95\%$ confidence intervals.
}
\label{fig:large-before-bars}
\end{figure}

\clearpage
\section{Hiding in Real-Life Networks}
\label{app:simulation-real}

We now evaluate the effectiveness of heuristics hiding the source of diffusion in real-life networks. In particular, we consider the following networks:

\begin{itemize}
\item \textbf{St Lucia}~\cite{thomas2016diffusion}---colocation network collected at the St Lucia campus of the University of Queensland using the WiFi network. The network consists of 302 nodes and 1149 edges. Nodes represent students, while an edge between any two students indicates that they were in the same location at the same time.

\item \textbf{Facebook}~\cite{leskovec2012learning}---fragment of the Facebook network, an ego-network of a student of an American university. The network consists of 333 nodes and 2523 edges. Nodes represent Facebook users, while an edge between any two users indicates that they are friends on Facebook.

\item \textbf{Copenhagen}~\cite{ahmed2010time}---retweet network of the tweets regarding the United Nations Climate Change conference held in Copenhagen in December 2009. The network consists of 761 nodes and 1029 edges. Nodes represent Twitter users, while an edge between two users indicates that at least one of them responded to a tweet by the other.

\item \textbf{Gnutella}~\cite{ripeanu2002mapping}---Gnutella peer-to-peer file sharing network snapshot from August 8, 2002. The network consists of 6299 nodes and 20776 edges. Nodes represent hosts, while an edge between two nodes indicates that the two hosts exchanged a file.

\item \textbf{PGP}~\cite{boguna2004models}---the network of users of the Pretty-Good-Privacy (PGP) algorithm for secure information exchange from 2004. The network consists of 10680 nodes and 24316 edges. Nodes represent  users of the algorithm, while an edge between two nodes indicates that they mutually signed their public keys within the protocol (i.e., they formed a trust relationship).

\item \textbf{Prostitution}~\cite{rocha2010information}---the giant component of a sexual contacts network between escorts and customers from a Brazilian online community. The network consists of 15810 nodes and 38540 edges. A node represents either an escort or a customers, while an edge between an escort and a customer indicates that the former reported a sexual contact with the latter.
\end{itemize}

The experimental procedure for the first three (smaller) real-life networks remains the same as for random networks with $1,000$ nodes in the main article, whereas the procedure for the last three (larger) real-life networks is the same as for random networks with $100,000$ nodes.
The results of our simulations for smaller networks are presented in Figures from~\ref{fig:subo-simulation-line-real-small} to~\ref{fig:edge-simulation-heat-real-small}, while the results for smaller networks are presented in Figures from~\ref{fig:subo-simulation-line-real-large} to~\ref{fig:edge-simulation-heat-real-large}. For an evaluation of the source detection algorithms themselves, before any attempts of hiding, see Figures~\ref{fig:real-small-before-bars} and~\ref{fig:real-large-before-bars}.

Regarding the smaller networks, results seem to be largely consistent with observation made for the random networks. One noticeable difference is the reduced effectiveness of heuristics adding nodes against rumor source detection algorithm in the St Lucia and the Facebook networks. It might be caused by the relatively large density of these two networks, causing most of the BFS trees computed by the rumor algorithm to be star-like. In case of the larger networks, the evader is hidden even without any strategic manipulations to a greater degree than in randomly generated networks, suggesting that correctly identifying the source of diffusion in the real world might be even harder than our simulations on synthetic data indicate.

\begin{figure}[tbh]
\centering
\includegraphics[width=.86\linewidth]{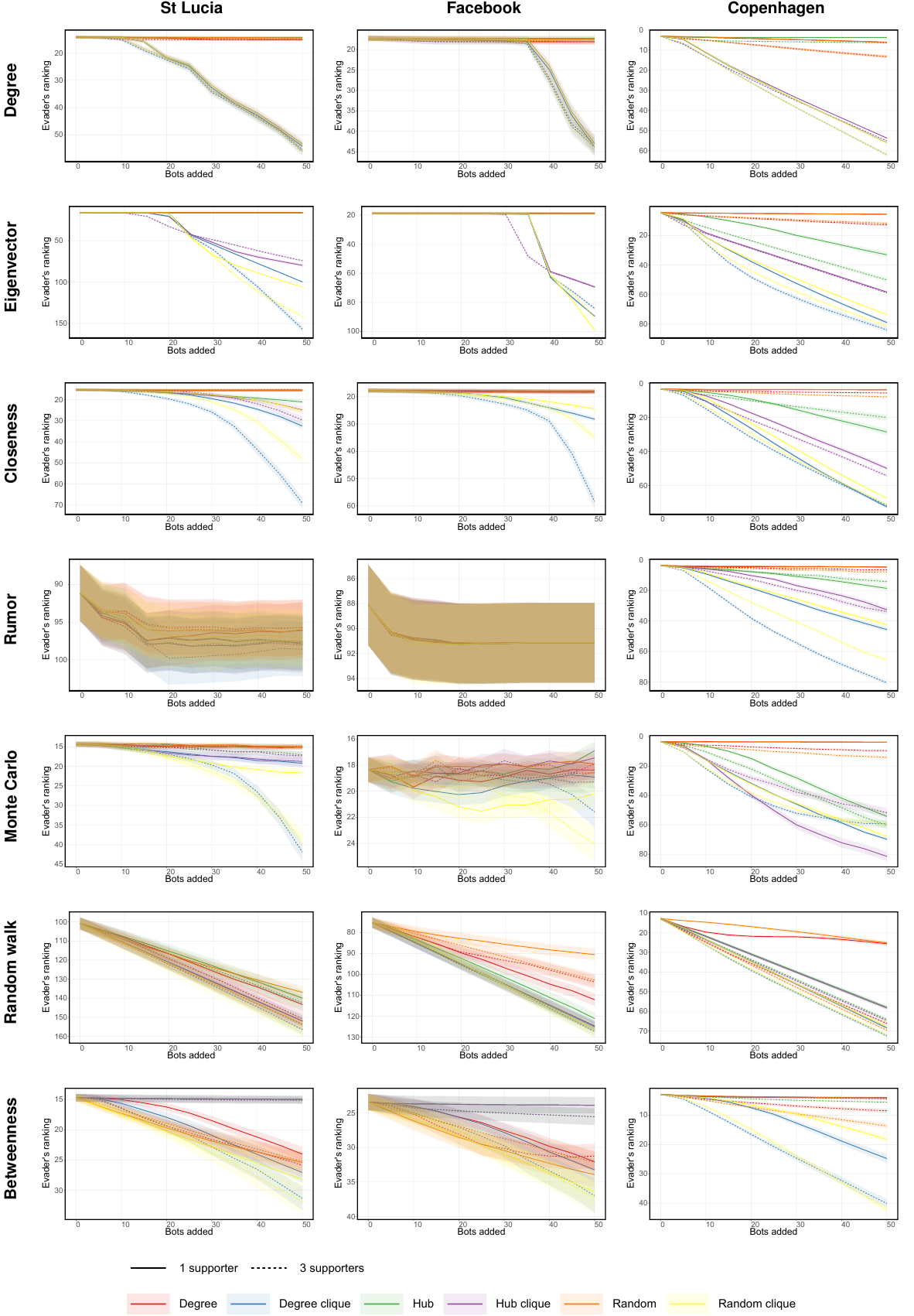}
\caption{
The same as Figure~\ref{fig:subo-simulation-line}, but for small real-life networks instead of randomly generated networks.
}
\label{fig:subo-simulation-line-real-small}
\end{figure}

\begin{figure}[tbh]
\centering
\includegraphics[width=.86\linewidth]{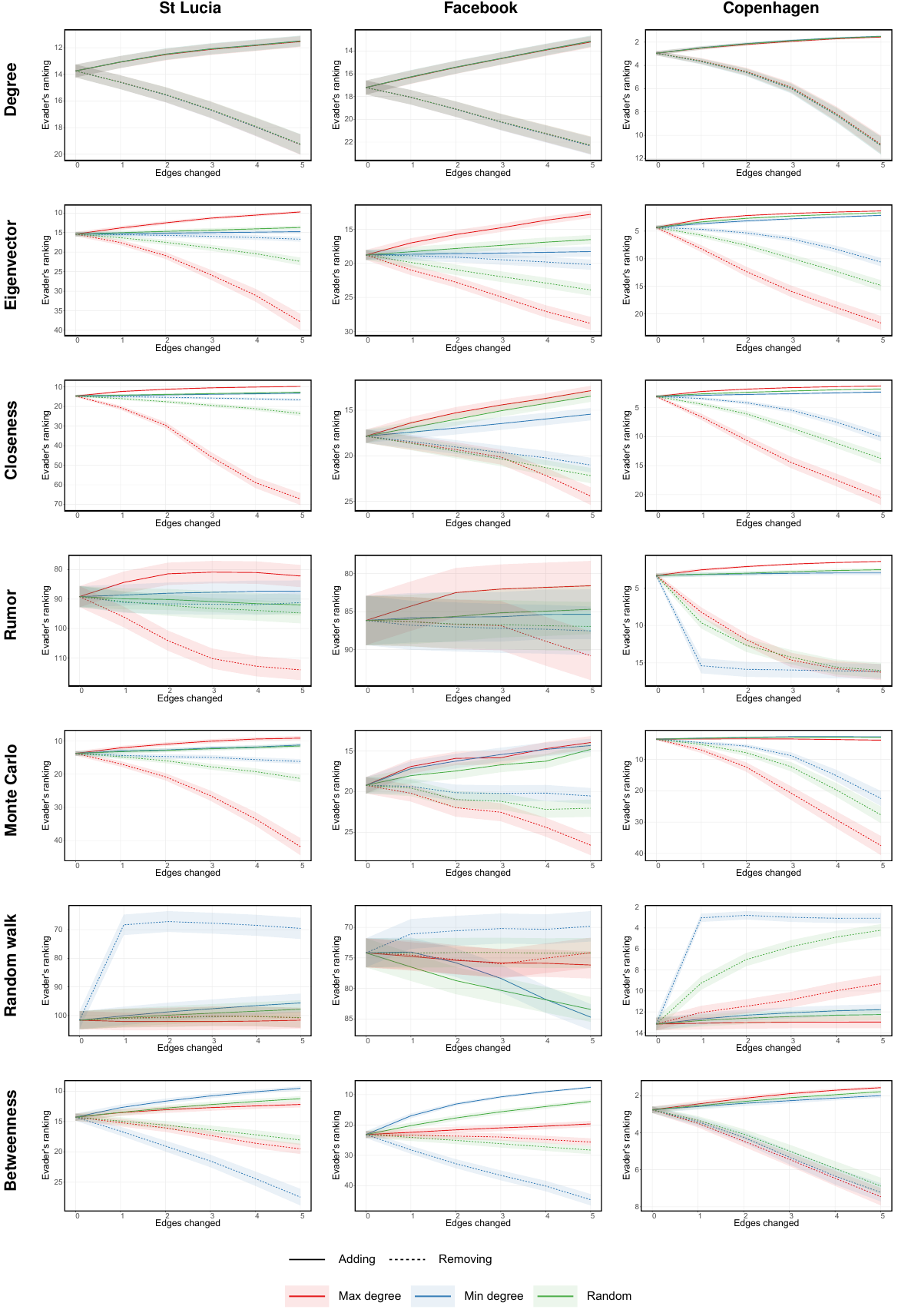}
\caption{
The same as Figure~\ref{fig:edge-simulation-line}, but for small real-life networks instead of randomly generated networks.
}
\label{fig:edge-simulation-line-real-small}
\end{figure}

\begin{figure}[tbh]
\centering
\includegraphics[width=\linewidth]{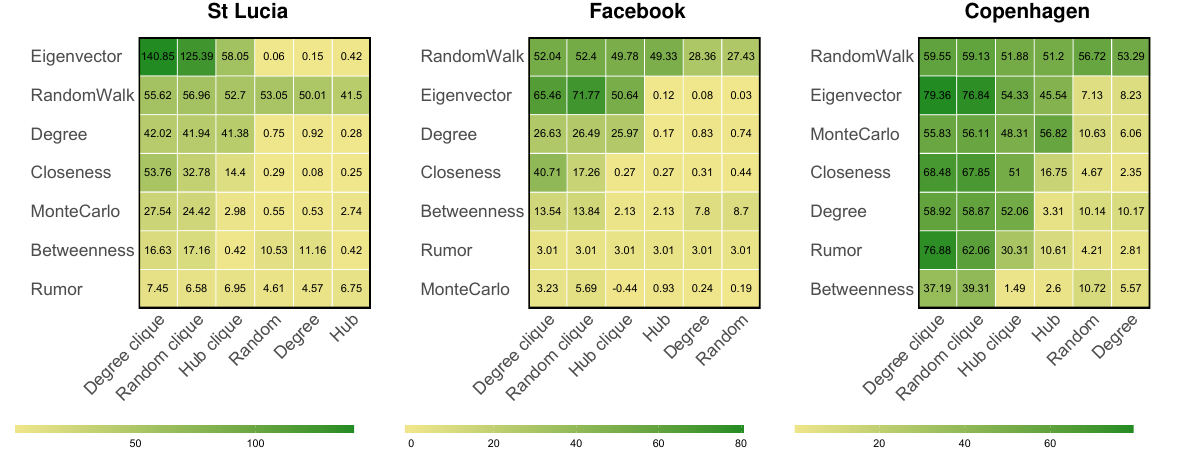}
\caption{
The same as Figure~\ref{fig:subo-simulation-heat}, but for small real-life networks instead of randomly generated networks.
}
\label{fig:subo-simulation-heat-real-small}
\end{figure}

\begin{figure}[tbh]
\centering
\includegraphics[width=\linewidth]{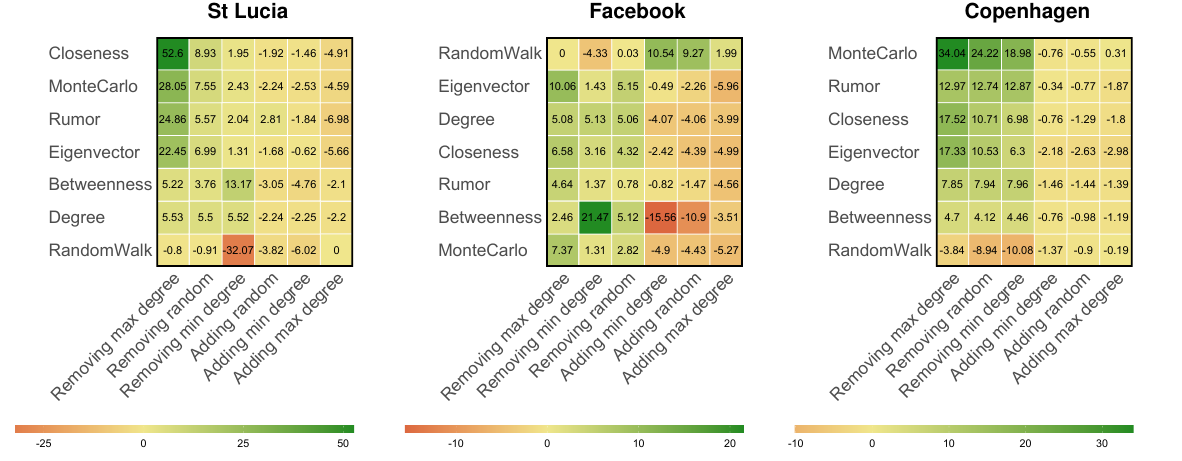}
\caption{
The same as Figure~\ref{fig:edge-simulation-heat}, but for small real-life networks instead of randomly generated networks.
}
\label{fig:edge-simulation-heat-real-small}
\end{figure}

\begin{figure}[tbh]
\centering
\includegraphics[width=.86\linewidth]{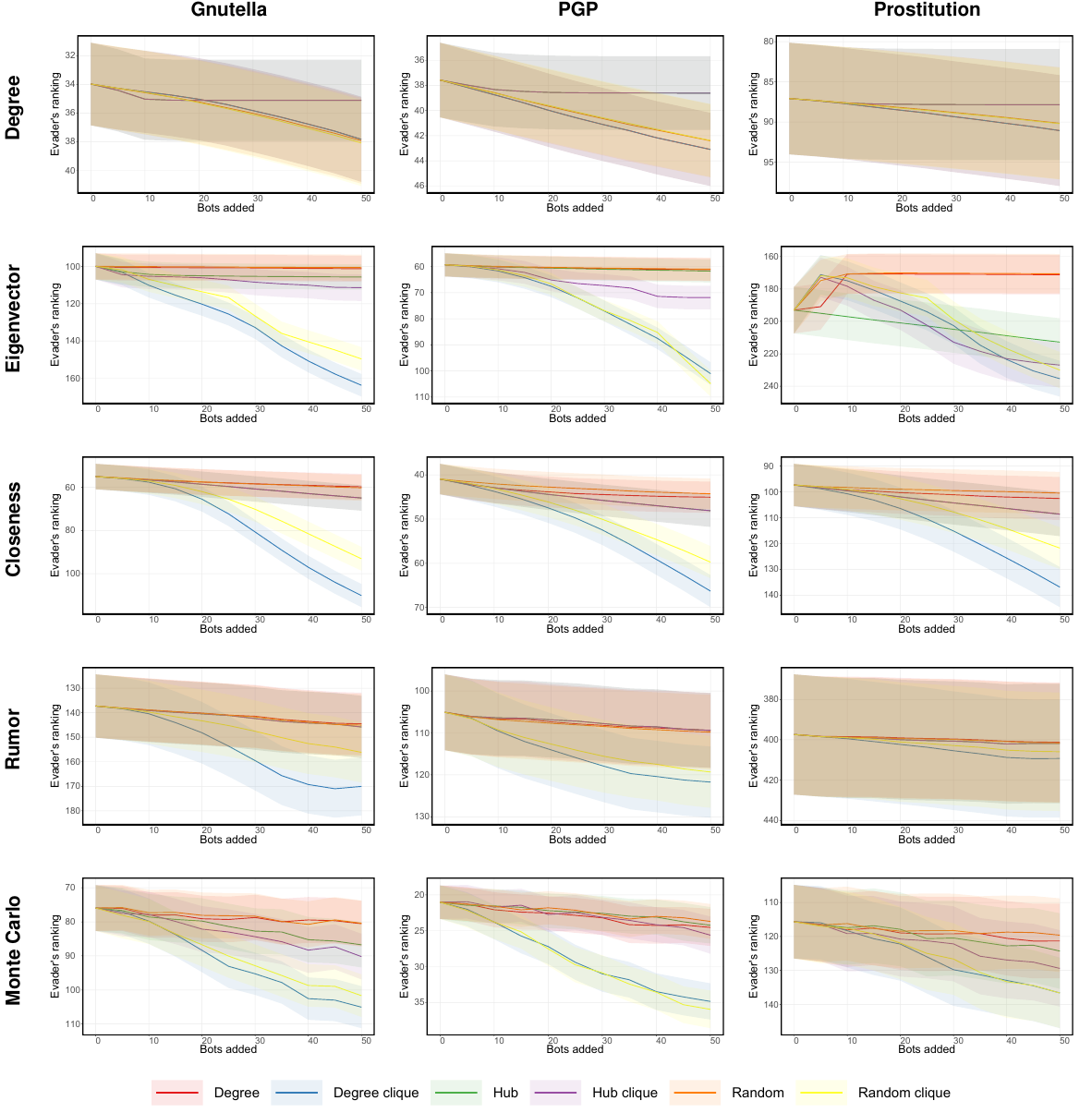}
\caption{
Same as Figure~\ref{fig:subo-simulation-line-large}, but for large real-life networks instead of randomly generated networks.
}
\label{fig:subo-simulation-line-real-large}
\end{figure}

\begin{figure}[tbh]
\centering
\includegraphics[width=.86\linewidth]{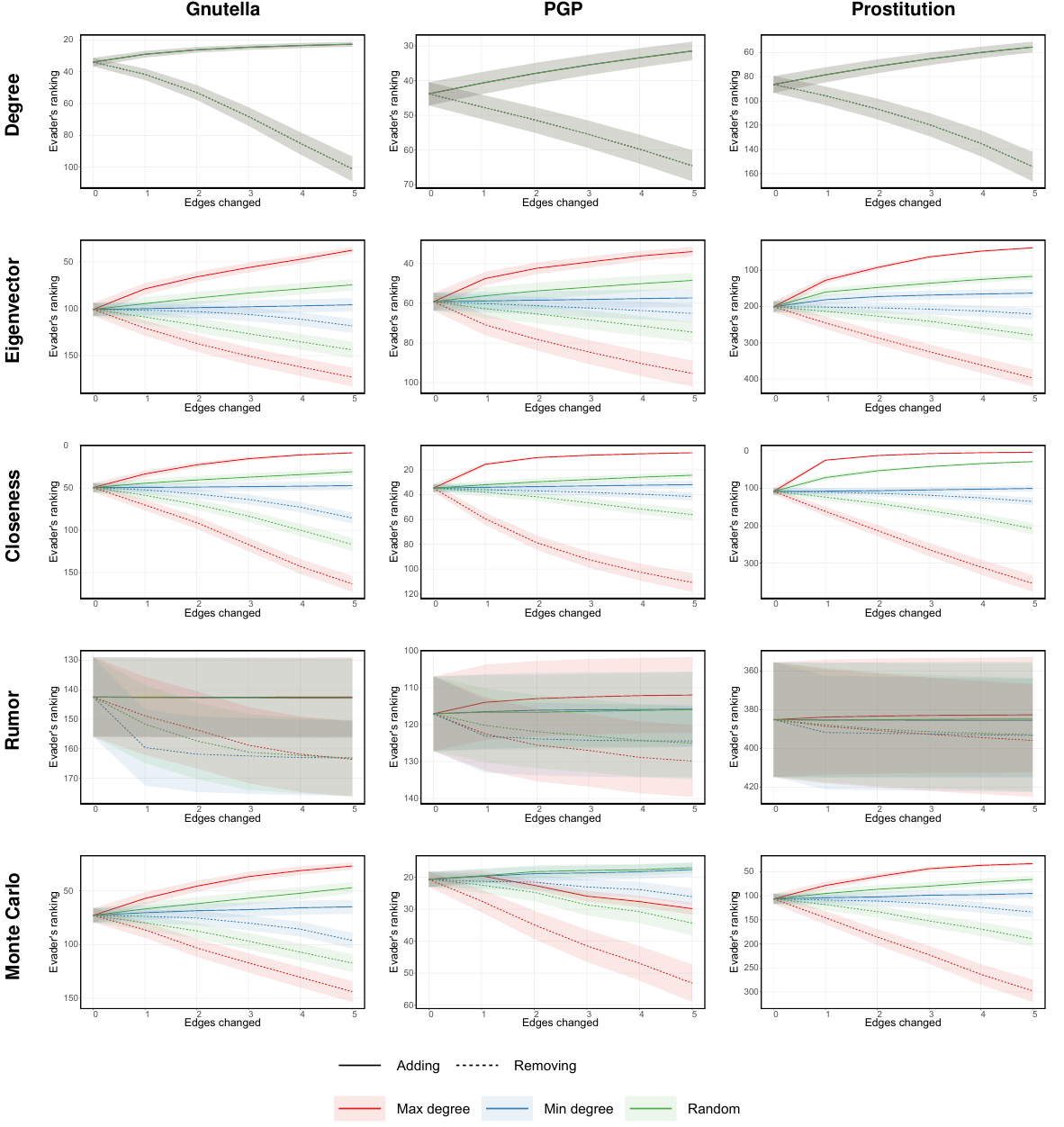}
\caption{
Same as Figure~\ref{fig:edge-simulation-line-large}, but for large real-life networks instead of randomly generated networks.
}
\label{fig:edge-simulation-line-real-large}
\end{figure}

\begin{figure}[tbh]
\centering
\includegraphics[width=\linewidth]{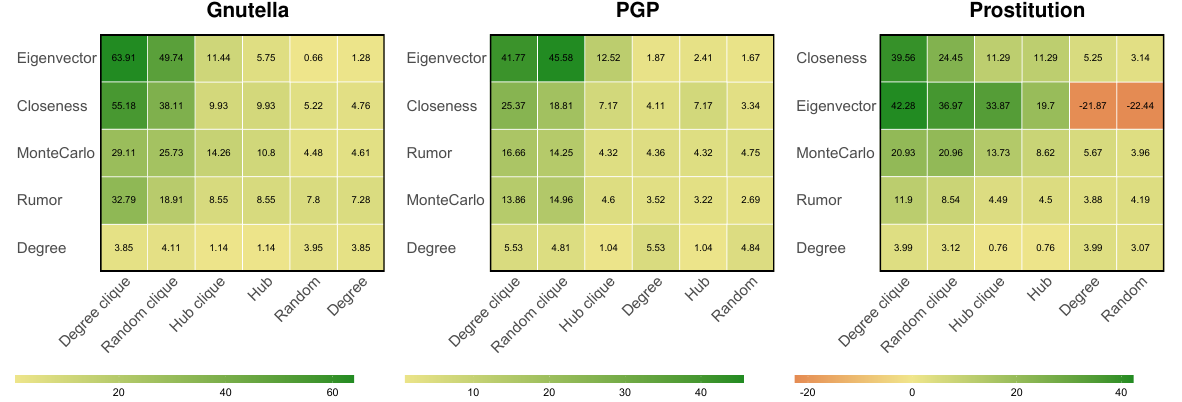}
\caption{
Same as Figure~\ref{fig:subo-simulation-heat-large}, but for large real-life networks instead of randomly generated networks.
}
\label{fig:subo-simulation-heat-real-large}
\end{figure}

\begin{figure}[tbh]
\centering
\includegraphics[width=\linewidth]{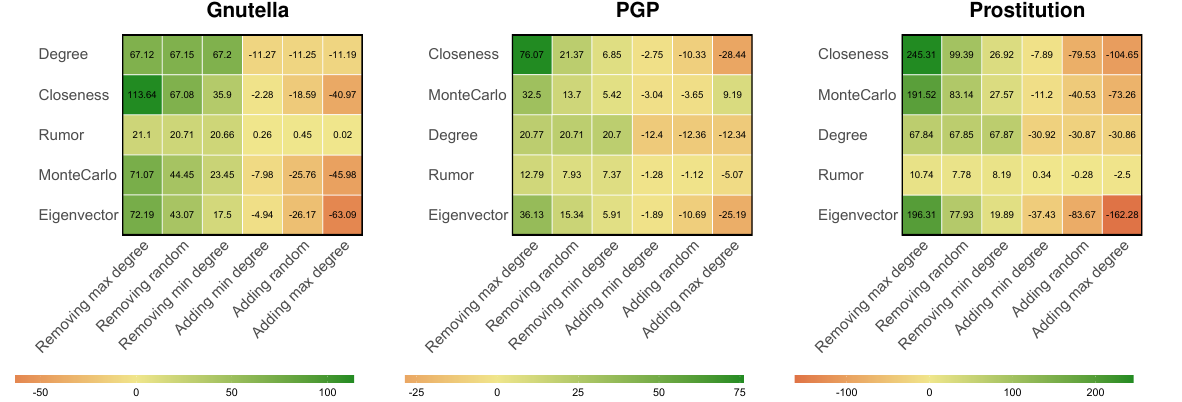}
\caption{
Same as Figure~\ref{fig:edge-simulation-heat-large}, but for large real-life networks instead of randomly generated networks.
}
\label{fig:edge-simulation-heat-real-large}
\end{figure}

\begin{figure}[tbh]
\centering
\includegraphics[width=\linewidth]{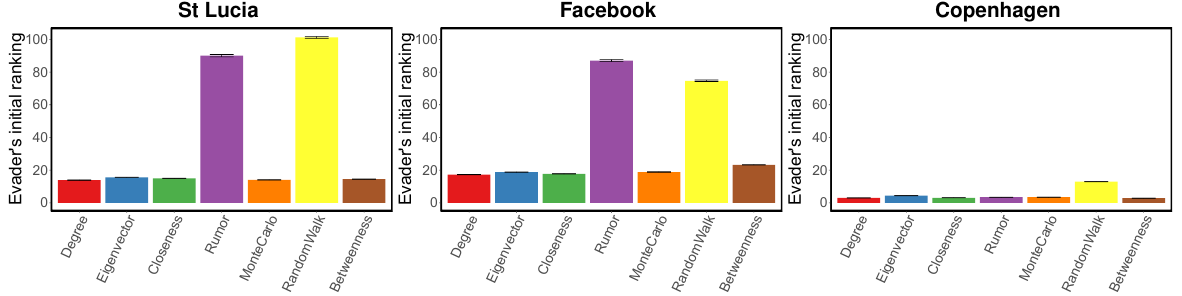}
\caption{
Comparison of the effectiveness of different source detection algorithms before the hiding process in three real-life networks, namely: St Lucia, Facebook, and Copenhagen.
The x-axis corresponds to different source detection algorithms, while the y-axis corresponds to the initial ranking of the evader.
The error bars represent $95\%$ confidence intervals.
}
\label{fig:real-small-before-bars}
\end{figure}

\begin{figure}[tbh]
\centering
\includegraphics[width=\linewidth]{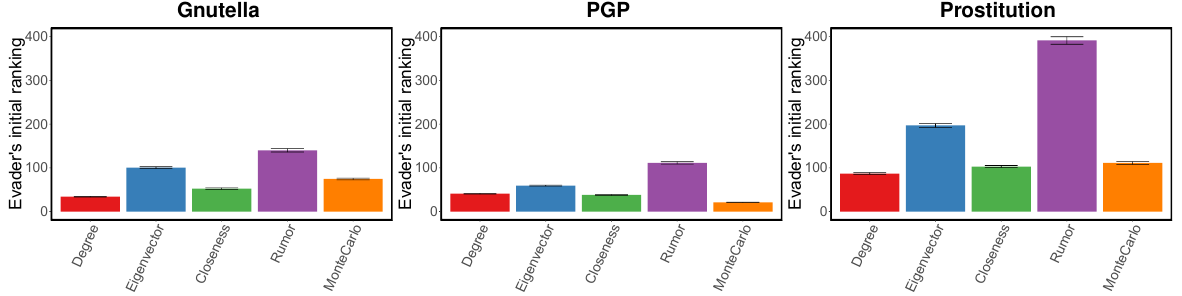}
\caption{
Comparison of the effectiveness of different source detection algorithms before the hiding process in three real-life networks, namely: Gnutella, PGP, and Prostitution.
The x-axis corresponds to different source detection algorithms, while the y-axis corresponds to the initial ranking of the evader.
The error bars represent $95\%$ confidence intervals.
}
\label{fig:real-large-before-bars}
\end{figure}

\clearpage
\section{Hiding the Source of a Real Cascade}
\label{app:real-cascades}

Previously in Section~\ref{app:simulation-real}, the networks that were analyzed were real, but the cascades were not, as they were generated using the Susceptible-Infected (SI) model. In this section, we provide a more realistic analysis where not only the network is real, but also the cascades themselves. This section analyzes a more realistic in which both the network and the cascade are real. To this end, we use the Twitter dataset created by M{\o}nsted \etal~\cite{monsted2017evidence}.
More specifically, the authors created 39 new Twitter accounts and gathered some followers. Then, using these accounts, the authors started spreading eight hashtags that previously did not exist on Twitter, namely \textit{\#getyourflushot}, \textit{\#highfiveastranger}, \textit{\#somethinggood}, \textit{\#HowManyPushups}, \textit{\#turkeyface}, \textit{\#SFThanks}, \textit{\#blackfridaystories}, and \textit{\#BanksySF}.
All tweets and retweets containing the new hashtags were then recorded, thereby capturing eight realistic information cascades which are included in the dataset.

Using the above dataset, we construct a network whereby a node corresponds to a Twitter account, and an edge between two nodes indicates that one is a follower or a friend of the other, or that one retweeted content posted the other.
The resulting network consists of $241,698$ nodes and $366,539$ edges.
For each hashtag listed above, we consider the account that first tweeted it as the source, and all accounts that tweeted or retweeted it at least once as the set of infected nodes.
We then attempt to hide the source from source detection algorithms, using the heuristics described in the main article.
We note that the Monte Carlo and random walk source detection algorithms are not included in this analysis, since they both require information about the diffusion model used to generate the set of infected nodes.

Given this network, Figure~\ref{fig:before-bars-fixed} depicts the average performance of source detection algorithms, taken over the eight cascades, while Figures~\ref{fig:line-fixed} and \ref{fig:heat-fixed} evaluate the effectiveness of our heuristics.
Altogether, these results (which are based on real-world cascades) validate our earlier findings (which were based on simulated cascades).
Specifically, our heuristics are capable of reducing the effectiveness of source detection algorithms. Moreover, the most effective heuristics on synthetic data tend to also be among the most effective ones on real data. Finally, adding edges can backfire and end up exposing the source even more, while adding nodes and removing edges rarely backfires (again, just like the patterns observed in our experiments with simulated cascades).

\begin{figure}[tbh]
\centering
\includegraphics[width=.4\linewidth]{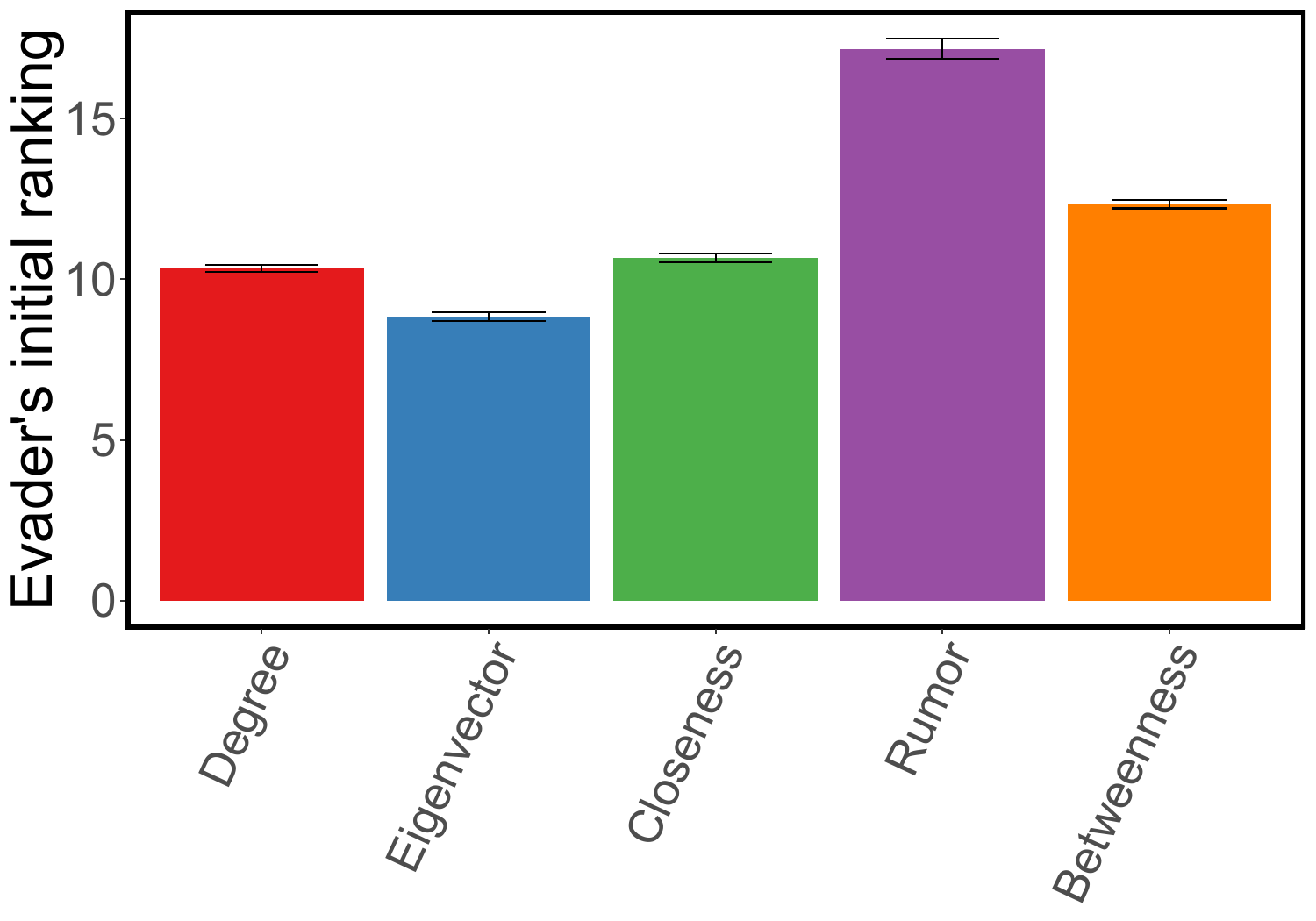}
\caption{
Comparing the effectiveness of source detection algorithms in real-life cascades before the hiding process is initiated.
The x-axis corresponds to different source detection algorithms, while the y-axis corresponds to the initial ranking of the evader (i.e., their ranking before attempting to hide).
Results are averaged over the eight cascades in the real-life dataset, with error bars representing $95\%$ confidence intervals.
}
\label{fig:before-bars-fixed}
\end{figure}

\begin{figure}[tbh]
\centering
\includegraphics[width=.79\linewidth]{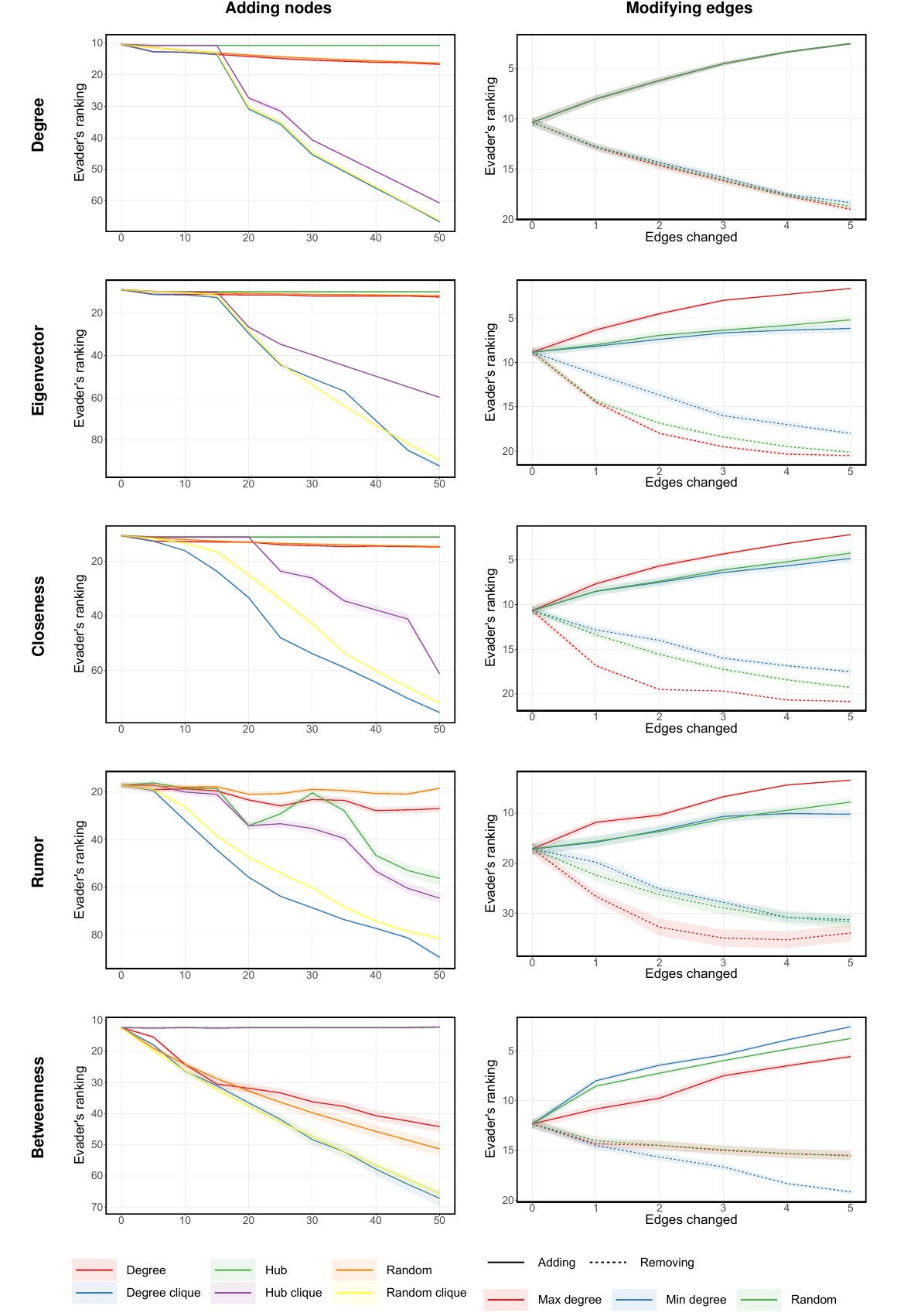}
\caption{
Results of hiding the source of real-life cascades. The y-axis corresponds to the ranking of the evader according to the source detection algorithm
(greater values indicate more efficient hiding),
while the x-axis corresponds to the number of nodes added to the network (left column) or the number of edges added to, or removed from, the network (right column). Each line correspond to different heuristics. Results are averaged over the eight cascades in the real-life dataset, with shaded areas representing $95\%$ confidence intervals.
}
\label{fig:line-fixed}
\end{figure}

\begin{figure}[tbh]
\centering
\includegraphics[width=.9\linewidth]{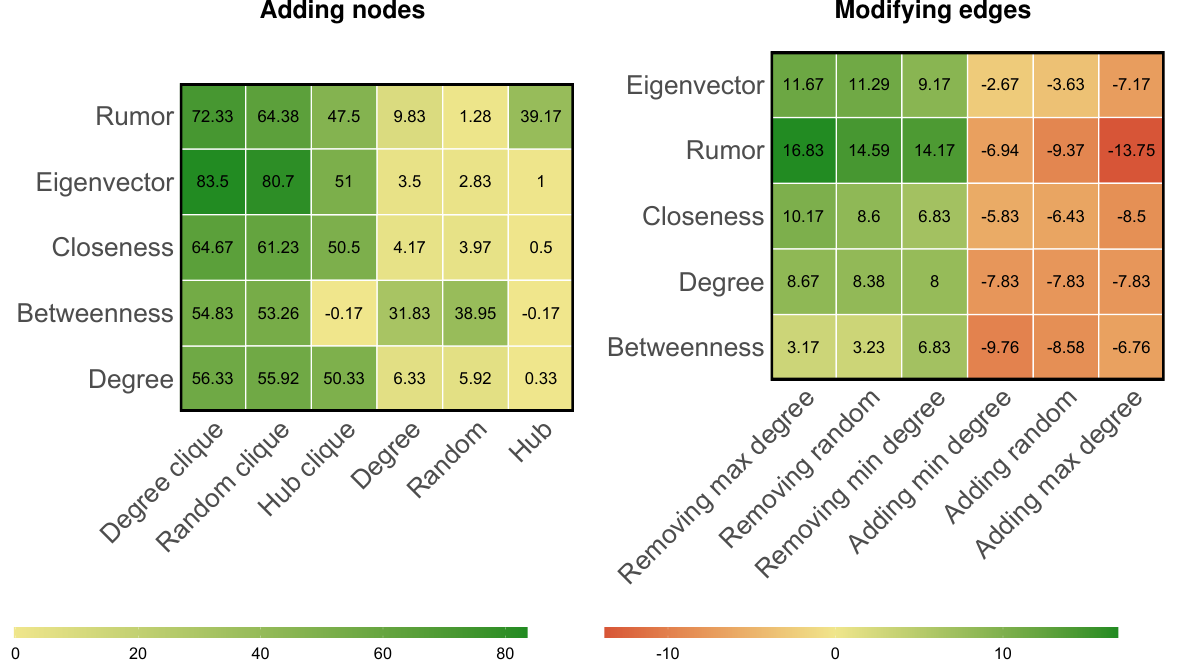}
\caption{
Results of hiding the source of real-life cascades. The y-axis of each heatmap corresponds to different source detection algorithms, whereas the x-axis corresponds to different heuristics. The value in each cell indicates the change in the evader's ranking according to the source detection algorithm after adding $50$ confederates to the network (left plot) or after adding or removing $5$ edges to the network (right plot). Rows and columns are sorted by average value, while the value in each cell is averaged over the eight cascades in the real-life dataset. 
}
\label{fig:heat-fixed}
\end{figure}

\clearpage
\section{Alternative Diffusion and Network Models}
\label{app:altmodels}

In the main article, we focused on a single diffusion model, namely SI, and on three network generation models, namely \BAn, \ERn, and Watts-Strogatz. In this section, we present simulations with alternative models of diffusion and network generation.
More specifically, we present the results for following three simulation settings:

\begin{itemize}

\item \textbf{Goel et al.~\cite{goel2016structural}}---Goel et al. attempted to create a formal model of the spreading of viral events on Twitter.
After evaluating several different models, they came to the conclusion that the one that produces results closest to reality runs the Susceptible-Infected-Recovered (SIR) diffusion in the Newman configuration networks.

The SIR diffusion model~\cite{kermack1927contribution} is similar to the SI model used in the remainder of our work, but an infected node can move to the Recovered state, after which it cannot infect others nor become infected again.
Goel et al. use a variation of the SIR model where an infected node moves to the Recovered state one round after it gets infected, and where the probability of infection is equal to $\frac{1}{2d^*}$, with $d^*$ being the average degree in the network.
The authors also assume that, after the diffusion process terminates, the information about whether a node was infected or not is available to the party analyzing the network (i.e., it is possible to discern between the nodes that are in the Susceptible state and those that are in the Recovered state).
In our simulations we use the exact same parameterization of this diffusion model.

As for the Newman configuration model~\cite{newman2003structure}, it is a model that generates networks with a given degree distribution.
Goal et al. used networks with degrees drawn from the power law distribution $P(d) \sim d^{-\alpha}$, where $\alpha=2.3$.
We use the same model of generating networks in our simulations.

\item \textbf{Kleinberg~\cite{kleinberg2007cascading}}---Kleinberg's model of diffusion is based on the assumption that in order for a node to adopt a certain phenomenon, a given number of the node's neighbors must adopt it first.
Before the diffusion starts, Kleinberg randomly generates for every node $v$ a threshold $\theta_v \in [0,1]$ from a uniform distribution.
At the beginning of the diffusion only the source node is active, while the rest of the nodes are inactive.
A given inactive node becomes active when the proportion of its active neighbors becomes greater than $\theta_v$, i.e., when:
\[
\frac{|\{w \in N(v) : w \in I\}|}{d(v)} > \theta_v
\]
where $I$ is the set of active nodes.
The diffusion ends when no new nodes can become active.
We run the simulations for this diffusion model in the same networks as for Goel et al.~\cite{goel2016structural}.

\item \textbf{Golub and Jackson~\cite{golub2012does}}---Golub and Jackson created the ``islands'' model that generates networks with a strong community structure.
The set of nodes is divided into $m$ equally-sized groups.
Every node forms a connection with another node from the same group with probability $p_{same}$, and with any other nodes from a different group with probability $p_{diff} \leq p_{same}$.
In our simulations we consider the islands networks with $10$ groups of size $100$ each, and with the average degree $4$, where each node is on average connected with $3$ others from the same group and with $1$ node from a different group.
With these parameters, we obtain networks of the same size and average degree as those used in the experiments presented in the main article. As for the diffusion model, we use the same one that was used in the main article.

\end{itemize}

The results of our simulations with alternative diffusion and network generation models are presented in Figures from~\ref{fig:subo-simulation-line-altmodel} to \ref{fig:edge-simulation-heat-altmodel}. 
When we compare them with analogical figures in Sections~\ref{app:simulation-subos} and~\ref{app:simulation-edges}, we find that they exhibit similar trends. This similarity may seem counterintuitive at first, especially in the case of the Goel et al. simulations and the Kleinberg simulations, since they not only use different networks, but also use different diffusion models.
However, it should be noted that the source detection algorithms studied in this work either require perfect information about the diffusion model, in which case they can easily adapt to any diffusion model, or do not incorporate any information whatsoever about the diffusion model, since they only analyze the connections between the infected nodes. These findings suggest that the source detection algorithms considered in our simulations are suitable for any diffusion model that preserves information about whether a node was infected in the past.

\begin{figure}[tbh]
\centering
\includegraphics[width=.8\linewidth]{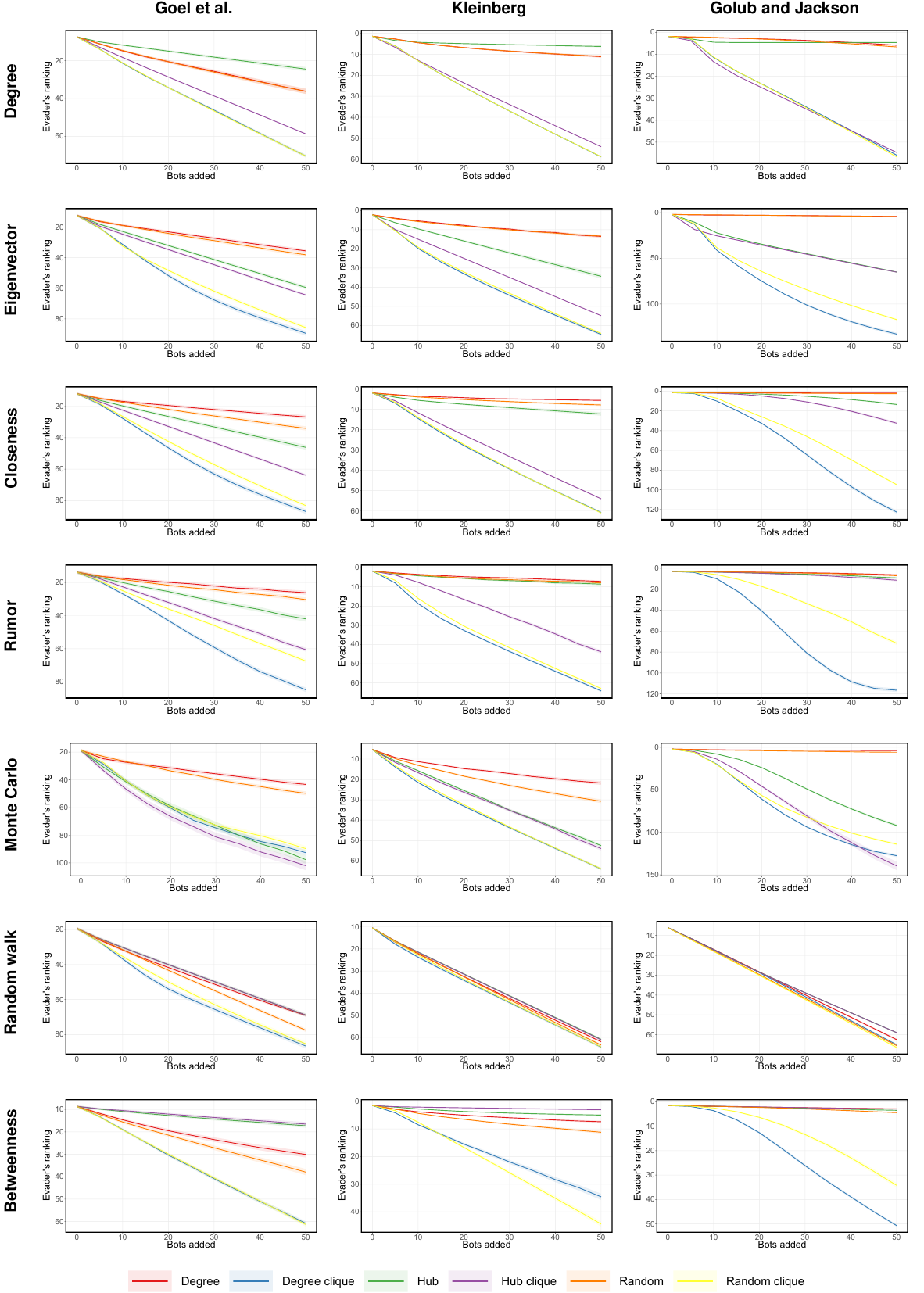}
\caption{
Results of hiding the source of diffusion by adding confederates for alternative models of diffusion and network generation. The y-axis corresponds to the ranking of the evader according to the source detection algorithm (greater values indicate more efficient hiding), while the x-axis corresponds to the number of nodes added to the network. Each color corresponds to a different heuristic, with each confederate being connected to three supporters. Shaded areas represent $95\%$ confidence intervals.
}
\label{fig:subo-simulation-line-altmodel}
\end{figure}

\begin{figure}[tbh]
\centering
\includegraphics[width=.8\linewidth]{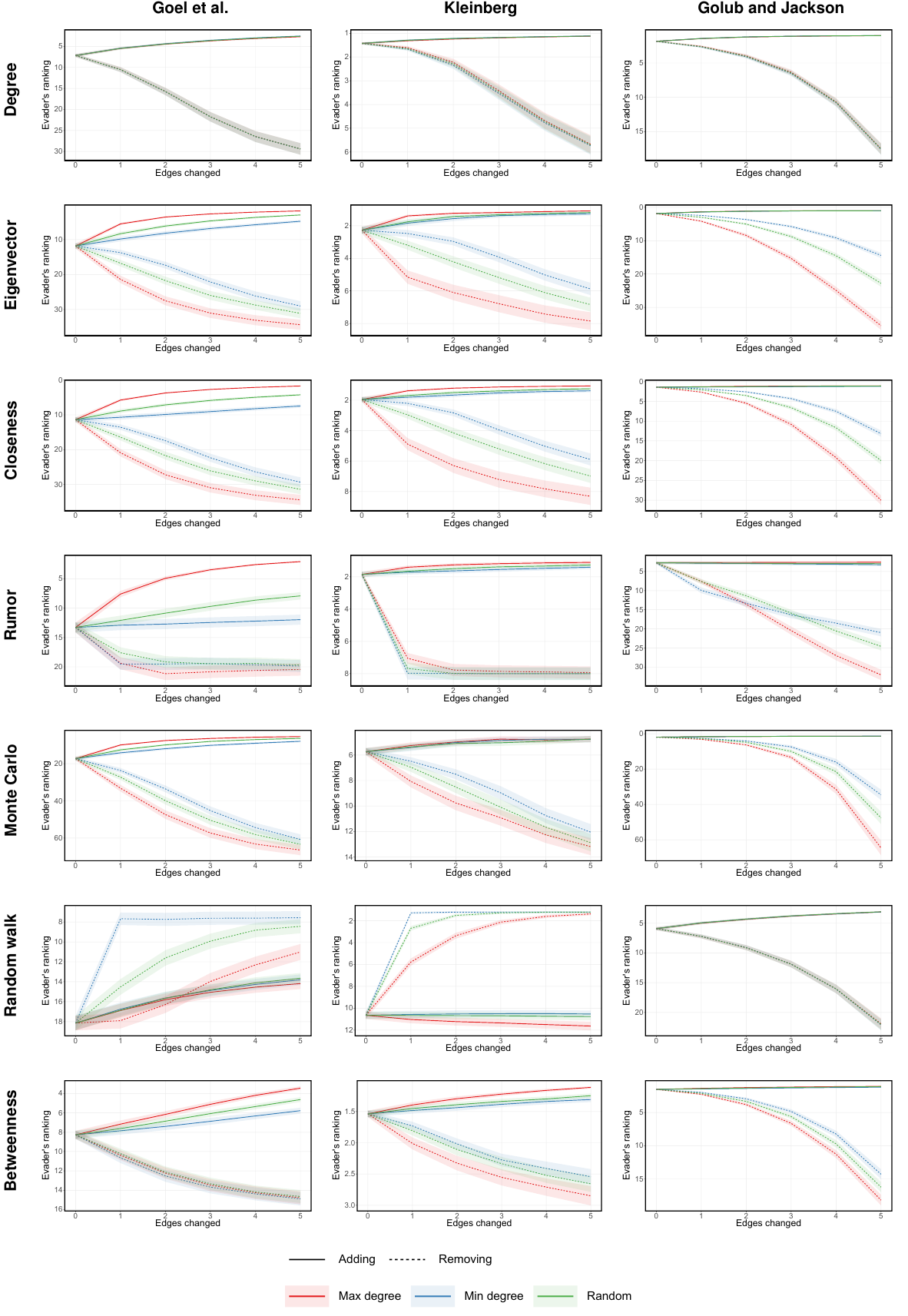}
\caption{
Results of hiding the source of diffusion by modifying edges for alternative models of diffusion and network generation. The y-axis represents the evader's ranking according to the source detection algorithm (greater value indicates more effective hiding); the x-axis corresponds to the number of edges added to, or removed from, the network. Each color corresponds to a different way of choosing edges, while each line type corresponds to either adding or removing. Shaded areas represent $95\%$ confidence intervals.
}
\label{fig:edge-simulation-line-altmodel}
\end{figure}

\begin{figure}[tbh]
\centering
\includegraphics[width=\linewidth]{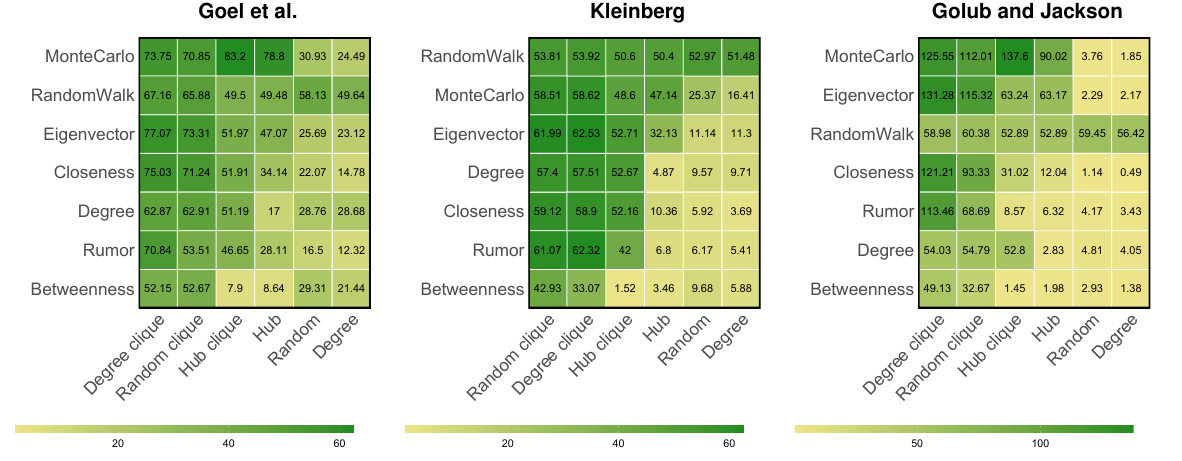}
\caption{
Results of simulations with hiding the source of diffusion by adding nodes for alternative models of diffusion and network generation. The y-axis of each heatmap corresponds to different source detection algorithms, whereas the x-axis corresponds to different heuristics. The value in each cell indicates the change in the evader's ranking according to the source detection algorithm after adding $50$ confederates to the network using the heuristic. Rows and columns are sorted by average value.
}
\label{fig:subo-simulation-heat-altmodel}
\end{figure}

\begin{figure}[tbh]
\centering
\includegraphics[width=\linewidth]{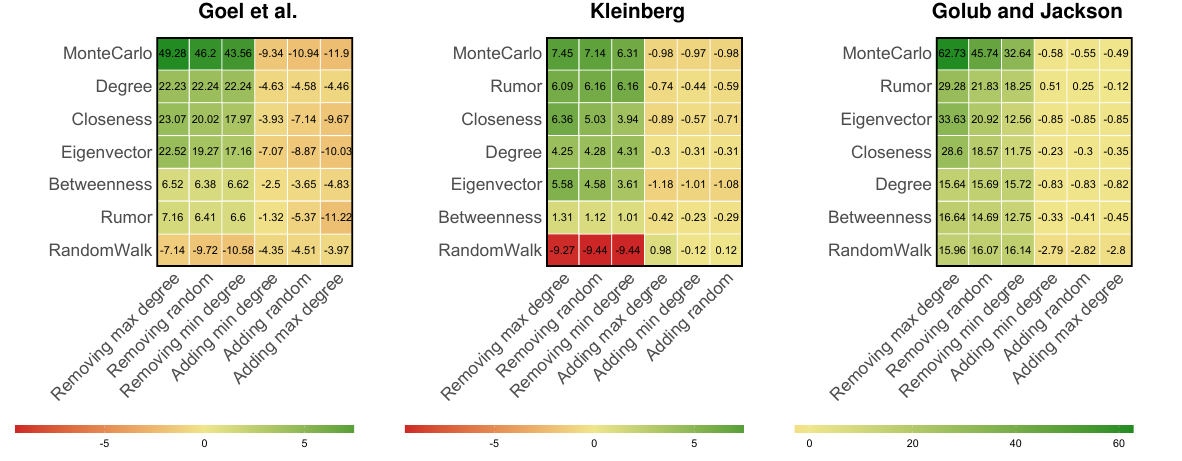}
\caption{
Results of hiding the source of diffusion by modifying edges for alternative models of diffusion and network generation. In each heatmap, rows correspond to different source detection algorithms, while columns correspond to different heuristics. The value in each cell indicates the change in the evader's ranking according to the source detection algorithm as a result of adding or removing $5$ edges to the network, depending on the heuristic. Positive values indicate that the evader became more hidden, with greater values indicated a more effective disguise. In contrast, negative values indicate that the evader became less hidden. Rows and columns are sorted by average value.
}
\label{fig:edge-simulation-heat-altmodel}
\end{figure}

\clearpage
\section{Effects of Hiding the Nearby Nodes}
\label{app:distance}

In this section, we investigate how the evader's ranking is affected when another node in the evader's direct network vicinity runs the hiding process. To this end, we run the basic version of our experiments (i.e., the one described in the main article) but this time it is not the evader who runs the heuristic, but rather one of the infected nodes at distance $5$ or less from the evader. Here, we evaluate the two heuristics that proved to be the most effective in our basic experiments, i.e., the one that adds nodes and the one that modifies edges.

Figure~\ref{fig:distance-line} presents the results of our simulations. As expected, the evader is best hidden when they are the ones running the heuristic. Moreover, when another node $v$ runs the heuristic, the evader also becomes hidden, but this effect tends to decrease as the distance between $v$ and the evader increases, apart from a few exceptions.
In should be noted that running the heuristic that adds nodes is particularly ineffective when performed by the evader's neighbors (as indicated by a sharp increase in the plots), because the shortest paths from the nodes added in this manner to the rest of the network have a high probability of running though the evader.

\begin{figure}[tbh]
\centering
\includegraphics[width=\linewidth]{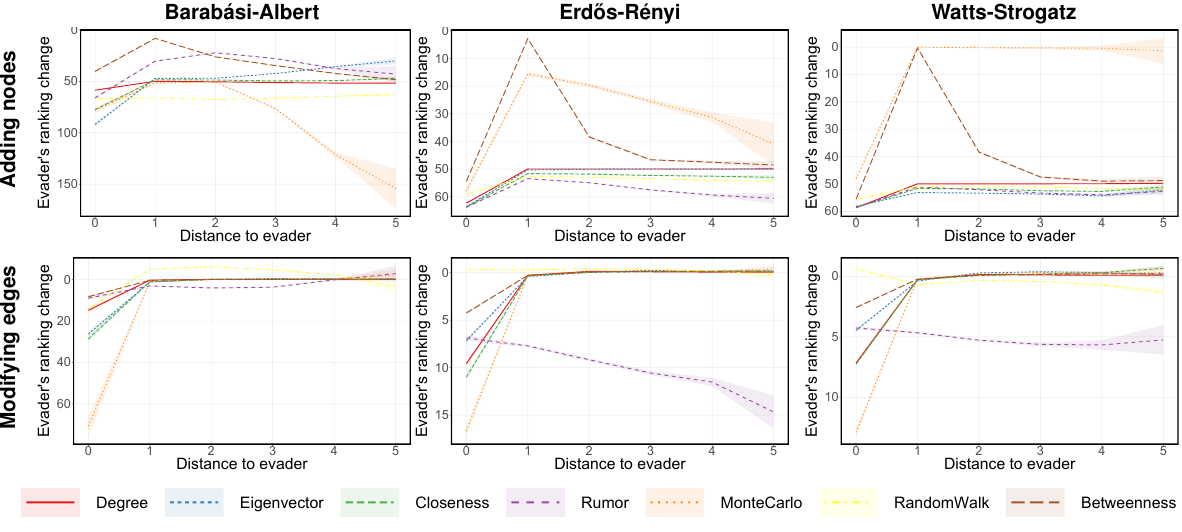}
\caption{
Evaluating how the evader's ranking is affected when nearby nodes attempt to hide, given networks consisting of $1,000$ nodes with an average degree of $4$---the same average degree used in our basic experiments. The x-axis corresponds to the distance between the node running the heuristic and the evader in an unmodified network, i.e., before hiding (notice that distance $0$ indicates that the evader is the one running the heuristic). The y-axis represents the change in the evader's ranking according to the source detection algorithm (greater value indicates more effective hiding). Shaded areas represent $95\%$ confidence intervals.
}
\label{fig:distance-line}
\end{figure}

\clearpage
\section{The Seeker with Imperfect Knowledge}
\label{app:fuzzy}

We now investigate how the completeness of the seeker's knowledge about the network's structure affects the effectiveness of the evader's hiding. To this end, we run the simulations where the seeker knows only about a certain percentage of the network's edges, which are selected uniformly at random at the beginning of each experiment. Then, we evaluate the two heuristics that proved to be the most effective in our basic experiments, i.e., the one that adds nodes and the one that modifies edges.

The results of our simulations are presented in Figure~\ref{fig:fuzzy-line}.
As shown in the figure, in the vast majority of cases the evader's hiding becomes significantly more effective as the knowledge of the seeker dwindles.
One notable exception is the Monte Carlo source detection algorithm, where the accuracy of the seeker's prediction can in some cases increase when its knowledge about the network decreases.
This might be caused by the fact that when certain edges are hidden from the seeker, the network appears (in the eyes of the seeker) to be fragmented into several connected components. In this case, the Monte Carlo algorithm runs the diffusion in the connected components that do not include the evader, thereby yielding an outcome that is significantly different from the original diffusion course.

\begin{figure}[tbh]
\centering
\includegraphics[width=\linewidth]{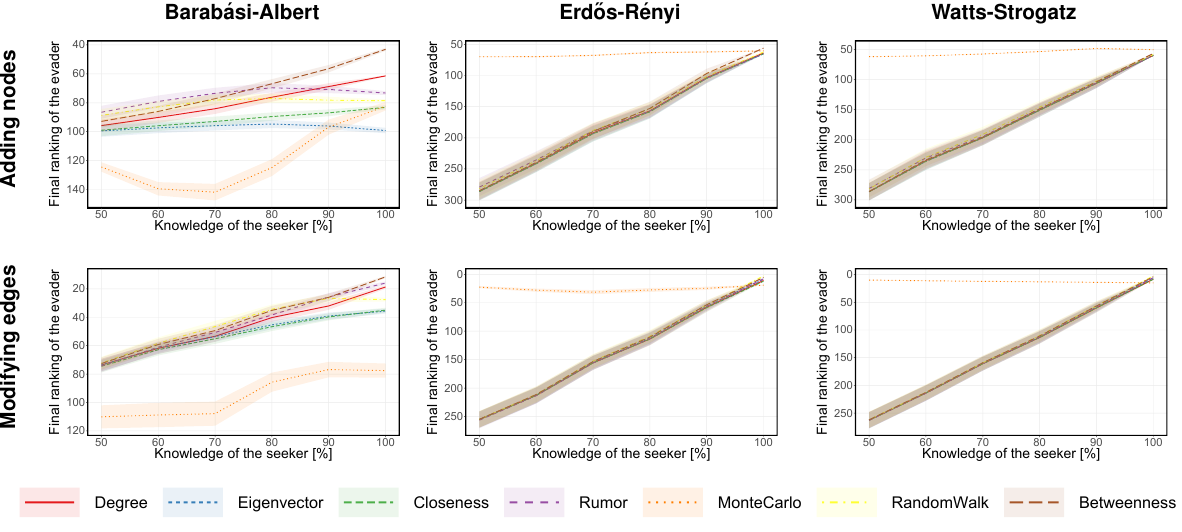}
\caption{
Effects of imperfect knowledge of the seeker given networks consisting of $1,000$ nodes with an average degree of $4$---the same average degree used in our basic experiments. The x-axis corresponds to the percentage of the network's edges that are visible to the seeker. The y-axis represents the final ranking the evader according to different source detection algorithms after the hiding process (greater value indicates more effective hiding). Shaded areas represent $95\%$ confidence intervals.
}
\label{fig:fuzzy-line}
\end{figure}

\clearpage
\section{Randomly Selected Evader}
\label{app:any-evader}

In this section, instead of selecting the evader randomly from the $10\%$ of nodes with the highest degrees (as we did in the main article), we select the evader uniformly at random from all nodes in the network. All other parameters of the simulations remain exactly the same as in the basic experiments presented elsewhere in our study.

Figures~\ref{fig:any-before-bars} to~\ref{fig:any-edge-simulation-heat} present the results of the simulations. Let us comment on these results, starting with Figure~\ref{fig:any-before-bars}. This figure depicts the position of the evader in the rankings produced by the different source detection algorithms, before the evader initiates the hiding process.
Compared to the results presented earlier in Figure~\ref{fig:standard-before-bars} of Appendix~\ref{app:hiding-profiles} (where the evader is selected from the $10\%$ of nodes with the highest degrees), we can see that the extent to which the source is exposed without any hiding remains broadly the same.
While this observation may seem surprising at first, the core principle underlying the majority of the source detection algorithms in our simulations is to consider only the part of the network that is induced by the set of infected nodes.
Thus, even if the source is not particularly important in the entire network, it is by necessity relatively central in the subnetwork induced by the infected nodes, simply because these nodes became infected as a result of a diffusion process that originated at the source.

Figures~\ref{fig:any-subo-simulation-line} to~\ref{fig:any-edge-simulation-heat} present the effectiveness of the hiding process, given the new way in which the evader is selected (randomly from the entire network).
As can be seen, the hiding heuristics are slightly less effective compared to the earlier case where the evader was selected from the $10\%$ with the highest degree. Nevertheless, these heuristics can still manage to significantly decrease the chance of identifying the evader by the source detection algorithms.
Finally, in terms of the relative differences between the heuristics themselves, the ones that were most effective before (when the evader is selected from the $10\%$ with highest degree) are still the most effective (when it is selected randomly from the entire network).

\begin{figure}[tbh]
\centering
\includegraphics[width=\linewidth]{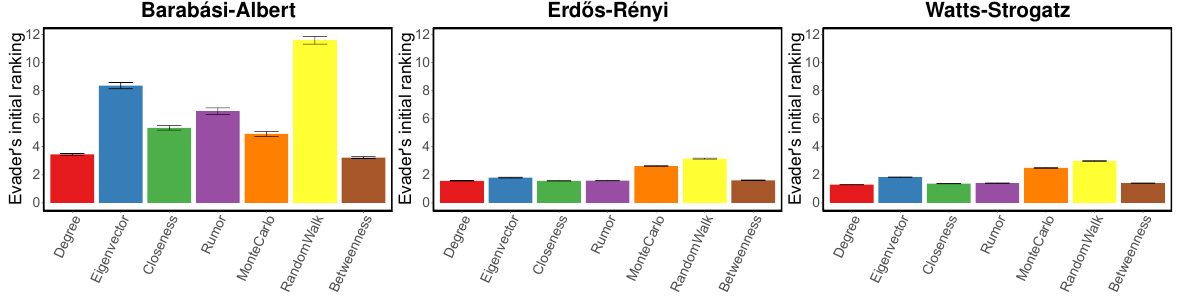}
\caption{
Comparison of the effectiveness of different source detection algorithms before the hiding process in random networks consisting of 1,000 nodes, when the evader is selected uniformly at random.
The x-axis corresponds to different source detection algorithms, while the y-axis corresponds to the evader's ranking according to the different algorithms.
The error bars represent $95\%$ confidence intervals.
}
\label{fig:any-before-bars}
\end{figure}

\begin{figure}[tbh]
\centering
\includegraphics[width=.8\linewidth]{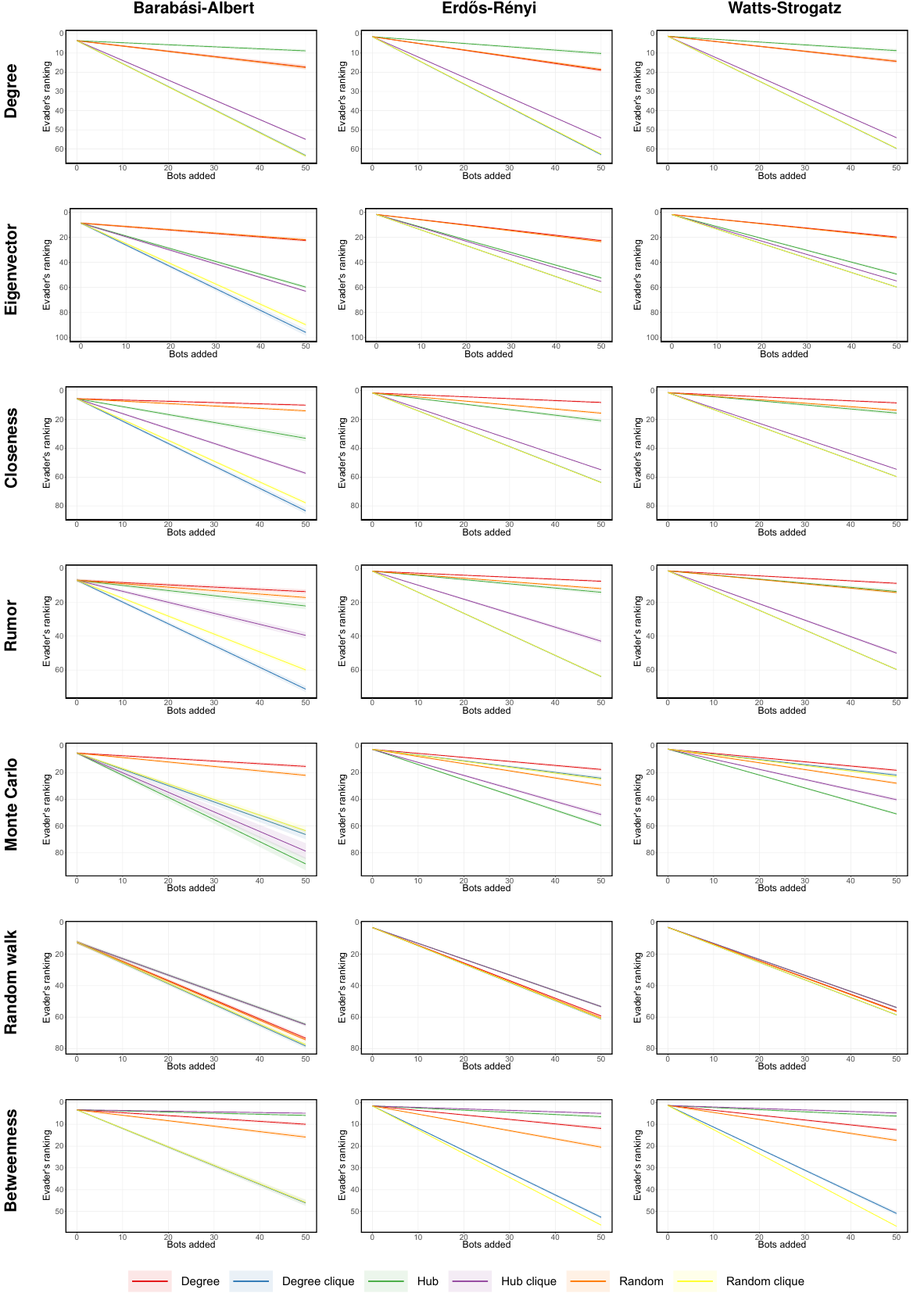}
\caption{
Results of hiding the source of diffusion by adding confederates for alternative models of diffusion and network generation, when the evader is selected uniformly at random. The y-axis corresponds to the ranking of the evader according to the source detection algorithm (greater values indicate more efficient hiding), while the x-axis corresponds to the number of nodes added to the network. Each color corresponds to a different heuristic, with each confederate being connected to three supporters. Shaded areas represent $95\%$ confidence intervals.
}
\label{fig:any-subo-simulation-line}
\end{figure}

\begin{figure}[tbh]
\centering
\includegraphics[width=.78\linewidth]{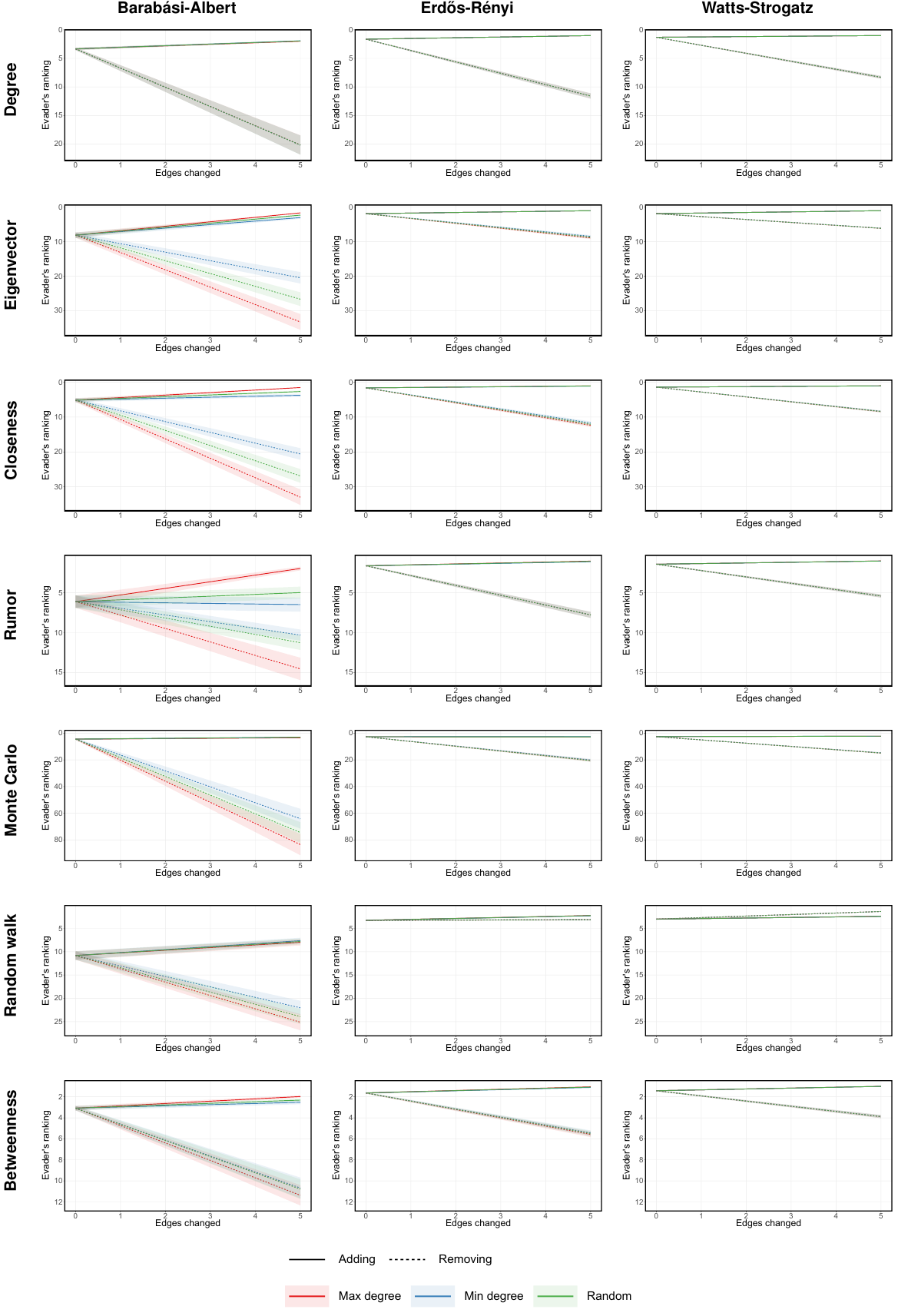}
\caption{
Results of hiding the source of diffusion by modifying edges for alternative models of diffusion and network generation, when the evader is selected uniformly at random. The y-axis represents the evader's ranking according to the source detection algorithm (greater value indicates more effective hiding); the x-axis corresponds to the number of edges added to, or removed from, the network. Each color corresponds to a different way of choosing edges, while each line type corresponds to either adding or removing edges. Shaded areas represent $95\%$ confidence intervals.
}
\label{fig:any-edge-simulation-line}
\end{figure}

\begin{figure}[tbh]
\centering
\includegraphics[width=\linewidth]{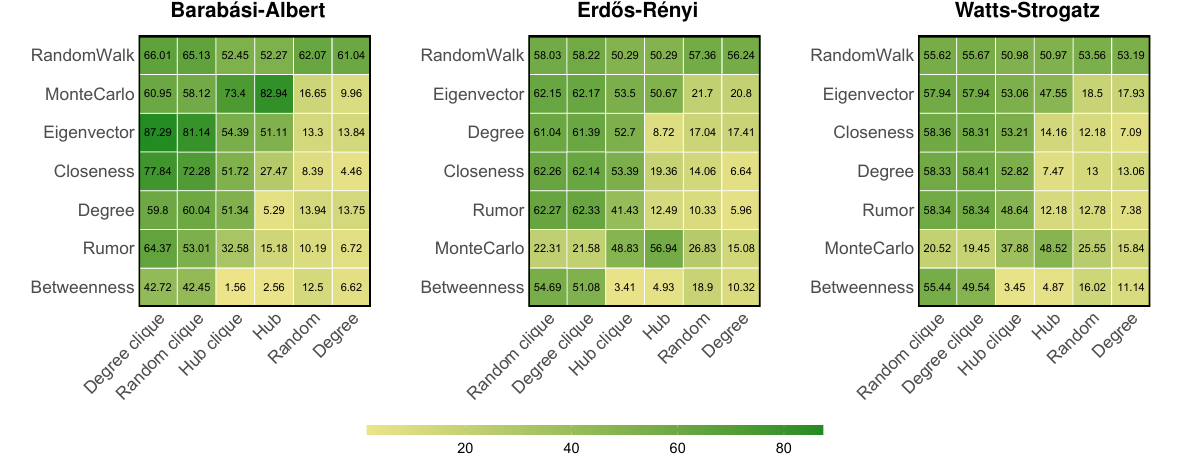}
\caption{
Results of simulations with hiding the source of diffusion by adding nodes for alternative models of diffusion and network generation, when the evader is selected uniformly at random. The y-axis of each heatmap corresponds to different source detection algorithms, whereas the x-axis corresponds to different heuristics. The value in each cell indicates the change in the evader's ranking according to the source detection algorithm after adding $50$ confederates to the network using the heuristic. Rows and columns are sorted by average value.
}
\label{fig:any-subo-simulation-heat}
\end{figure}

\begin{figure}[tbh]
\centering
\includegraphics[width=\linewidth]{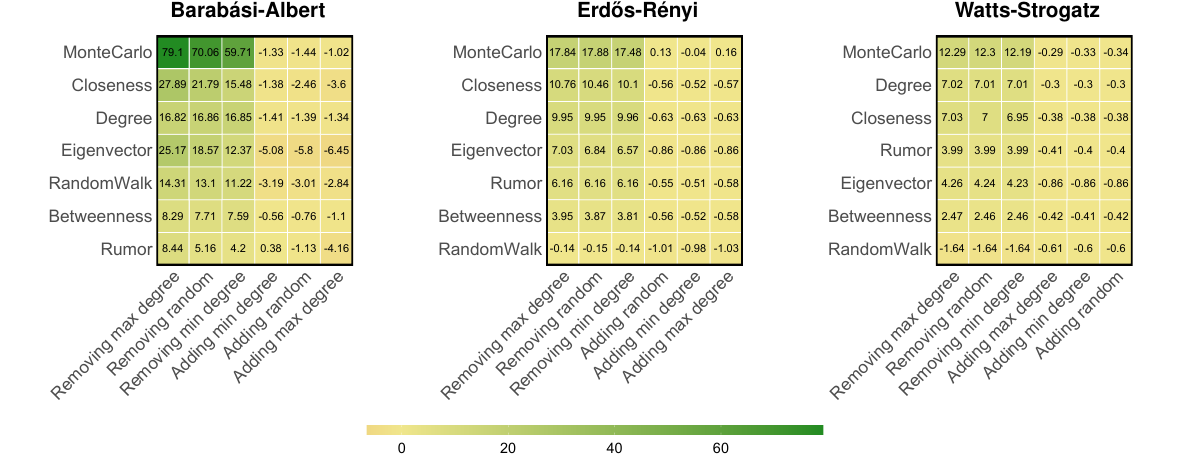}
\caption{
Results of hiding the source of diffusion by modifying edges for alternative models of diffusion and network generation, when the evader is selected uniformly at random. In each heatmap, rows correspond to different source detection algorithms, while columns correspond to different heuristics. The value in each cell indicates the change in the evader's ranking according to the source detection algorithm as a result of adding or removing $5$ edges to the network, depending on the heuristic. Positive values indicate that the evader became more hidden, with greater values indicated a more effective disguise. In contrast, negative values indicate that the evader became less hidden. Rows and columns are sorted by average value.
}
\label{fig:any-edge-simulation-heat}
\end{figure}

\clearpage
\section{Sensitivity Analysis}
\label{app:sensitivity}

This section provides a sensitivity analysis by modifying the parameters used in our original simulations.

\subsection{Selection of the Evader}
\label{app:sensitivity-selection}

We start by analyzing how the way of selecting the evader from among the set of nodes affects the effectiveness of the hiding process.
In the main article we select the evaders from among the top $10\%$ of nodes according to their degrees.
We now vary this parameter between $10\%$ and $90\%$.

Figures~\ref{fig:sensitivity-percentile-subos} and~\ref{fig:sensitivity-percentile-edge} present the results of our simulations. As can be seen, in the vast majority of cases the choice of this parameter has a relatively small effect on the change in the evader's ranking after the hiding process.

\begin{figure}[tbh]
\centering
\includegraphics[width=.8\linewidth]{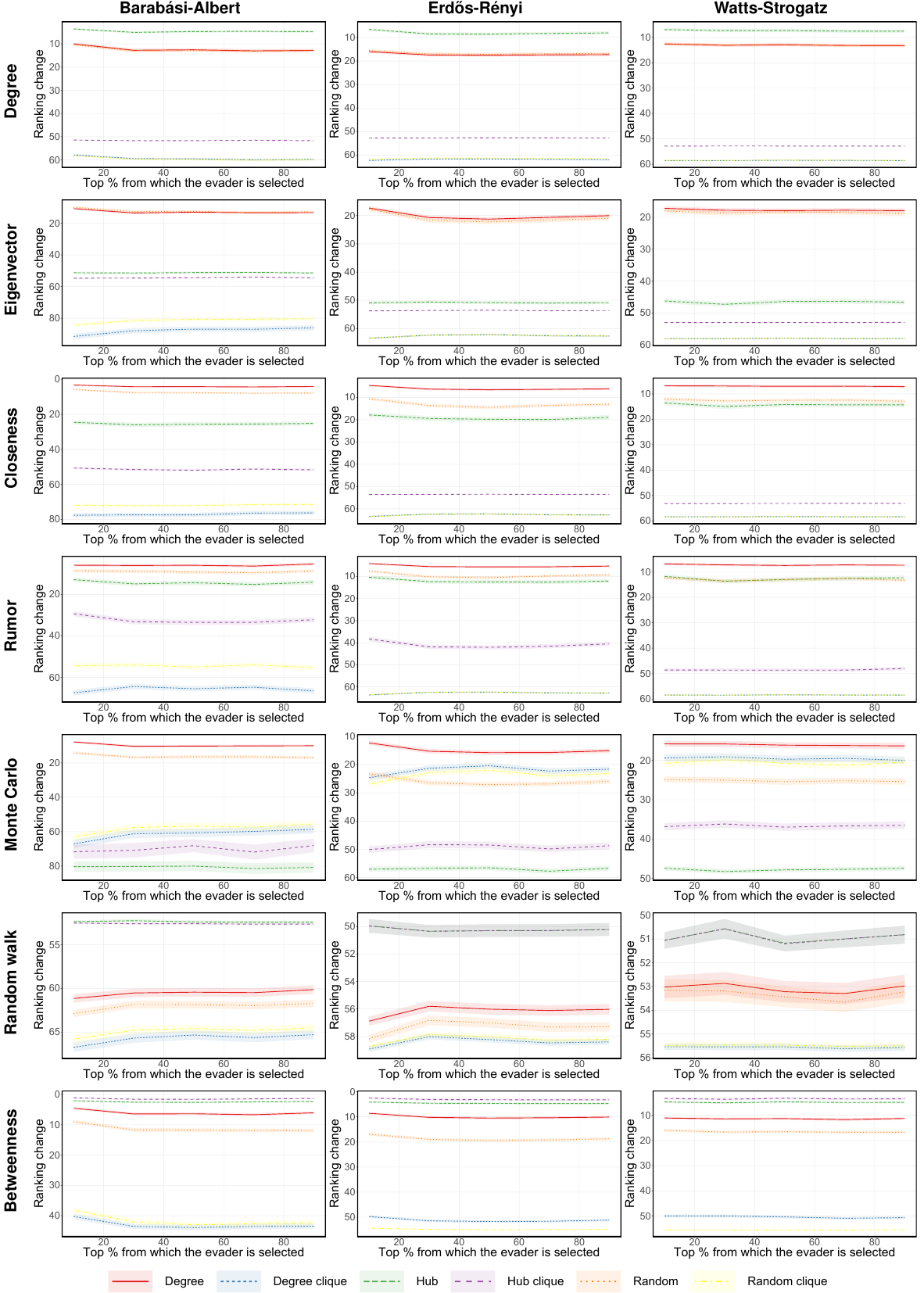}
\caption{
Impact of the way of selecting the evader on the effectiveness of the evader's hiding for heuristics that add nodes. The x-axis corresponds to the percentage of top-ranked nodes (according to degree) from which the evader is randomly selected (e.g., $x=30\%$ corresponds to the case when the evader is selected randomly from the $30\%$ with the highest degrees). The y-axis represents the change in the evader's ranking after the hiding process (greater value indicates more effective hiding). Shaded areas represent $95\%$ confidence intervals.
}
\label{fig:sensitivity-percentile-subos}
\end{figure}

\begin{figure}[tbh]
\centering
\includegraphics[width=.85\linewidth]{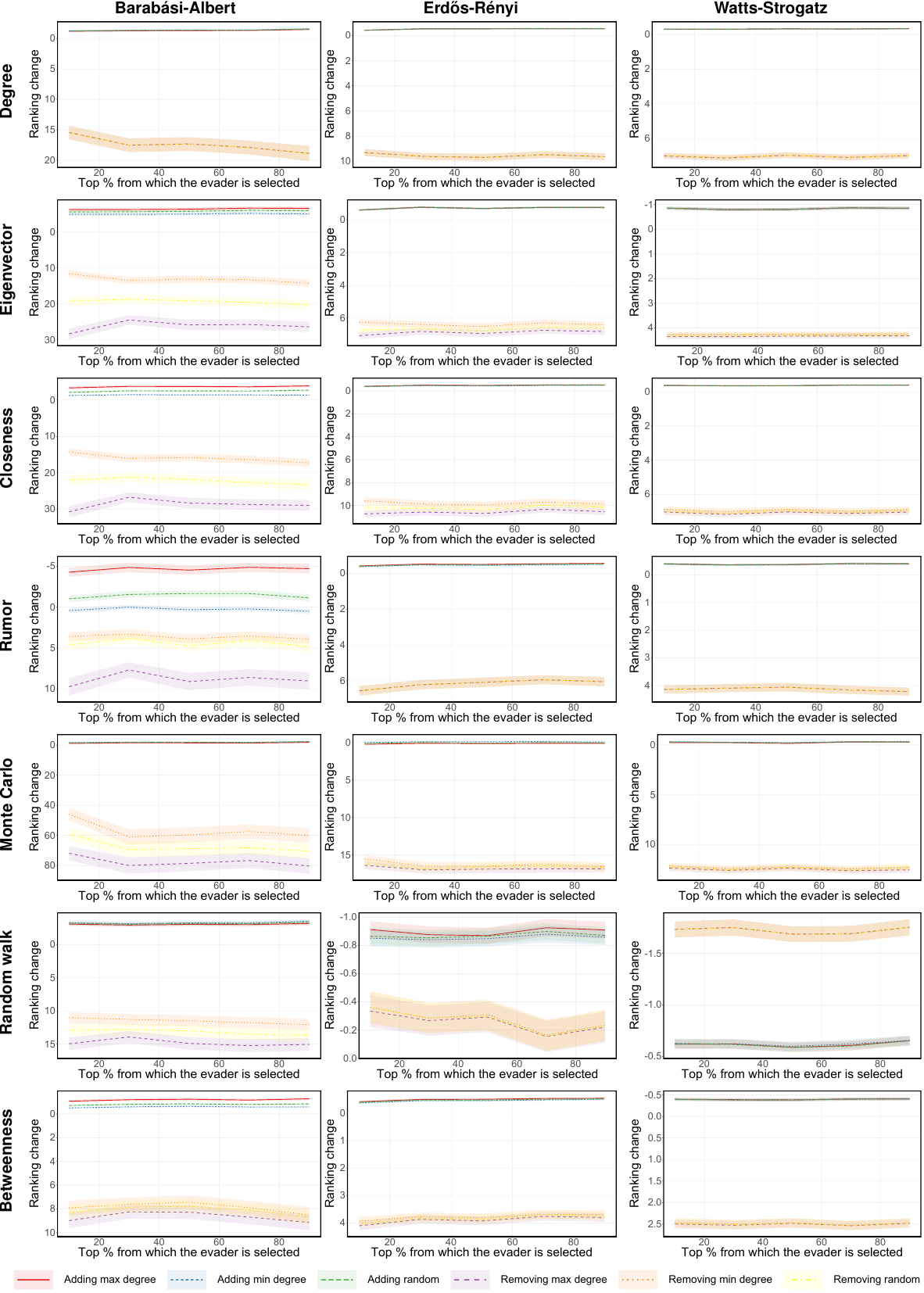}
\caption{
The same as Figure~\ref{fig:sensitivity-percentile-subos} but for heuristics that modify edges rather than for heuristics that add nodes.
}
\label{fig:sensitivity-percentile-edge}
\end{figure}

\subsection{Watts-Strogatz Model Rewiring Parameter}

In this subsection, we study how the rewiring parameter of the Watts-Strogatz network generation model affects the performance of our heuristics. In more detail, this parameter determines the probability that a given edge in an initially formed regular ring lattice gets rewired.
In the experiments presented in the main article we set this probability to $0.25$. Next, we run experiments with the probability ranging from $0.05$ to $0.5$. 

Figure~\ref{fig:sensitivity-wsrewire} presents the results of our simulations.
As can be seen, the hiding process is usually slightly more effective in small world networks generated using greater values of the parameter.
However, in the majority of cases, the difference in effectiveness between the parameter values $0.05$ and $0.5$ is just a few ranking positions.

\begin{figure}[tbh]
\centering
\includegraphics[width=.82\linewidth]{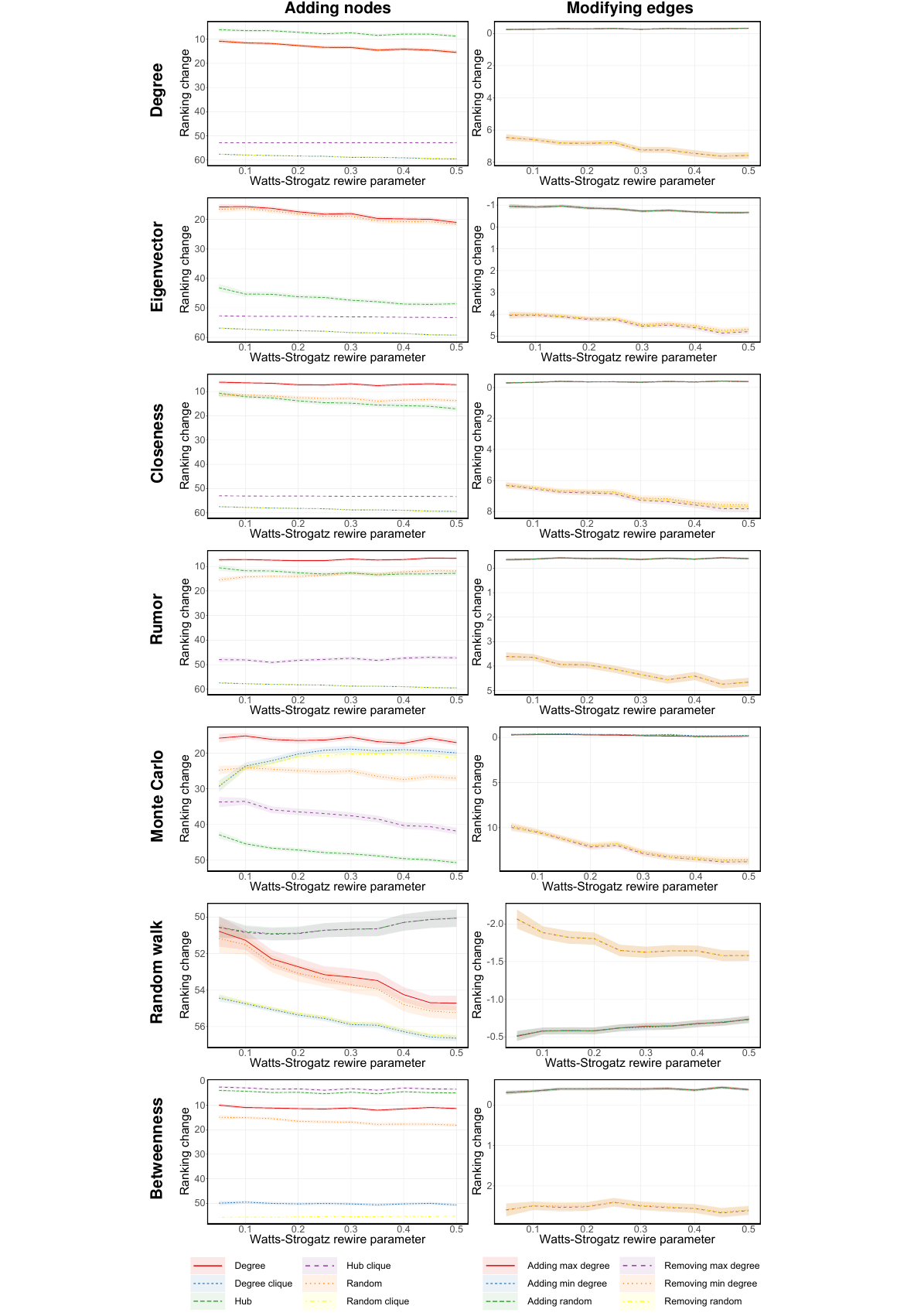}
\caption{
The impact of the rewiring parameter of the Watts-Strogatz model on the effectiveness of the evader's hiding. The x-axis corresponds to different values of the parameter. The y-axis represents the change in the evader's ranking after the hiding process (greater value indicates more effective hiding). Shaded areas represent $95\%$ confidence intervals.
}
\label{fig:sensitivity-wsrewire}
\end{figure}

\subsection{Comparison Group Size in Large Networks Simulations}

In the main article and in Section~\ref{app:simulation-large} of this document, we experimented with networks consisting of $100,000$ nodes. Ideally, using different source detection algorithms, one should compute scores for all $100,000$ nodes, and then rank them based on these score in order to determine the evader's position among all nodes. However, this would significantly increase the time required to run the simulations. Based on this, we resorted to approximating the evader's ranking by computing it among the $5,000$ infected nodes with the greatest degrees, and among another $5,000$ infected nodes chosen randomly out of the remaining nodes.
In this section, we analyze the impact of the size of the comparison group on the experiments' outcome. Specifically, instead of the aforementioned $5,000$, we now consider values ranging from $1,000$ to $9,000$.

Figures~\ref{fig:sensitivity-fast-subos} and~\ref{fig:sensitivity-fast-edge} present the results of our analysis.
As can be seen, modifying the size of the comparison group does not introduce systemic bias to the results obtained from the simulations.

\begin{figure}[tbh]
\centering
\includegraphics[width=\linewidth]{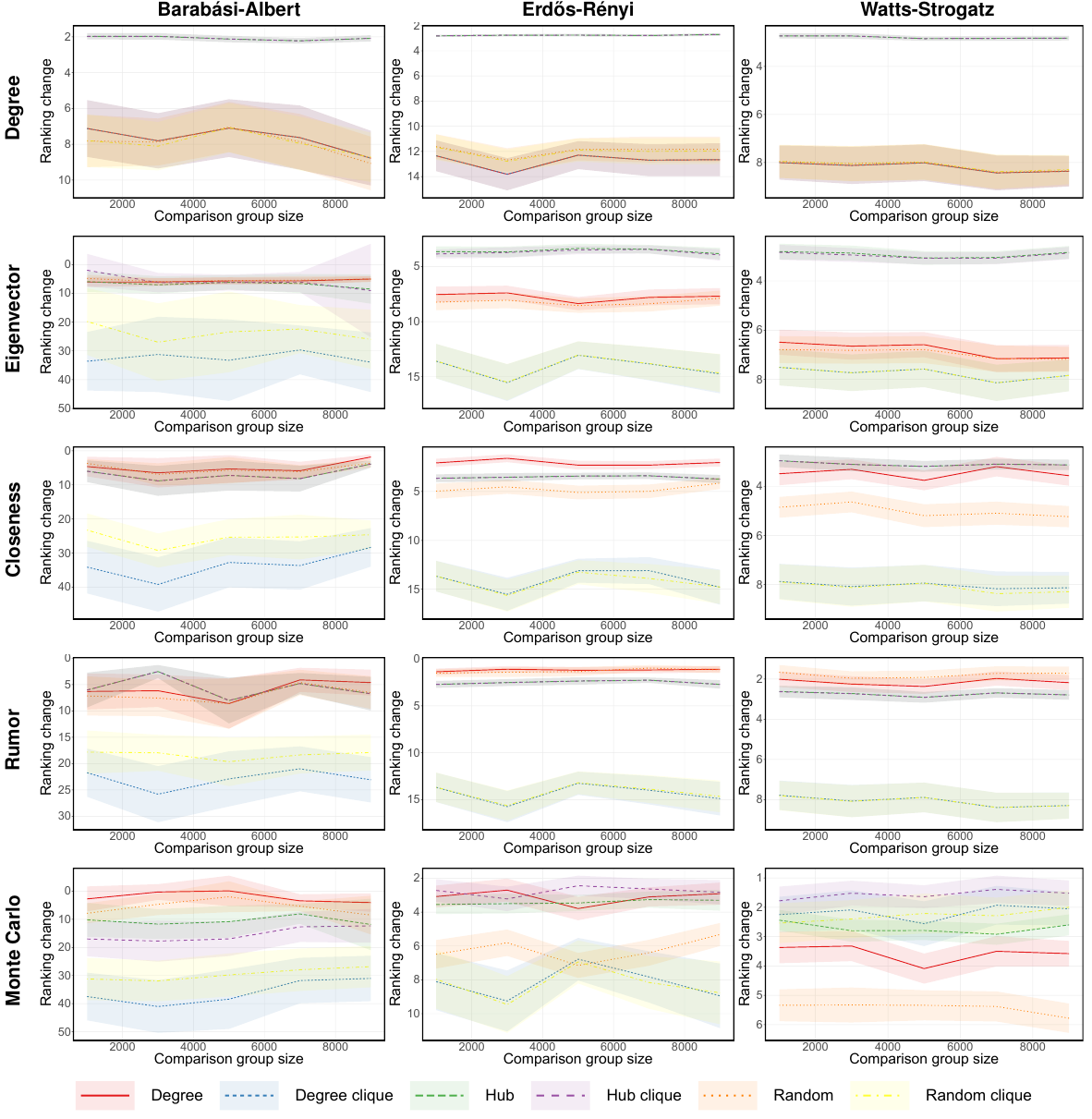}
\caption{
The impact of the size of the comparison group on the experiments with large networks for heuristics that add nodes. The x-axis corresponds to the size of the comparison group. The y-axis represents the change in the evader's ranking according to different source detection algorithms after the hiding process (greater value indicates more effective hiding). Shaded areas represent $95\%$ confidence intervals.
}
\label{fig:sensitivity-fast-subos}
\end{figure}

\begin{figure}[tbh]
\centering
\includegraphics[width=\linewidth]{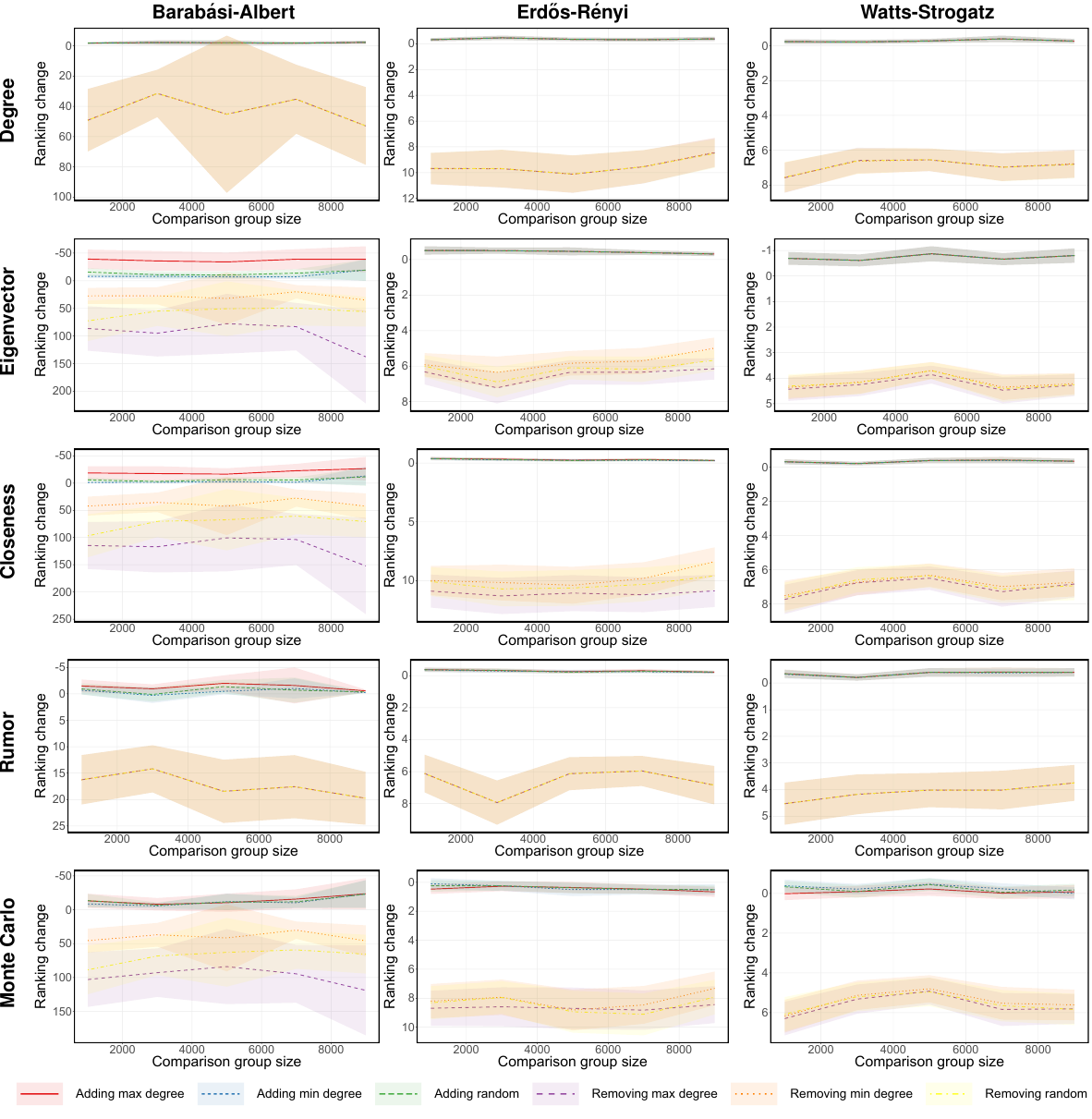}
\caption{
The same as Figure~\ref{fig:sensitivity-fast-subos} but for heuristics that modify edges rather than for heuristics that add nodes.
}
\label{fig:sensitivity-fast-edge}
\end{figure}

\end{document}